\documentclass[11pt]{article}
\usepackage{amsmath}
\usepackage{graphicx}
\usepackage{enumerate}
\usepackage{natbib}
\usepackage{url} 

\usepackage{setspace}
\usepackage[normalem]{ulem}

\usepackage{color}

\usepackage{xcolor}

\usepackage{floatrow}
\usepackage{graphicx}
\graphicspath{{images/}}
\usepackage{caption}
\usepackage{subcaption}

\usepackage{dirtytalk}

\usepackage{newtxtext}
\usepackage[subscriptcorrection]{newtxmath}

\usepackage{amsthm}

\newtheorem{theorem}{Theorem}
\newtheorem{proposition}{Proposition}

\newtheorem{lemma}{Lemma}
\newtheorem{remark}{Remark}
\newtheorem{assumption}{Assumption}

\usepackage[plain,noend]{algorithm2e}

\makeatletter
\renewcommand{\algocf@captiontext}[2]{#1\algocf@typo. \AlCapFnt{}#2} 
\def\@algocf@capt@plain{top}
\renewcommand{\algocf@makecaption}[2]{%
  \addtolength{\hsize}{\algomargin}%
  \sbox\@tempboxa{\algocf@captiontext{#1}{#2}}%
  \ifdim\wd\@tempboxa >\hsize
    \hskip .5\algomargin%
    \parbox[t]{\hsize}{\algocf@captiontext{#1}{#2}}
  \else%
    \global\@minipagefalse%
    \hbox to\hsize{\box\@tempboxa}
  \fi%
  \addtolength{\hsize}{-\algomargin}%
}
\makeatother

\usepackage{dirtytalk}
\usepackage{comment}

\usepackage{hyperref}
\usepackage{xr-hyper}
\usepackage{multirow}

\newcommand{\blind}{1}

\begin{document}

\def\spacingset#1{\renewcommand{\baselinestretch}%
{#1}\small\normalsize} \spacingset{1}


\if1\blind
{
  \title{\bf Spectral decomposition-assisted multi-study
  \\ factor analysis}
  \author{Lorenzo Mauri
  \hspace{.2cm}\\
    Department of Statistical Science, Duke University\\
    \small{\texttt{lorenzo.mauri@duke.edu}}\\
    Niccolò Anceschi \\
    Department of Statistical Science, Duke University\\
    and\\
    David B. Dunson\\
        Department of Statistical Science, Duke University\\
    }
  \maketitle
} \fi

\if0\blind
{
  \bigskip
  \bigskip
  \bigskip
  \begin{center}
    {\LARGE\bf Spectral decomposition-assisted multi-study \\ factor analysis}
\end{center}
  \medskip
} \fi

\bigskip


\begin{abstract}
This article focuses on covariance estimation for multi-study data. Popular approaches employ factor-analytic terms with shared and study-specific loadings that decompose the variance into (i) a shared low-rank component, (ii) study-specific low-rank components, and (iii) a diagonal term capturing idiosyncratic variability. Our proposed methodology estimates the latent factors via spectral decompositions, with a novel approach for separating shared and specific factors, and infers the factor loadings and residual variances via surrogate Bayesian regressions. The resulting posterior has a simple product form across outcomes, bypassing the need for Markov chain Monte Carlo sampling and facilitating parallelization. The proposed methodology has major advantages over current Bayesian competitors in terms of computational speed, scalability and stability while also having strong frequentist guarantees. The theory and methods also add to the rich literature on frequentist methods for factor models with shared and group-specific components of variation. The approximation error decreases as the sample size and the data dimension diverge, formalizing a blessing of dimensionality. We show favorable asymptotic properties, including central limit theorems for point estimators and posterior contraction, and excellent empirical performance in simulations. The methods are applied to integrate three studies on gene associations among immune cells.
\end{abstract}

\noindent%
\vfill

\newpage

\section{Introduction}

Due in part to the importance of reproducibility and generalizability, it is routine to collect the same type of data in multiple studies.
In omics there tends to be substantial study-to-study variation leading to difficulty replicating the findings \citep{aach00, irizarry03}. To make reasonable conclusions from such data, one should analyze the data together using statistical methods that allow inferences on common versus study-specific components of variation. There is a rich recent literature on appropriate methods \citep{bai15, Coroneo2016UnspannedMacroeconomic, Ando2016PanelData, Ando2017ClusteringHuge, franks_hoff_2019, msfa, bmsfa, pfa, bfr, bmsnnmf, combinatorial_msfa, vi_bmsfa, sufa, adaptive_fa}.

A popular approach reduces dimensionality introducing factor analytic terms and models the $i$-th observation from the $s$-th study as
\begin{equation} \label{eq:fact_model}
     \begin{aligned}
      \mathbf{y}_{si}  &=  \mathbf{\Lambda} \mathbf{\eta}_{si} + \mathbf{\Gamma}_s \mathbf{\phi}_{si} + \mathbf{e}_{si}, 
      \end{aligned}  \quad (i=1, \dots, n_s; s=1, \dots, S),
\end{equation}
where 
$s$ indexes the studies, $n_s$ is sample size of study $s$, 
$\mathbf{\eta}_{si}$ and $\mathbf{\phi}_{si}$ are shared and study-specific factors, respectively, and 
${\mathbf{\Lambda}} \in \mathbb R^{p \times k_0 }$ 
and ${\mathbf{\Gamma}}_s \in \mathbb R^{p \times q_s}$ are loadings on these two sets of factors. 
{Inference proceeds after imposing distributional or moment assumptions on the latent factors $\mathbf{\eta}_{si}$ and $\mathbf{\phi}_{si}$ and the residual components $\mathbf{e}_{si}$. For instance, if we assume 
\begin{equation}\label{eq:distributions}\mathbf{\eta}_{si}  \sim N_{k_0}(0, \mathbf{I}_{k_0}), \quad \mathbf{\phi}_{si}\sim N_{q_s}(0, \mathbf{I}_{q_s}), \quad \mathbf{e}_{si} \sim N_p(0, \mathbf{\Sigma}_s) \quad (i=1, \dots, n_s; s=1, \dots, S),\end{equation}} where $ \mathbf{\Sigma}_s = \text{diag}(\sigma_{s1}^2, \dots, \sigma_{sp}^2)$,
integrating out the latent factors, we obtain an equivalent representation: 
\begin{equation}
\mathbf{y}_{si}  \sim N_p\left(0, {\mathbf{\Lambda}} {\mathbf{\Lambda}}^\top + {\mathbf{\Gamma}}_s{\mathbf{\Gamma}}_s^\top + \mathbf{\Sigma}_s \right),\quad (i=1, \dots, n_s; s=1, \dots, S).
    \label{eq:fact_model_int}
\end{equation}
Model \eqref{eq:fact_model_int} decomposes the covariance of observations in each study as a sum of three components: (i) a low-rank  component shared across studies (${\mathbf{\Lambda}} {\mathbf{\Lambda}}^\top $), (ii) a study-specific low-rank component (${\mathbf{\Gamma}}_s{\mathbf{\Gamma}}_s^\top$), and (iii) a diagonal term capturing idiosyncratic variability of each variable ($\mathbf{\Sigma}_s$).

\citet{msfa} developed an expectation conditional maximization algorithm for maximum likelihood estimation of model
\eqref{eq:fact_model}-\eqref{eq:distributions}, while \citet{bmsfa} adopt a Bayesian approach to perform posterior computations using a Gibbs sampler. \citet{bfr} extend the model, allowing some of the variability to be explained by observed covariates, while \citet{combinatorial_msfa} and \citet{adaptive_fa} let some of the loadings be shared only by subsets of studies. {\citet{bai15} focus on dynamic factor models, treating factors as fixed and unknown, and studying the conditions needed to identify factor loading matrices. They rely on Gibbs sampling for posterior computation.}
Both expectation maximization and Gibbs sampling tend to suffer from slow convergence or mixing when the number of variables $p$ is large, which motivated \citet{vi_bmsfa} to develop variational approximations. Such variational inference algorithms lack theoretical guarantees, can massively underestimate uncertainty, and, as we show in the numerical experiments section, the estimation accuracy is often unsatisfactory. 

Subtle identifiability issues arise with the model \eqref{eq:fact_model}. As noted in \citet{sufa}, model \eqref{eq:fact_model_int} is not identifiable without further restriction, since $\mathbf{y}_{si}$ would have the same marginal distribution if ${\mathbf{\Lambda}}$ and ${\mathbf{\Gamma}}_s$ were replaced with a matrix of $0$'s and ${\mathbf{\bar \Gamma}}_s = [{\mathbf{\Lambda}} ~ {\mathbf{\Gamma}}_s]$ respectively. \citet{msfa} ensure identifiability, up to orthogonal rotations of the loadings, by requiring the matrix obtained combining ${\mathbf{\Lambda}}$ and the ${\mathbf{\Gamma}}_s$'s to be full rank; see Assumption \ref{assumption:li} in Section \ref{sec:theory} of this paper. However, they do not impose this constraint in their estimation procedure, leading to poor empirical performance in high dimensions.   \citet{sufa} achieve identifiability through a shared subspace restriction that lets ${\mathbf{\Gamma}}_s = {\mathbf{\Lambda}} A_s$, with $A_s \in \mathbb R^{k_0 \times q_s}$. This choice may be too restrictive in cases with substantial variation between studies in that it requires relatively few study-specific factors.
{There are many other Bayesian approaches to (single-study) factor analysis \citep[e.g.][]{sparse_infinite, cusp, rotate}, but such identifiability issues complicate their adaptation to model \eqref{eq:fact_model}}. \citet{pfa} proposes an alternative to \eqref{eq:fact_model}, which incorporates a study-specific multiplicative perturbation from a shared factor model.

There is a parallel line of research developing spectral estimation techniques for multi-group data where some of the principal axes are shared across groups. \citet{boik02} propose a joint model for the eigenstructure in multiple covariance matrices, with Fisher scoring used to compute maximum likelihood estimates. \citet{hoff_09} develops a related approach using a hierarchical model to allow eigenvectors to be similar but not equal between groups. These approaches have the disadvantage of modeling covariances as exactly low rank without a residual noise term. \citet{franks_hoff_2019} address this problem via a spiked covariance model that incorporates a shared subspace, using an expectation maximization algorithm to infer the shared subspace and a Gibbs sampler for group-specific covariance matrices.  Alternatively, \citet{hu_et_al} jointly estimate multiple covariance matrices while shrinking estimates towards a pooled sample covariance. {\citet{Ando2016PanelData, Ando2017ClusteringHuge} develop a spectral decomposition to a generalization of model \eqref{eq:fact_model}, which also includes observed covariates and unknown grouping structure. However, such methods do not provide uncertainty quantification. Moreover, as we show in our numeric experiments, their estimates are not competitive unless the number of outcomes is extremely large ($p\gg 200$) and the loadings are close to orthogonal.}

\citet{fable} highlighted the poor computational performance of Gibbs samplers in single-study factor models. They estimate factors via a singular value decomposition, and infer loadings and residual error variances via conjugate normal-inverse gamma priors, showing concentration of the induced posterior on the covariance at the true values and validity of the coverage of entrywise credible intervals. This approach is related to joint maximum likelihood or \textit{maximum a posteriori} estimates, which have been shown to be consistent for generalized linear latent variable models \citep{gllvm}, when both the sample size and the data dimensionality $p$ diverge \citep{chen_jmle, chen_identfiability, lee24, flair, fama, basil}. {Although we build on some aspects of their approach, modifications to multi-study contexts require substantial methodological and theoretical novelty.}

Motivated by these considerations, we propose a Bayesian Latent Analysis through Spectral Training (\texttt{BLAST}) methodology for inference under model \eqref{eq:fact_model}. \texttt{BLAST} starts with a novel approach to inferring shared and study-specific factors from spectral decompositions. This is achieved by estimating the directions of variation in each study and inferring the subset spanned by a common ${\mathbf{\Lambda}}$. This, in turn, disentangles the variation generated by the $\mathbf{\eta}_{si}$'s and the $\mathbf{\phi}_{si}$'s respectively, allowing their estimation.   Conditionally on estimated factors, multi-study factor analysis reduces to $p$ separate regression problems. We derive regularized ordinary least squares estimators for loadings matrices ${\mathbf{\Lambda}}$ and ${\mathbf{\Gamma}}_s$, which are interpretable as posterior means under conditionally Gaussian priors. We quantify uncertainty in parameter inference via posterior distributions under surrogate regression models. For independent priors on rows of ${\mathbf{\Lambda}}$ and ${\mathbf{\Gamma}}_s$, the conditional posterior factorizes and inferences can be implemented in parallel, greatly reducing computational burden.

From the posterior on the loadings and residual variances, we induce a posterior on the different components of the covariance in \eqref{eq:fact_model_int}. We provide strong support for the resulting point estimates and credible intervals through a high-dimensional asymptotic theory that allows the data dimensionality $p$ to grow with sample size. Our theory provides consistency and concentration rate results, a central limit theorem for estimators, and even a Bernstein-von Mises theorem. The latter result implies that credible intervals have a slight under-coverage asymptotically, which can be adjusted via variance inflation factors that can be calculated analytically.
To automate the methodology and favor greater data adaptivity, we develop empirical Bayes methodology for hyperparameter estimation.

Hence, \texttt{BLAST} provides a fast algorithm to obtain accurate point and interval estimates in multi-study factor analysis, with excellent theory and empirical support and without expensive and brittle Markov chain Monte Carlo sampling.

\section{Methodology}
\subsection{Notation}\label{subsec:notation}

For a matrix $\mathbf A$, we denote its spectral, Frobenius, and entrywise infinity norm by $||\mathbf A||$, $||\mathbf A||_F$, $||\mathbf A||_{\infty}$ respectively.
We denote by $s_l(\mathbf A)$, its $l$-th largest singular value. 
For a vector $v$, we denote its Euclidean and entrywise infinity norm by $||v||$, $||v||_\infty$, respectively.
For two sequences $(a_n)_{n \geq 1}$, $(b_n)_{n \geq 1}$, and we say $a_n \lesssim b_n$ if $a_n \leq C b_n$ for every $n > N$ for some finite constants $N < \infty$ and $C < \infty$. We say $a_n \asymp b_n$ if and only if $a_n \lesssim b_n$ and $b_n \lesssim a_n$.

\subsection{Latent factor estimation}\label{subsec:factor_pretraining}
This section describes our methodology for estimating latent factors. First, we rewrite model \eqref{eq:fact_model} in its equivalent matrix form:
\begin{equation*}
    \mathbf{Y}_s =  \mathbf{M_s} {\mathbf{\Lambda}}^\top + \mathbf{F}_s {\mathbf{\Gamma}}_s^\top + \mathbf{E}_s, \quad (s=1, \dots, S),
\end{equation*}
where 
$\mathbf{Y}_s = [\mathbf{y}_{s1}  \cdots \mathbf{y}_{s n_s}]^{\top} \in R^{n_s \times p},$
$\mathbf{E}_s = [
\mathbf{e}_{s1} \cdots     \mathbf{e}_{s n_s}]^{\top} \in \mathbb R^{n_s \times p},$
$\mathbf{M_s} = [
        \mathbf{\eta}_{s1} \cdots  \mathbf{\eta}_{s n_s}
]^\top \in \mathbb R^{n_s \times k_0}$, and
$\mathbf{F}_s = [
        \mathbf{\phi}_{s1} \cdots 
        \mathbf{\phi}_{s n_s}
]^\top \in \mathbb R^{n_s \times q_s}$. We denote by $k_s=k_0+q_s$ the latent dimension of each study, summing the shared and study-specific dimensions. Our initial goal is to estimate the latent factor matrices $\mathbf{M_s}$'s and $\mathbf{F}_s$'s. We start by computing the singular value decompositions of each $\mathbf{Y}_s$ 
and use them to identify common axes of variation. Specifically, we take $\mathbf{P}_s = \mathbf{V}_s \mathbf{V}_s^\top$, where $\mathbf{V}_s \in \mathbb R^{p \times k_s}$ denotes the matrix of right singular vectors of $\mathbf{Y}_s$ that correspond to the leading $k_s$ singular values. The proposition \ref{prop:V_s_outer} in the supplementary material shows that the orthogonal projection matrix $\mathbf{P}_s$ approximates the projection into the space spanned by $[{\mathbf{\Lambda}} ~ {\mathbf{\Gamma}}_s]$ in a spectral norm sense for large values of $n_s$ and $p$. Next, we obtain the singular value decomposition of the average of the $\mathbf{P}_s$,
\begin{equation}\label{eq:P_tilde}
    \mathbf{\tilde P} = \frac{1}{S}\sum_{s=1}^S \mathbf{P}_s,
\end{equation}
and set $\mathbf{\bar P} =   \mathbf{\bar V} \mathbf{\bar V}^\top$, where $\mathbf{\bar V}$ is the matrix of singular vectors associated to the leading $k_0$ singular values of $\mathbf{\tilde P}$. 
Recall that each $\mathbf{P}_s$ is a projection onto a space roughly spanned by $q_s$ directions specific to study $s$ and $k_0$ directions shared by all studies. The signal along shared axes is preserved by the averaging operation, while individual directions are dampened, particularly as $S$ increases.
Consequently, the spectrum of $\mathbf{\tilde P}$ consists of $k_0$ leading directions with singular values close to 1, corresponding to the shared directions of variation, well separated from the remaining $\sum_s q_s$ study-specific directions, with singular values $\ll 1$. Hence, $\mathbf{\bar P}$ and $\mathbf{\bar Q} = \mathbf{I}_p - \mathbf{\bar P}$ approximately project onto the space spanned by ${\mathbf{\Lambda}}$ and its orthogonal complement, respectively. 
\begin{remark}
    {In \eqref{eq:P_tilde}, we aggregate the study-specific projections via simple averaging. Alternatives, such as weighting each projection proportionally to the study sample size, are also possible.}
\end{remark}

Letting $\mathbf{V}$ be the matrix of left singular vectors of the true ${\mathbf{\Lambda}}$, Proposition \ref{prop:recovery_P_0} in the supplementary material bounds the spectral norm of the difference of $\mathbf{\bar P}$ and $ \mathbf{V}\mathbf{V}^\top$ in high probability by a multiple of $1/n + 1/p$, where $n = \sum_{s=1}^S n_s$ is the total sample size.

Having identified the shared axes of variation, we can now proceed to estimate the latent factors. By post-multiplying each $\mathbf{Y}_s$ by $\mathbf{\bar Q}$, we eliminate almost all the variation along the common axes of variation, with the remaining signal being predominantly study-specific:
$$\mathbf{\hat Y}_s^\perp = \mathbf{Y}_s \mathbf{\bar Q} \approx  (\mathbf{M_s} {\mathbf{\Lambda}}^\top + \mathbf{F}_s {\mathbf{\Gamma}}_s^\top + \mathbf{E}_s) (\mathbf{I}_p - \mathbf{V}\mathbf{V}^\top) =  \mathbf{F}_s {\mathbf{\Gamma}}_s^\top (\mathbf{I}_p - \mathbf{V}\mathbf{V}^\top)   + \mathbf{E}_s (\mathbf{I}_p - \mathbf{V}\mathbf{V}^\top).
$$
This allows us to estimate the latent factors corresponding to the study-specific variation as $\mathbf{\hat F_s} = \sqrt{n_s} \mathbf{U}_s^\perp$, where $ \mathbf{U}_s^\perp \in \mathbb R^{n_s \times q_s}$ is the matrix of left singular vectors associated to the leading $q_s$ singular vectors of $ \mathbf{\hat Y_s}^\perp$. 

Finally, to estimate the factors corresponding to the shared variation we regress out $\mathbf{\hat  F_s}$ from $\mathbf{Y}_s$, 
eliminating the variation explained by the factors corresponding to the study-specific variation, 
\begin{eqnarray}
   \mathbf{\hat Y}_s^c = (\mathbf{I}_{n_s} - \mathbf{U}_s^\perp \mathbf{U}_s^{\perp \top})  \mathbf{Y}_s \approx 
     (\mathbf{I}_{n_s} - \mathbf{U}_s^\perp \mathbf{U}_s^{\perp \top}) ( \mathbf{M_s} {\mathbf{\Lambda}}^\top+ \mathbf{F}_s {\mathbf{\Gamma}}_s^\top + \mathbf{E}_s) \approx \mathbf{M_s} {\mathbf{\Lambda}}^\top  + \mathbf{E}_s. 
     \label{eq:hatYsc}
\end{eqnarray}
Focusing on the concatenated elements $\mathbf{\hat Y}^c =  [\mathbf{\hat Y}_1^{c\top} \cdots\mathbf{\hat Y}_S^{c\top}]^\top$ and $\mathbf{M} = [ \mathbf{M}_1^\top \cdots  \mathbf{M}_S^\top]^\top$, we estimate the latent factors corresponding to the shared variation, $\mathbf{M}$, by $\mathbf{\hat M} = [\mathbf{\hat M}_1^\top \cdots\mathbf{\hat M}_S^\top]^\top = \sqrt{n} \mathbf{U}^c$, where $\mathbf{U}^c  \in \mathbb R^{n \times k_0} $ is the matrix of left singular vectors associated to the leading singular vectors of $\mathbf{\hat Y}^c$.
Theorem \ref{thm:factors_procrustes_error} shows that this procedure recovers the true latent factors (up to orthogonal transformations) in the high-dimensional and sample size limit.
The above procedure for factor estimation is summarized in Algorithm \ref{alg:factor_pretraining}.

\begin{algorithm}
\caption{Spectral estimation of shared and study-specific factors.}
\label{alg:factor_pretraining}\vspace{-2em}
\begin{tabbing}
   \qquad \enspace \text{Input:} The data matrices $\{\mathbf{Y}_s\}_{s=1}^S$, 
   and the studies' latent dimensions $\{k_s\}_{s=1}^S$.\\
   \qquad \enspace \text{Step 1:} For each $s=1, \dots, S$, compute the singular value decomposition of $\mathbf{Y}_s$\\ \qquad \enspace \hspace{2.8em} and take $\mathbf{P}_s = \mathbf{V}_s \mathbf{V}_s^\top$, where $\mathbf{V}_s \in \mathbb R^{p \times k_s}$ denotes the matrix of the right\\ \qquad \enspace \hspace{2.8em} singular vectors corresponding to the leading $k_s$ singular vectors. \\
   \qquad \enspace  \text{Step 2:} Compute the empirical average of the $\mathbf{P}_s$'s as $\mathbf{\tilde P} = \frac{1}{S} \sum_{s=1}^S \mathbf{P}_s$ and let \\ \qquad \enspace \hspace{2.8em} $\mathbf{\bar P} = \mathbf{\bar V} \mathbf{\bar V}^\top$ and $\mathbf{\bar Q} = \mathbf{I}_p - \mathbf{\bar P}$ where $\mathbf{\bar V} \in \mathbb R^{p \times k_0}$ is the matrix of singular vectors\\ \qquad \enspace \hspace{2.8em}  associated to the leading $k_0$ singular values of $\mathbf{\tilde P}$ and $k_0$ is chosen according \\ \qquad \enspace \hspace{2.8em} to the criterion in \eqref{eq:k_hat}.
\\
\qquad \enspace  \text{Step 3:} For each $s=1, \dots, S$, compute the singular value decomposition of \\ \qquad \enspace \hspace{2.8em} $\mathbf{\hat Y_s}^\perp =\mathbf{Y}_s \mathbf{\bar Q}$ and  let $\mathbf{\hat  F_s} = \sqrt{n_s} \mathbf{U}_s^\perp$, where $ \mathbf{U}_s^\perp \in \mathbb R^{n_s \times q_s}$ is the matrix of left \\ \qquad \enspace \hspace{2.8em} singular vectors associated to the leading $q_s$ singular vectors of $ \mathbf{\hat Y}_s^\perp$.\\

\qquad \enspace  \text{Step 4:} For each $s=1, \dots, S$, compute $\mathbf{\hat Y_s}^c = (\mathbf{I}_{n_s} - \mathbf{U}_s^\perp \mathbf{U}_s^{\perp \top})\mathbf{Y}_s$, define \\ \qquad \enspace \hspace{2.8em} $\mathbf{\hat Y}^c =  [\mathbf{\hat Y}_1^{c\top} \dots\mathbf{\hat Y}_S^{c\top}]^\top$ and let $ \mathbf{\hat M} = [\mathbf{ \hat M}_1^\top \dots \mathbf{\hat M}_S^\top]^\top =  \sqrt{n} \mathbf{U}^c$, where \\ \qquad \enspace \hspace{2.8em} $n = \sum_{s=1}^S n_s$ and $\mathbf{U}^c  \in \mathbb R^{n \times k_0} $ is the matrix of left singular vectors \\
\qquad \enspace  \hspace{2.8em} associated to the $k_0$ leading singular vectors of $\mathbf{\hat Y}^c$.\\

\qquad \enspace \text{Output:} The estimates for the latent factors $\{\mathbf{\hat  F_s}\}_{s=1}^S$ and $\mathbf{\hat M}$, and the matrix $\mathbf{\hat Y}^c$.
\end{tabbing}
\end{algorithm}

\subsection{Inference on shared and study-specific loadings}\label{subsec:inference_Lambda}
This section describes how factor loading matrices are estimated. Conditionally on the latent factors, we first sample the posterior of ${\mathbf{\Lambda}}$ and then propagate its uncertainty in estimating the ${\mathbf{\Gamma}}_s$'s.  
Recall that $\mathbf{\hat Y_s}^c$ is obtained by projecting out $\mathbf{\hat  F_s}$ from $\mathbf{Y}_s$ through \ref{eq:hatYsc}. We propose inferring ${\mathbf{\Lambda}}$ via the surrogate regression problem, 
\begin{equation*} 
   \mathbf{\hat Y}^c = [
       \mathbf{\hat Y}_1^{c\top}  ~ \cdots ~ \mathbf{\hat Y}_S^{c\top}
    ]^\top =\mathbf{\hat M}  {\mathbf{\tilde \Lambda}}^\top + \mathbf{\tilde E}, \quad \text{vec}\big(\mathbf{\tilde E}\big) \sim N_{np}(0, \mathbf{\tilde \Sigma} \otimes \mathbf{I}_n), \quad \mathbf{\tilde \Sigma} = \text{diag}\big(\tilde \sigma_1^2, \dots, \tilde \sigma_p^2\big),
\end{equation*}
where we introduced the new parameters ${\mathbf{\tilde \Lambda}} = [
    \mathbf{\tilde \lambda}_1 ~ \cdots ~\mathbf{\tilde \lambda}_p
]^\top$ and $\{\tilde \sigma_j^2\}_{j=1}^p$. 
We adopt conjugate normal-inverse gamma priors on the rows of ${\mathbf{\tilde \Lambda}}$ and residual error variances ($\{\mathbf{\tilde \lambda}_j, \tilde \sigma_j^2\}_{j}$),
\begin{equation*}
     \mathbf{\tilde \lambda}_{j} \mid \tilde \sigma_{j}^2 \sim N_{k_0} \big(0, \tau_{\Lambda}^2 \tilde \sigma_{j}^2 \mathbf{I}_{k_0}\big),  \quad \tilde \sigma_{j}^2 \sim IG \big( \frac{\nu_0}{2}, \frac{\nu_0 \sigma_0^2}{2}  \big),
\quad (j=1, \dots, p).
      \end{equation*} 
This implicitly assumes that residual error variances do not vary across studies, that is
\begin{equation*}
    \mathbf{\Sigma}_s = \mathbf{\Sigma}, \quad (s=1, \dots, S), \quad \text{where } \mathbf{\Sigma} = \text{diag}\big( \sigma_1^2, \dots,  \sigma_p^2\big).
\end{equation*}
In Section \ref{sec:theory} we discuss implications when this condition is not met,  and in the supplementary material we present alternative inference schemes that take into account the heteroscedastic case. 
This prior specification leads to conjugate posterior distributions, 
\begin{equation*}\begin{aligned}
    (\mathbf{\tilde \lambda}_j, \tilde \sigma_j^2) \mid\mathbf{\hat y}^{c (j)}& \sim  NIG \big(\mathbf{\tilde \lambda}_j, \tilde \sigma_j^2; \mathbf{\mu}_{\lambda_j}, \mathbf{K}, \mathbf{\gamma}_n/2, \mathbf{\gamma}_n \delta_j^2/2\big)\\
    &= N_{k_0}\big(\mathbf{\tilde \lambda}_j; \mathbf{\mu}_{\lambda_j}, \tilde \sigma_j^2 \mathbf{K}\big)IG \big(\tilde \sigma_j^2; \mathbf{\gamma}_n/2, \mathbf{\gamma}_n \delta_j^2/2\big),
\end{aligned}
\end{equation*}
where 
\begin{equation*}
    \begin{aligned}
         \mathbf{\mu}_{\lambda_j} &=  \big(\mathbf{\hat M}^\top\mathbf{\hat M} + \tau_{\Lambda}^{-2} \mathbf{I}_{k_0} \big)^{-1}\hat M^{\top}\mathbf{\hat y}^{c (j)},\\          \mathbf{K} & = \big(\mathbf{\hat M}^\top\mathbf{\hat M} + \tau_{\Lambda}^{-2} \mathbf{I}_{k_0} \big)^{-1} = \frac{1}{n + \tau_{\Lambda}^{-2}} \mathbf{I}_{k_0},\\
         \mathbf{\gamma}_n &= \nu_0 + n,\\
         \delta_j^2 &= \frac{1}{\mathbf{\gamma}_n} \big(\nu_0 \sigma_0^2  +\mathbf{\hat y}^{c (j) \top}\mathbf{\hat y}^{c (j)} - \mathbf{\mu}_{\lambda_j} \mathbf{K}^{-1}\mathbf{\mu}_{\lambda_j}\big),
    \end{aligned}
\end{equation*}
and $\mathbf{\hat y}^{c (j)} $ is the $j$-th column of $\mathbf{\hat Y}^c$. The posterior mean for ${\mathbf{\tilde \Lambda}}$ is given by
\begin{equation}\label{eq:mu_Lambda}
    \mathbf{\mu}_{\Lambda} = [
        \mathbf{\mu}_{\lambda_1} 
        \cdots  
        \mathbf{\mu}_{\lambda_p}
    ]^\top =\mathbf{\hat Y}^{c\top}\mathbf{\hat M} \big(\mathbf{\hat M}^\top\mathbf{\hat M} + \tau_{\Lambda}^{-2} \mathbf{I}_{k_0} \big)^{-1}  = \frac{n^{1/2}}{n + 1/\tau_{\Lambda}^2} \mathbf{V}^{c} \mathbf{D}^c,
\end{equation}
where $\mathbf{D}^c = \text{diag}(d_1^c, \dots, d_{k_0}^c)$, with $d_j^c$ being the $j$-th largest singular value of $\mathbf{\hat Y}^c$ and $\mathbf{V}^c \in \mathbb R^{p \times k_0}$ is the matrix of corresponding right singular vectors. The posterior mean $\mathbf{\mu}_\Lambda$ is the solution to a ridge regression problem \citep{ridge} where $\mathbf{\hat M}$ is treated as the observed design matrix and $\mathbf{\hat Y}^c$ as the matrix of outcome variables:
\begin{equation*}
    \mathbf{\mu}_\Lambda = \underset{ {\mathbf{\Lambda}} \in \mathbb R^{p \times k_0}}{\operatorname{argmin}} ||\mathbf{\hat Y}^c -\mathbf{\hat M} {\mathbf{\Lambda}}^\top||_2^2 + \frac{1}{\tau_{\Lambda}^2}||{\mathbf{\Lambda}}||_2^2.
\end{equation*}
\begin{remark}
    {There is a large literature developing theory of spectral estimators for high-dimensional factor models \citep{Bai2003InferentialTheory, Doz2012QuasiMaximum, Bai2016MaximumLikelihood, Bai2020SimplerProofs}. In such cases, the estimate for the loading matrix can be interpreted as the OLS solution in a surrogate regression task where the spectral estimator for latent factors is used as the observed covariates matrix. The estimator in \eqref{eq:mu_Lambda} differs from existing approaches since the procedure to estimate latent factors is novel and due to the shrinkage effect induced by the prior distribution.}
    \end{remark}
The induced posterior mean for ${\mathbf{\Lambda}} {\mathbf{\Lambda}}^\top$ is
\begin{equation*}\label{eq:post_mean_Lambda_outer}
    \mathbf{\mu}_\Lambda \mathbf{\mu}_\Lambda^\top + \mathbf{\Psi}, \quad \text{where }  \mathbf{\Psi} = \frac{k_0 \mathbf{\gamma}_n}{(n + \tau_{\Lambda}^{-2})(\mathbf{\gamma}_n - 2)}\text{diag}\big(  \delta_1^2, \dots, \delta_p^2\big).
\end{equation*}
This procedure leads to excellent empirical performance and $\mathbf{\mu}_\Lambda \mathbf{\mu}_\Lambda^\top + \mathbf{\Psi}$ can be shown to be a consistent estimate for ${\mathbf{\Lambda}} {\mathbf{\Lambda}}^\top$ as the sample sizes and outcome dimension diverge. However, credible intervals do not have valid frequentist coverage and suffer from mild undercoverage.

To amend this, we propose a simple coverage correction strategy. In particular, we inflate the conditional variance of each $\mathbf{\tilde \lambda}_j$ by a factor $\rho_{\Lambda}^2>1$, which is tuned to achieve asymptotically valid frequentist coverage. More specifically, we introduce the coverage corrected posterior, where 
\begin{equation}\label{eq:lambda_sigma_cc}
\begin{aligned}
    (\mathbf{\tilde \lambda}_j, \tilde \sigma_j^2) \mid\mathbf{\hat y}^{c (j)} \sim & NIG \big(\mathbf{\tilde \lambda}_j, \tilde \sigma_j^2; \mathbf{\mu}_{\lambda_j}, \rho_{\Lambda}^2 K, \mathbf{\gamma}_n/2, \mathbf{\gamma}_n \delta_j^2/2\big)\\
    &= N_{k_0}\big(\mathbf{\tilde \lambda}_j; \mathbf{\mu}_{\lambda_j}, \tilde  \sigma_j^2 \rho_{\Lambda}^2 K\big)IG \big(\tilde \sigma_j^2; \mathbf{\gamma}_n/2, \mathbf{\gamma}_n \delta_j^2/2\big).
\end{aligned}
\end{equation}
In this case, the posterior mean for ${\mathbf{\Lambda}} {\mathbf{\Lambda}}^\top$ is trivially modified to $\mathbf{\mu}_\Lambda \mathbf{\mu}_\Lambda^\top + \rho_{\Lambda}^2\mathbf{\Psi}$.

Finally, we infer study-specific loadings for study $s$, considering the surrogate multivariate linear regression 
problem,
\begin{equation*} 
   \mathbf{Y}_s =\mathbf{\hat M}_s {\mathbf{\tilde \Lambda}}^\top + \mathbf{\hat  F_s} {\mathbf{\tilde \Gamma}}_s^\top +  \mathbf{\tilde E_s}, 
\end{equation*}
with the new parameter ${\mathbf{\tilde \Gamma}}_s = \big[
    \mathbf{\tilde \gamma}_{s1} ~ \cdots ~ \mathbf{\tilde \gamma}_{sp}\big]^\top$
Again, we adopt conjugate priors on rows of ${\mathbf{\tilde \Gamma}}_s$, 
\begin{equation*}
     \mathbf{\tilde \gamma}_{sj} \mid \tilde \sigma_{sj}^2 \sim N_{q_s} \big(0, \tau_{\Gamma_s}^2 \tilde \sigma_{sj}^2 \mathbf{I}_{q_s}\big), 
\quad (j=1, \dots, p).
      \end{equation*}
By propagating uncertainty about ${\mathbf{\tilde \Lambda}}$ and $\mathbf{\tilde \Sigma}$, a sample for $\mathbf{\tilde \gamma}_{sj}$ can be obtained as 
\begin{equation*}
    \mathbf{\tilde \gamma}_{sj} \mid \mathbf{Y}_s,\mathbf{\hat M}_s, \mathbf{\hat  F_s}, \mathbf{\tilde \lambda}_j , \tilde \sigma_{j}^2  \sim N_{q_s} \Big( \frac{1}{n_s + \tau_{\Gamma_s}^{-2}} \mathbf{\hat  F_s}^\top\mathbf{\tilde y}_{sj}, \frac{1}{n_s + \tau_{\Gamma_s}^{-2}} \tilde \sigma_{j}^2 \mathbf{I}_{q_s} \Big), \quad \mathbf{\tilde y}_{sj} = \mathbf{y}_{sj} -\mathbf{\hat M}_s \mathbf{\tilde \lambda}_j,
\end{equation*}
where a sample from $\mathbf{\tilde \lambda}_j$ is obtained via the previous step. Similarly as above, this procedure leads to estimates with excellent empirical performance but credible intervals with mild undercoverage. Therefore, we apply a similar coverage correction strategy, and sample $ \mathbf{\tilde \gamma}_{sj}$ as
\begin{equation}\label{eq:gamma_cc}
    \mathbf{\tilde \gamma}_{sj} \mid \mathbf{Y}_s,\mathbf{\hat M}_s, \mathbf{\hat  F_s}, \mathbf{\tilde \lambda}_j,  \tilde \sigma_{j}^2  \sim N_{q_s} \big( \frac{1}{n_s + \tau_{\Gamma_s}^{-2}} \mathbf{\hat  F_s}^\top \mathbf{\tilde y}_{sj}, \frac{1}{n_s + \tau_{\Gamma_s}^{-2}} \rho_{\Lambda}^2 \tilde \sigma_{j}^2\mathbf{I}_{q_s}  \big), \quad \mathbf{\tilde y}_{sj} = \mathbf{y}_{sj} -\mathbf{\hat M}_s \mathbf{\tilde \lambda}_j.
\end{equation}
Section \ref{subsec:rho} describes how to tune the factors $\rho_{\Lambda}$ and $\rho_{\Gamma_s}$'s.
The posterior mean for ${\mathbf{\tilde \Gamma}}_s$ is $\mathbf{\mu}_{\Gamma_s} = [
    \mathbf{\mu}_{\gamma_{s1}} ~  \cdots ~   \mathbf{\mu}_{\gamma_{sp}}]^\top$
where
\begin{equation}\label{eq:mu_Gamma_s}
    \mathbf{\mu}_{\gamma_{sj}} =  \frac{1}{n_s + \tau_{\Gamma_s}^{-2}} \mathbf{\hat  F_s}^\top    \big(\mathbf{y}_{sj} -\mathbf{\hat M}_s \mathbf{\mu}_{\lambda_j}\big).
\end{equation}
Similarly as above, $ \mathbf{\mu}_{\Gamma_{s}}$ admits an interpretation as a regularized ordinary least squares estimator:
\begin{equation*}
   \mathbf{\mu}_{\Gamma_{s}} = \underset{ {\mathbf{\Gamma}}_s \in \mathbb R^{p \times q_s}}{\operatorname{argmin}} || \mathbf{\bar Y}_s - \mathbf{\hat  F_s} {\mathbf{\Gamma}}_s^\top||_2^2 +\frac{1}{\tau_{\Gamma_s}^2}||{\mathbf{\Gamma}}_s||_2^2, \quad \text{where }\mathbf{\bar Y_s} = \mathbf{Y}_s -\mathbf{\hat M} \mathbf{\mu}_\Lambda^\top.
\end{equation*}
Letting $ \mathbf{\Psi}_s = \frac{q_s \rho_{\Gamma_s}^2 \mathbf{\gamma}_n}{(n_s + \tau_{\Gamma_s}^{-2})(\mathbf{\gamma}_n - 2)}\text{diag}\big(  \delta_1^2, \dots, \delta_p^2\big)$ and $\mathbf{\tilde \Psi}_s = \frac{\rho_{\Lambda}^2 \mathbf{\gamma}_n}{(n + \tau_{\Lambda}^{-2})(\mathbf{\gamma}_n - 2)} \text{diag}\big(  \delta_1^2 \psi_{s}, \dots, \delta_p^2\psi_{s}\big)$, where 
  and $\psi_{s} = E[\mathbf{\eta}^\top\mathbf{\hat M}_s^\top \mathbf{\hat  F_s} \mathbf{\hat  F_s}^\top\mathbf{\hat M}_s  \mathbf{\eta}]$ with $\mathbf{\eta} \sim N_{k_0}(0, \mathbf{I}_{k_0})$ and the expectation being taken with respect to $\mathbf{\eta}$, the induced posterior mean for ${\mathbf{\tilde \Gamma}}_s {\mathbf{\tilde \Gamma}}_s^\top$ is 
\begin{equation*}\begin{aligned}
     \mathbf{\mu}_{\Gamma_s} \mathbf{\mu}_{\Gamma_s}^\top + \mathbf{\Psi}_s +  \mathbf{\tilde \Psi}_s  \approx  \mathbf{\mu}_{\Gamma_s} \mathbf{\mu}_{\Gamma_s}^\top + \mathbf{\Psi}_s,
\end{aligned}
\end{equation*}
where the approximation is due to the fact $\mathbf{ \hat M_s}^\top \mathbf{\hat  F_s} \approx 0$ as shown in the supplemental. We denote the distribution of ${\mathbf{\tilde \Lambda}}, \{{\mathbf{\tilde \Gamma}}_s\}_s$ and $\mathbf{\tilde \Sigma}$ by $\tilde \Pi$. 
Due to the conjugate prior specification, we can sample independently from $\tilde \Pi$, leading to massive improvements over Gibbs samplers, which tend to have slow mixing.
{Moreover, differently from other spectral estimators for model \eqref{eq:fact_model} \citep[e.g.][]{Ando2016PanelData, Ando2017ClusteringHuge} \texttt{BLAST} does not need to iterate between estimating shared factor analytic components and study-specific ones.}

\subsection{Estimation of latent dimensions}\label{subsec:estimation_k}
The latent dimensions are not known in general and need to be estimated. {\citet{Ando2017ClusteringHuge} develop a strategy for estimating latent dimensions in model \eqref{eq:fact_model} that requires fitting each candidate model, which may be computationally intensive. Here, we propose a faster} strategy based on a combination of information criteria and an elbow rule. First, we estimate the latent dimension of study $s$, $k_s = k_0 + q_s$, by minimizing the joint likelihood based information criterion (JIC) proposed in \citet{chen_jic},
\begin{equation*}\label{eq:JIC}
    \text{JIC}_s(k) = -2 l_{sk} + k \max(n_s,p) \log \{\min(n_s,p)\},
\end{equation*}
where $l_{sk}$ is the value of the joint log-likelihood for the $s$-th study computed at the joint maximum likelihood estimate when the latent dimension is equal to $k$. Calculating the joint maximum likelihood estimate for each value of $k$ can be computationally expensive. Therefore, we approximate $l_{sk}$ with $ l_{sk} \approx \hat l_{sk} $, where $\hat l_{sk}$ is the likelihood of the study-specific data matrix where we estimate the latent factors as the left singular vectors associated to the $k$ leading singular values of $\mathbf{Y}_s$ scaled by $\sqrt{n_s}$ and the factor loadings by their conditional mean given such an estimate. More details are provided in the supplementary material.  
Thus, we set
\begin{equation}\label{eq:k_s_hat}
    \hat k_s = \underset{k_s=1, \dots, k_{max}}{\operatorname{argmin}} \hat{ \text{JIC}}_s(k_s), \qquad \hat{ \text{JIC}}_s(k)= -2 \hat l_{sk} + k \max(n_s,p) \log\{\min(n_s,p)\},
\end{equation}
where $k_{max}$ is a conservative upper bound to the latent dimension. 

Next, $k_0$ is estimated by leveraging the gap in the spectrum of $\mathbf{\tilde P}$, defined as in \eqref{eq:P_tilde}. In particular, we set $\hat k_0$ to the number of singular values of $\mathbf{\tilde P}$ that are larger than $1- \tau$ 
\begin{equation}\label{eq:k_hat}
    \hat k_0 = \underset{j=1, \dots, \min\{\hat k_1, \dots, \hat k_S\}}{\operatorname{argmax}} \big\{ \, j \mid  s_{j}(\mathbf{\tilde P}) > 1 - \tau \, \big\}, 
\end{equation}
for some small threshold $\tau$. The rationale for this choice follows from the separation in the spectrum of $\mathbf{\tilde P}$.
Recall that $\mathbf{\tilde P}$ is expected to have $k_0$ singular vectors with singular values close to $1$, corresponding to the directions that are repeated across studies, and the remaining ones with singular values $\ll 1$, corresponding to study-specific directions. 

Finally, given the estimates $\hat k_0$ and $\hat k_s$, we let the estimate of the number of factor loadings specific to the study $q_s$ simply be $\hat q_s = \hat k_s - \hat k_0$. {This procedure picked the correct latent dimensions in all the numerical experiments considered in Section \ref{sec:numerical_experiments}}.

\subsection{Variance inflation terms}\label{subsec:rho}
This section discusses how the variance inflation terms are tuned, which is inspired by \citet{fable}. In particular, we let 
\begin{equation}\label{eq:b}
    b_{jj'} = \begin{cases}
        \big( 1 +  \frac{||\mathbf{\mu}_{\lambda_j}||^2 ||\mathbf{\mu}_{\lambda_{j'} } ||^2 + (\mathbf{\mu}_{\lambda_j}^\top \mathbf{\mu}_{\lambda_{j'}})^2 }{ V_j^2 ||\mathbf{\mu}_{\lambda_{j'}} ||^2 + V_{j'}^2||\mathbf{\mu}_{\lambda_{j}} ||^2} \big)^{1/2}, \quad &\text{if } j \neq j'   \\   \big(1 + \frac{ ||\mathbf{\mu}_{\lambda_j} ||^2 }{ 2V_j^2 }\big)^{1/2}, \quad &\text{otherwise,} 
              \end{cases}
\end{equation}
and, for $s=1, \dots, S$, 
\begin{equation}\label{eq:b_s}
    b_{sjj'} = \begin{cases}
    \begin{aligned}
        \big[ &1 +  (V_j^2 ||\mathbf{\mu}_{\gamma_{sj'}}||^2 + V_{j'}^2||\mathbf{\mu}_{\gamma_{sj}} ||^2)^{-1}\big(||\mathbf{\mu}_{\gamma_{sj}}||^2 ||\mathbf{\mu}_{\gamma_{sj'} } ||^2 + (\mathbf{\mu}_{\gamma_{sj}}^\top \mathbf{\mu}_{\gamma_{sj'}})^2  \\ &+ ||\mathbf{\mu}_{\gamma_{sj}} ||^2 ||\mathbf{\mu}_{\lambda_{j'}}||^2 
      + ||\mathbf{\mu}_{\gamma_{sj'}} ||^2 ||\mathbf{\mu}_{\lambda_{j}}||^2+  2\mathbf{\mu}_{\gamma_{sj}} ^\top \mathbf{\mu}_{\gamma_{sj'}} \mathbf{\mu}_{\lambda_{j}}^\top \mathbf{\mu}_{\lambda_{j'}}\big) \big]^{1/2} 
    \end{aligned}
       , \quad &\text{if } j \neq j'   \\   \big(1 + \frac{ ||\mathbf{\mu}_{\gamma_{sj}} ||^2 +2||\mathbf{\mu}_{\lambda_{j}}||^2 }{ 2V_j^2 }\big)^{1/2}, \quad &\text{otherwise,} 
              \end{cases}
\end{equation}
where $V_j = ||(\mathbf{I}_n - \mathbf{U}^c \mathbf{U}^{c\top})\mathbf{\hat y}^{c (j)}||^2 / n$. Setting $\rho_{\Lambda} =  b_{jj'}$ and $\rho_{\Gamma_s} = b_{sjj'}$ ensures that the credible intervals for the $j,j'$-th elements of ${\mathbf{\Lambda}} {\mathbf{\Lambda}}^\top$ and ${\mathbf{\Gamma}}_s {\mathbf{\Gamma}}_s^\top$ have valid asymptotic coverage. Then, choosing $\rho_{\Lambda} = \max_{j,j'} b_{jj'}$ and $\rho_{\Gamma_s} = \max_{j,j'} b_{sjj'}$ guarantees entrywise asymptotic valid coverage of credible intervals for $\mathbf{\Lambda}\mathbf{\Lambda}^\top$ and $\mathbf{\Gamma}_s \mathbf{\Gamma}_s^\top$ respectively. Alternatively, if $\rho_{\Lambda} = \binom{p}{2}^{-1} \sum_{1 \leq j \leq j' \leq p} b_{jj'}$ and $\rho_{\Gamma_s} = \binom{p}{2}^{-1}\sum_{1 \leq j \leq j' \leq p} b_{sjj'}$ credible intervals have approximately correct valid coverage on average, which is our default choice. We refer to section \ref{sec:theory} for a more in-depth discussion on the impact of the choice of the variance inflation terms on coverage properties of credible intervals.

\subsection{Hyperparameter selection}\label{subsec:hyperparam}
We estimate the prior variances $\tau_{\Lambda}$ and $\{\tau_{\Gamma_s}\}_{s=1}^S$ in a data-adaptive manner as follows. The conditional prior expectation of the squared Frobenius norm of ${\mathbf{\Lambda}}$ can be expressed as $E(||{\mathbf{\Lambda}}||^2 \mid \sigma_1^2, \dots, \sigma_p^2, \tau_{\Lambda}) = k_0 \tau_{\Lambda} \sum_{j=1}^p \sigma_j^2$. We let $\Omega = \sum_{j=1}^p V_j$, where $V_j = ||(\mathbf{I}_n - \mathbf{U}^c \mathbf{U}^{c\top})\mathbf{\hat y}^{c (j)}||^2/n$ as in Section \ref{subsec:rho} and $\Theta = ||\mathbf{\hat Y}^{c \top} \mathbf{U}^c \mathbf{U}^{c\top}\mathbf{\hat Y}^c||^2/n$, which are consistent estimators for $\sum_{j=1}^p \sigma_j^2$ and $||{\mathbf{\Lambda}}||^2$ respectively, and estimate $\tau_{\Lambda}$ via $\hat \tau_{\Lambda} = \frac{\Theta}{k_0 \Omega }$. Analogously, we estimate $\tau_{\Gamma_s}$ via $\hat \tau_{\Gamma_s} = \frac{\Theta_s}{q_s \Omega}$, where $\Theta_s = || \mathbf{\tilde Y_s}^{\top} \mathbf{U}_s^\perp \mathbf{U}_s^{\perp \top} \mathbf{ \tilde Y}^c||^2/n_s$ and $\mathbf{\tilde Y_s} = \mathbf{Y}_s -\mathbf{\hat M}_s \mathbf{\mu}_\Lambda^\top$.
We set $\nu_0$ and $\sigma_0^2$ to 1 as a default value.

Algorithm \ref{alg:ms_fable} summarizes the procedure obtained by combining all the previous sections.

\begin{algorithm}
\caption{\texttt{BLAST} procedure to obtain $N_{MC}$ approximate posterior samples.}\label{alg:ms_fable}\vspace{-2em}
\begin{tabbing}
   \qquad \enspace Input: The data matrices $ \{\mathbf Y_s\}_{s=1}^S$, the number of Monte Carlo samples $N_{MC}$, \\ \qquad \enspace \hspace{2.5em} and an upper bound on the number of factors $k_{max}$. \\
   \qquad \enspace Step 1: For each $s=1, \dots, S$, estimate the number of latent factors for each study \\ \qquad \enspace \hspace{2.8em} $k_s$ via equation \eqref{eq:k_s_hat}.\\
   
   \qquad \enspace  \text{Step 2:} Obtain the estimates for the latent factors $\{\mathbf{\hat  F_s}\}_{s=1}^S$ and $\mathbf{\hat M} = [\mathbf{\hat M}_1^\top \dots\mathbf{\hat M}_S^\top]^\top $  \\ \qquad \enspace \hspace{2.8em} and the data matrix with the shared variation $\mathbf{\hat Y}^c$ using Algorithm \ref{alg:factor_pretraining}. \\
\qquad \enspace  \text{Step 3:} Compute the variance inflation term for ${\mathbf{\Lambda}}$, as $\rho_{\Lambda} =\binom{p}{2}^{-1} \sum_{1 \leq j \leq j' \leq p} b_{jj'}$ \\ \qquad \enspace \hspace{2.8em} where the $b_{jj'}$'s are defined in \eqref{eq:b}.
    
\\ \qquad \enspace \text{Step 4:} Estimate the hyperparameters $\tau_{\Lambda}$ and $\{\tau_{\Gamma_s}\}_{s=1}^S$ as described in Section \ref{subsec:rho}.\\

\qquad \enspace \text{Step 5:} Estimate the mean for $ {\mathbf{\tilde \Lambda}}$, $\mathbf{\mu}_\Lambda$, as in \eqref{eq:mu_Lambda} and, for each $j=1, \dots, p$ in \\ \qquad \enspace \hspace{2.8em} parallel, for $t = 1, \dots, N_{MC}$, sample independently $(  \mathbf{\tilde \lambda}_j^{(t)},  \tilde \sigma_j^{2(t)})$ from \eqref{eq:lambda_sigma_cc}.  \\
\qquad \enspace  \text{Step 6:}  For each $s=1, \dots, S$, compute the variance inflation term for ${\mathbf{\Gamma}}_s$, as \\ \qquad \enspace \hspace{2.8em} $\rho_{\Gamma_s} =\binom{p}{2}^{-1} \sum_{1 \leq j \leq j' \leq p} b_{sjj'}$ where the $b_{sjj'}$'s are defined in \eqref{eq:b_s}.
 \\
\qquad \enspace  \text{Step 7:} For each $s=1, \dots, S$, estimate the mean for $ \tilde {\mathbf{\Gamma}}_s$, $\mathbf{\mu}_{\Gamma_s}$, as in \eqref{eq:mu_Gamma_s} and, for each \\ \qquad \enspace \hspace{2.8em}  $j=1, \dots, p$ in parallel, for $t = 1, \dots, N_{MC}$, sample independently $\mathbf{\tilde \gamma}_{sj}^{(t)} $ \\ \qquad \enspace \hspace{2.8em} from \eqref{eq:gamma_cc}. 

\\ \qquad \enspace \text{Output:} $N_{MC}$ samples of the shared low rank components \\ \qquad \enspace \hspace{2.8em} $ {\mathbf{\tilde \Lambda}}^{(1)} {\mathbf{\tilde \Lambda}}^{(1)\top}, \ldots,  {\mathbf{\tilde \Lambda}}^{(N_{MC})} {\mathbf{\tilde \Lambda}}^{(N_{MC})\top}$, 
of study-specific low rank components \\ \qquad \enspace \hspace{2.8em}  $ {\mathbf{\tilde \Gamma}}_s^{(1)} {\mathbf{\tilde \Gamma}}_s^{(1)\top}, \ldots,  {\mathbf{\tilde \Gamma}}_s^{(N_{MC})} {\mathbf{\tilde \Gamma}}_s^{(N_{MC})\top}$ for  $s=1, \dots, S$, and the residual variances \\ \qquad \enspace \hspace{2.8em} $\{ \tilde \sigma_j^{2 (1)}\}_{j=1}^p, \dots,\{ \tilde \sigma_j^{2 (N_{MC})}\}_{j=1}^p$.\\

\end{tabbing}
\end{algorithm}

\section{Theoretical support}\label{sec:theory}
In this section, we present theoretical support for our methodology. We show favorable properties in the double asymptotic regime, that is, when both the sample sizes and data dimension diverge.
We start by defining some regularity conditions.

\begin{assumption}\label{assumption:model}
    Data are generated under the following model
    {\begin{equation}\label{eq:model_0}
        Y_s = \mathbf{M}_{0s} \mathbf{\Lambda}_0^\top +  \mathbf{F}_{0s} \mathbf{\Gamma}_{0s}^\top + \mathbf E_s, 
        \quad (s=1, \dots, S),
    \end{equation}} with true shared loading matrix ${\mathbf{\Lambda}}_0$ and study-specific loading matrices $\{{\mathbf{\Gamma}}_{0s}\}_s$.  
    We denote by $\{\mathbf{M}_{0s}\}_s$ and $\{\mathbf{F}_{0s}\}_s$ the true latent factors responsible for the shared variation and study-specific variation, respectively. 
    We let $\mathbf{M}_0 =  \big[ \mathbf{M}_{01}^\top ~  \cdots ~  \mathbf{M}_{0S}^\top \big]^\top$. {We consider the number of studies $S \in \mathbb N$ as fixed.}
\end{assumption}
{
\begin{assumption}\label{assumption:distributions}
     We assume the idiosyncratic errors to be zero-mean sub-Gaussian random variables with row and column covariances $\{ \mathbf{\Sigma}_{0sr}, \mathbf{\Sigma}_{0sc} \}_s$, and to be independent across studies. We assume the entries of $\{\mathbf{M}_{0s}\}_s$ and $\{\mathbf{F}_{0s}\}_s$ to be independent standard Gaussian random variables
\end{assumption}
}
{The model in \eqref{eq:model_0} allows cross-sectional or across-units dependence, as $\mathbf{\Sigma}_{0sr}$ and $\mathbf{\Sigma}_{0sc}$ are not required to be diagonal. Assumptions \ref{assumption:model} and \ref{assumption:distributions} are common in the Bayesian multi-study factor modeling literature \citep{msfa, bmsfa, vi_bmsfa} and allow full probabilistic inference on model unknowns. In the supplemental, we show robustness to model misspecification by extending the consistency results to the case where only moment assumptions are made. }

\begin{assumption}[Linear Independence of loading matrices]
\label{assumption:li}
   The matrix $\big[
        {\mathbf{\Lambda}_0} ~  {\mathbf{\Gamma}}_{01} ~  \cdots ~  {\mathbf{\Gamma}}_{0S}
  \big] \in \mathbb R^{p \times k_0 + \sum_{s=1}^S q_s}$ has full column rank. 
\end{assumption}
Assumption \ref{assumption:li} implies $\mathcal C({\mathbf{\Lambda}_0}) \cap \mathcal C({\mathbf{\Gamma}}_{0s}) = \{0\}$ for every $s$ and $\mathcal C({\mathbf{\Gamma}}_{0s}) \cap \mathcal C({\mathbf{\Gamma}}_{0s'}) = \{0\}$ for $s \neq s'$, where $\mathcal C(A) $ denotes the column space of $A$. This is a common requirement in the literature \citep{msfa} to ensure identifiability of model \eqref{eq:fact_model}.
\begin{assumption}\label{assumption:Lambda}
    $s_l({\mathbf{\Lambda}}_0) \asymp ||{\mathbf{\Lambda}}_0|| \asymp \sqrt{p}$ for $l=1, \dots k$, $||{\mathbf{\Lambda}_0}||_{\infty} < \infty$, and $\min_{j=1, \dots, p }||\mathbf{\lambda}_{0j}^2||>c_{\mathbf{\lambda}}$ for some constant $c_{\mathbf{\lambda}}>0$.
\end{assumption}

{Assumption \ref{assumption:Lambda} requires all singular values of $\mathbf \Lambda_0$ to scale as $\sqrt{p}$ as $p \to \infty$, which is required for the signal to be identifiable from the noise in the asymptotic regime, and excludes the presence of outcomes with null or vanishing signals. These are standard requirements in high-dimensional factor analysis \citep{Ando2017ClusteringHuge, Bai2020SimplerProofs}.} 
\begin{assumption}\label{assumption:Gammas}
 $s_{l}\{(\mathbf{I}_p - \mathbf{V}_0\mathbf{V}_0^\top ) {\mathbf{\Gamma}}_{0s}\}  \asymp \sqrt{p }$ for $l=1, \dots q_s$, $||{\mathbf{\Gamma}}_{0s}||_{\infty} < \infty$, and $\min_{j=1, \dots, p }||\mathbf{\gamma}_{0j}^2||>c_{\gamma}$ for some constant $c_{\gamma}>0$, for $s=1, \dots, S$, where $\mathbf{V}_0 \in \mathbb R^{p \times k_0}$ is the matrix of left singular vectors of ${\mathbf{\Lambda}}_0$.
\end{assumption}
{Assumption \ref{assumption:Gammas} implies similar requirements on $\mathbf \Gamma_{0s}$ to those of Assumption \ref{assumption:Lambda}, but with different conditions on the singular values imposed on the 
study-specific loading matrix after projecting out the shared loading matrix. This is required to be able to recover the study-specific signal in Step 3 of Algorithm \ref{alg:factor_pretraining}. This requirement can equivalently be expressed in terms of the singular values of the components of $\mathbf \Gamma_{0s}$ that lie in the null space of $\mathbf \Lambda_0$. }

\begin{assumption}\label{assumption:sv_A}
$s_1\big(\frac{1}{S} \sum_{s=1}^S \mathbf{V}_{0s}^\perp \mathbf{V}_{0s}^{\perp \top}\big) \leq 1- \delta$ for some $\delta >0$, {which does not depend on $p$}, eventually as $p \to \infty$, where $\mathbf{V}_{0s}^\perp\in \mathbb R^{p \times q_s}$ is the matrix of left singular vectors of $(\mathbf{I}_p - \mathbf{V}_0\mathbf{V}_0^\top ) {\mathbf{\Gamma}}_{0s}$.
\end{assumption}
{ Assumption \ref{assumption:sv_A} fails if the column spaces of the components of the $\mathbf \Gamma_{0s}$'s in the null space of $\mathbf{\Lambda}_{0}$ have a common axis up to a perturbation decreasing to 0. However, as long as this does not hold for at least one study, the condition is satisfied. Since the $\mathbf \Gamma_{0s}$'s model batch effects, which tend to vary substantially across studies, we expect this requirement to be satisfied in most cases.}
{\begin{assumption}\label{assumption:sigma}
$s_1\big(\mathbf{\Sigma}_{0sr}\big) \lesssim C_{sr, n_s}^2 \lesssim (\log n_s)^c$ and $s_1\big(\mathbf{\Sigma}_{0sc}\big) \lesssim C_{sc, p}^2 \lesssim (\log p)^c$, for $s=1, \dots, S$ for some fixed $c < \infty$. Letting $ C_{r, \{n_s\}_{s}^{S}} = \sum_{s=1}^S  C_{sr, n_s}$, $ C_{c, p} = \max_{s=1, \dots, S}  C_{sc, p}$, and $\sigma_{0sri}$ and $\sigma_{0scj}$ be the $i$-th and $j$-th diagonal elements of $\mathbf{\Sigma}_{0sr}$ and $\mathbf{\Sigma}_{0sc}$ respectively, we have $\max_{s=1,\dots, S; i=1,\dots, n_s} \sigma_{0sri}^2 \leq C_{\sigma}^2$ and $\max_{s=1,\dots, S; j=1,\dots, p} \sigma_{0scj}^2 \leq C_{\sigma}^2$ for some constant 
    $C_{\sigma} < \infty$.
    \end{assumption}
Assumption \ref{assumption:sigma} controls the violation from the modeling assumption of independence of the idiosyncratic components.}

First, we present a result which bounds the Procrustes error for the latent factor estimates. 
\begin{theorem}[Recovery of latent factors]\label{thm:factors_procrustes_error}
   Suppose Assumptions \ref{assumption:model}--\ref{assumption:sigma} hold and $n_s = \mathcal O(n_{\min}^2)$, where $n_{\min} = \min_{s=1, \dots, S} n_s$, for all $s=1, \dots, S$, then, as $n_1, \dots, n_s, p \to \infty$,
    with probability at least $1-o(1)$,
   { \begin{equation*}
        \begin{aligned}
           \min_{\mathbf{R}_s \in \mathbb R^{k_0 \times k_0} : \mathbf{R}_s^\top \mathbf{R}_s = \mathbf{I}_{k_0}} \frac{1}{\sqrt{n_s}}||\mathbf{\hat M}_s \mathbf{R}_s- \mathbf{M}_{0s}||& \lesssim \frac{1}{\sqrt{n_s}} + (C_{r, \{n_s\}_{s}^{S}} C_{c, p})^2\frac{\sqrt{n/n_s}}{p},\\
             \min_{\mathbf{R}_s \in \mathbb R^{q_s \times q_s} : \mathbf{R}_s^\top \mathbf{R}_s = \mathbf{I}_{q_s}} \frac{1}{\sqrt{n_s}}|| \mathbf{\hat  F_s} \mathbf{R}_s- \mathbf{F}_{0s}||& \lesssim \frac{1}{\sqrt{n_s}} + \frac{(C_{r, \{n_s\}_{s}^{S}} C_{c, p})^2}{p}.
        \end{aligned}
    \end{equation*}}
\end{theorem}
Theorem \ref{thm:factors_procrustes_error} supports the use of spectral decomposition-based estimates for latent factors. We consider the Procrustes error, since the latent factors and loadings in \eqref{eq:fact_model} are identifiable only up to orthogonal transformations.
\begin{remark}
{Theorem 1 establishes a similar result to Proposition 2 of \citet{Bai2003InferentialTheory} and Theorem 3.3 of \citet{Fan2013LargeCovariance}, which analyze spectral estimators for single-study high dimensional factor models but is weaker in not providing row-wise error control.}
\end{remark}

The next theorem characterizes the consistency of point estimators and posterior contraction around the true parameters.
\begin{theorem}[Consistency and posterior contraction]\label{thm:posterior_contraction_Lambda_outer}
    Suppose Assumptions \ref{assumption:model}--\ref{assumption:sigma} hold and $n_s = \mathcal O(n_{\min}^2)$, where $n_{\min} = \min_{s=1, \dots, S} n_s$, for all $s=1, \dots, S$, then, as $n_1, \dots, n_s, p \to \infty$,
    with probability at least $1-o(1)$,
    {\begin{equation}
    \begin{aligned}
          \frac{ \left|\left| \mathbf{\mu}_\Lambda \mathbf{\mu}_\Lambda^\top - {\mathbf{\Lambda}}_0 {\mathbf{\Lambda}}_0^\top  \right|\right|}{ \left|\left|  {\mathbf{\Lambda}}_0 {\mathbf{\Lambda}}_0^\top  \right|\right|} &\lesssim  (\sqrt{\log n} + C_{r, \{n_s\}_{s}^{S}} C_{c, p}) \frac{1}{\sqrt{n}}  + C_{r, \{n_s\}_{s}^{S}} C_{c, p}\frac{1}{\sqrt{p}},\\ \frac{ \left|\left| \mathbf{\mu}_{\Gamma_s} \mathbf{\mu}_{\Gamma_s}^\top - {\mathbf{\Gamma}}_{0s} {\mathbf{\Gamma}}_{0s}^\top  \right|\right|}{ \left|\left| {\mathbf{\Gamma}}_{0s} {\mathbf{\Gamma}}_{0s}^\top  \right|\right|} &\lesssim (\sqrt{\log n_s} + C_{sr, n_s} C_{sc, p}) \frac{1}{\sqrt{n_s}}  + C_{sr, n_s} C_{sc, p}\frac{1}{\sqrt{p}}, \quad (s=1, \dots, S).
    \end{aligned}
    \end{equation}}
    Moreover, there exist finite constants $D, D_1, \dots, D_S < \infty $ such that
    {\begin{equation}
    \begin{aligned}
        E\left(\tilde \Pi\left[\frac{ \left|\left| {\mathbf{\tilde \Lambda}} {\mathbf{\tilde \Lambda}}^\top - {\mathbf{\Lambda}}_0 {\mathbf{\Lambda}}_0^\top  \right|\right|}{ \left|\left|  {\mathbf{\Lambda}}_0 {\mathbf{\Lambda}}_0^\top  \right|\right|} > D \left\{(\sqrt{\log n} + C_{r, \{n_s\}_{s}^{S}} C_{c, p}) \frac{1}{\sqrt{n}}  + C_{r, \{n_s\}_{s}^{S}} C_{c, p}\frac{1}{\sqrt{p}} \right\}\right] \right) &\to 0,\\
        E\left(\tilde \Pi\left[\frac{ \left|\left| {\mathbf{\tilde \Gamma}}_s {\mathbf{\tilde \Gamma}}_s^\top -  {\mathbf{\Gamma}}_{0s} {\mathbf{\Gamma}}_{0s}^\top  \right|\right|}{ \left|\left|   {\mathbf{\Gamma}}_{0s} {\mathbf{\Gamma}}_{0s}^\top  \right|\right|} > D_s \left\{(\sqrt{\log n_s} + C_{sr, n_s} C_{sc, p}) \frac{1}{\sqrt{n_s}}  + C_{sr, n_s} C_{sc, p}\frac{1}{\sqrt{p}}\right\} \right] \right) &\to 0, \quad (s=1, \dots, S).
    \end{aligned}
    \end{equation}}
\end{theorem}
The first part of Theorem \ref{thm:posterior_contraction_Lambda_outer} justifies the use of $\mathbf{\mu}_\Lambda \mathbf{\mu}_\Lambda^\top$ and $\mathbf{\mu}_{\Gamma_s} \mathbf{\mu}_{\Gamma_s}^\top$ as point estimates for ${\mathbf{\Lambda}} {\mathbf{\Lambda}}^\top$ and ${\mathbf{\Gamma}}_s {\mathbf{\Gamma}}_s^\top$ in the high-dimensional and high-sample size limit. The second part of the Theorem is a stronger statement and characterizes the concentration of the measure induced on ${\mathbf{\tilde \Lambda}} {\mathbf{\tilde \Lambda}}^\top$ and ${\mathbf{\tilde \Gamma}}_s {\mathbf{\tilde \Gamma}}_s^\top $ around the true parameter at rates $\frac{1}{n^{1/2}} + \frac{1}{p^{1/2}}$ and $\frac{1}{n_s^{1/2}} + \frac{1}{p^{1/2}}$ (up to logarithmic terms), respectively. In this sense, our method enjoys a blessing of dimensionality, recovering the true parameters if and only if both sample size and data dimension diverge.
We consider relative errors by dividing by the norm of ${\mathbf{\Lambda}}_0 {\mathbf{\Lambda}}_0^\top$ and ${\mathbf{\Gamma}}_{0s} {\mathbf{\Gamma}}_{0s}^\top$, respectively, to make the results comparable as $p$ increases. 
\begin{remark}
    {\citet{Bai2020SimplerProofs} upper bound the error of the factor loading matrix for single-study factor models with rate $1/\sqrt{n} + 1/\sqrt{p}$, which implies the same rate (up to logarithmic terms) as ours for the low-rank component. }
\end{remark}

\begin{remark}[Consistency holds under heteroscedasticity]
    The results in Theorem \ref{thm:posterior_contraction_Lambda_outer} hold even under a heteroscedastic design, that is, if the assumption $\mathbf{\Sigma}_{0s}=\mathbf{\Sigma}_0$ for $s=1, \dots, S$ does not hold; this is a misspecified case, as our method of inferring ${\mathbf{\Lambda}}$ assumes homoskedasticity. 
\end{remark}
\begin{remark}[Extension to Frobenius loss] Similar results to those in Theorems \ref{thm:factors_procrustes_error} and \ref{thm:posterior_contraction_Lambda_outer} can be derived for the Frobenius error due to the low-rank structure of the parameters.
\end{remark}

Next, we present results for each entry of the low-rank components. {These results require additional assumptions on the residual error variances. 
\begin{assumption}[No across-units dependence]\label{assumption:no_unit_dependence}
    We assume there is no dependence across-units in the idiosyncratic component, i.e $\mathbf \Sigma_{0sr} = \mathbf I_{n_s}$ for all $s=1, \dots, S$.
\end{assumption}
\begin{assumption}[Homoscedasticity]\label{assumption:homoscedasticity}
We assume the variances of the idiosyncratic components to be equal across studies, i.e. $\sigma_{0scj}^2 = \sigma_{0j}^2$ for some $\sigma_{0j}^2>0$ for all $j=1, \dots, p$ and $s=1, \dots, S$. 
\end{assumption}
\begin{assumption}\label{assumption:sigma_lb}
    We have $ \min_{s=1,\dots, S; j=1,\dots, p} \sigma_{0sj}^2 > c_{\sigma}^2 $ for some constant $c_{\sigma}>0$.
\end{assumption}
\begin{assumption}\label{assumption:Sigma_diagonal}
    We assume the cross-sectional covariances $\{\mathbf \Sigma_{0sc}\}_s$ to be diagonal.
\end{assumption}}
{Assumption \ref{assumption:Sigma_diagonal} assumes that, conditionally on latent factors, there is no residual correlation among variables. In the supplemental, we consider the case where Assumption \ref{assumption:Sigma_diagonal} does not hold.}

The first result is a central limit theorem for point estimators.

\begin{theorem}[Central limit theorem]\label{thm:clt_mu_Lambda_outer}
    Suppose Assumptions \ref{assumption:model}--\ref{assumption:Sigma_diagonal} hold, $n_s = \mathcal O(n_{\min}^2)$, where $n_{\min} = \min_{s=1, \dots, S} n_s$, for all $s=1, \dots, S$, $(C_{r, \{n_s\}_{s}^{S}} C_{c, p})^2\sqrt{n}/p = o(1)$ and $\log^2 p / \sqrt{n_{\min}}$. For $1 \leq j \leq j' \leq p$, let 
    \begin{equation}\label{eq:S_0_sq}
        \mathcal S_{0jj'}^2 = \begin{cases}
             \sigma_{0j}^2 ||\mathbf{\lambda}_{0j'}||^2 + \sigma_{0j'}^2 ||\mathbf{\lambda}_{0j}||^2 + ||\mathbf{\lambda}_{0j}||^2 ||\mathbf{\lambda}_{0j'}||^2 +  (\mathbf{\lambda}_{0j}^\top \mathbf{\lambda}_{0j'})^2 \quad &\text{if } j \neq j',\\
       2 ||\mathbf{\lambda}_{0j}||^4 +4\sigma_{0j}^2 ||\mathbf{\lambda}_{0j}||^2 
\quad &\text{otherwise,}
        \end{cases}
    \end{equation}
and
   \begin{equation}\label{eq:S_0_s_sq}
        \mathcal S_{0sjj'}^2 = \begin{cases}
        \begin{aligned}
               &\sigma_{0j}^2 ||\mathbf{\gamma}_{0sj'}||^2 + \sigma_{0j'}^2 ||\mathbf{\gamma}_{0sj}||^2 + ||\mathbf{\gamma}_{0sj}||^2 ||\mathbf{\gamma}_{0sj'}||^2 +  (\mathbf{\gamma}_{0sj}^\top \mathbf{\gamma}_{0sj'})^2  \quad  \\ & + ||\mathbf{\gamma}_{0sj}||^2 ||\mathbf{\lambda}_{0j}||^2 
             + ||\mathbf{\gamma}_{0sj'}||^2 ||\mathbf{\lambda}_{0j}||^2 +  2\mathbf{\gamma}_{0sj}^\top \mathbf{\gamma}_{0sj'}\mathbf{\lambda}_{0j}^\top \mathbf{\lambda}_{0j'}
        \end{aligned} \quad   &\text{if } j \neq j',
          \\
       2 ||\mathbf{\gamma}_{0sj}||^4 + 4||\mathbf{\gamma}_{0sj}||^2||\mathbf{\lambda}_{0j}||^2 +4 \sigma_{0j}^2||\mathbf{\gamma}_{0sj}||^2 
 &\text{otherwise.}
        \end{cases}
    \end{equation}


    Then, as $n_1, \dots, n_s, p \to \infty$, we have
    \begin{equation}\label{eq:mu_lambda_clt}
    \begin{aligned}
        &\frac{\sqrt{n}}{S_{0, jj'}} \left(\mathbf{\mu}_{\lambda_j}^\top \mathbf{\mu}_{\lambda_{j'}}-   \mathbf{\lambda}_{0j}^\top\mathbf{\lambda}_{0j'}\right) \Longrightarrow N(0, 1),\\
        & \frac{\sqrt{n_s}}{S_{0, sjj'}} \left(\mathbf{\mu}_{\gamma_{sj}}^\top \mathbf{\mu}_{\gamma_{sj'}} - \mathbf{\gamma}_{0sj}^\top \mathbf{\gamma}_{0sj'}\right) \Longrightarrow N(0, 1), \quad (s=1, \dots, S). 
    \end{aligned} 
    \end{equation}
\end{theorem}
While Theorem \ref{thm:posterior_contraction_Lambda_outer} justifies $\mathbf{\mu}_\Lambda\mathbf{\mu}_\Lambda^\top$ and $\mathbf{\mu}_{\Gamma_s} \mathbf{\mu}_{\Gamma_s}^\top$ as point estimators for ${\mathbf{\Lambda}} {\mathbf{\Lambda}}^\top$ and ${\mathbf{\Gamma}}_s {\mathbf{\Gamma}}_s^\top$, respectively, Theorem \ref{thm:clt_mu_Lambda_outer} provides entry-wise control of these estimators in large samples. In particular, for large values of $p$ and sample sizes, 
$\mathbf{\mu}_{\lambda_j}^\top \mathbf{\mu}_{\lambda_{j'}}$ and $\mathbf{\mu}_{\gamma_{sj}}^\top \mathbf{\mu}_{\gamma_{sj'}}$ are approximately normally distributed centered on the corresponding true values with the variance given by $\mathcal S_{0jj'}^2$ divided by the total sample size $n$ and $\mathcal S_{0sjj'}^2$ divided by the study-specific sample size, respectively.
\begin{remark}
    {If Assumption \ref{assumption:homoscedasticity} is not met, results in Theorem \ref{thm:clt_mu_Lambda_outer} still hold replacing $\sigma_{0j}^2$ with the weighted average of study-specific variances. }
\end{remark}

Next, we characterize the asymptotic behavior of the measure of ${\mathbf{\Lambda}}{\mathbf{\Lambda}}^\top$ and ${\mathbf{\Gamma}}_s {\mathbf{\Gamma}}_s^\top$ induced by $\tilde \Pi$.
\begin{theorem}[Bernstein–von Mises theorem]\label{thm:bvm_Lambda_outer}
    If Assumptions \ref{assumption:model}--\ref{assumption:Sigma_diagonal} hold, $n_s = \mathcal O(n_{\min}^2)$, where $n_{\min} = \min_{s=1, \dots, S} n_s$, for all $s=1, \dots, S$, and $\sqrt{n} / p = o(1)$, and for $1 \leq j , j' \leq p$, and $s=1, \dots, S$, let 
    \begin{equation}\label{eq:l_0_sq}
        l_{0, jj'}^2(\rho) = 
        \begin{cases}
           \rho^2 \left(\sigma_{0j}^2 ||\mathbf{\lambda}_{0j'}||^2 + \sigma_{0j'}^2 ||\mathbf{\lambda}_{0j}||^2 \right), \quad &\text{if } j \neq j',\\
4 \rho^2 \sigma_{0j}^2 ||\mathbf{\lambda}_{0j}||^2,  \quad &\text{otherwise,} 
        \end{cases}
    \end{equation}
    and
    \begin{equation}\label{eq:l_0_s_sq}
        l_{0, sjj'}^2(\rho) = 
        \begin{cases}
           \rho^2 \left(\sigma_{0j}^2 ||\mathbf{\gamma}_{0sj'}||^2 + \sigma_{0j'}^2 ||\mathbf{\gamma}_{0sj}||^2 \right), \quad &\text{if } j \neq j',\\
4 \rho^2 \sigma_{0j}^2 ||\mathbf{\gamma}_{0sj}||^2,  \quad &\text{otherwise.} 
        \end{cases}
    \end{equation}
    Then, as $n_1, \dots, n_s, p \to \infty$, with probability at least $1- o(1)$,  
    \begin{equation}\label{eq:bvm}
    \begin{aligned}
        \sup_{x \in \mathbb R} \left| \tilde \Pi\left\{\frac{\sqrt{n} \left(\mathbf{\tilde \lambda}_j^\top \mathbf{\tilde \lambda}_{j'} - \mathbf{\mu}_{\lambda_j}^\top \mathbf{\mu}_{\lambda_{j'}}\right)}{l_{0, jj'}^2(\rho_{\Lambda})} \leq x\right\} - \Phi(x)\right| &\to 0, \\
        \sup_{x \in \mathbb R}
        \left| \tilde \Pi\left\{ \frac{\sqrt{n_s} \left(\mathbf{\tilde \gamma}_{sj}^\top \mathbf{\tilde \gamma}_{sj'} - \mathbf{\mu}_{\gamma_{sj}}^\top \mathbf{\mu}_{\gamma_{sj'}}\right)}{ 
        l_{0, sjj'}^2(\rho_{\Gamma_s})} \leq x\right\} - \Phi(x)\right| &\to 0, \quad (s=1, \dots, S),
    \end{aligned}\quad (1 \leq j, j' \leq p),
    \end{equation}
    where $\Phi(\cdot)$ denotes the cumulative distribution function of a standard Gaussian random variable.
\end{theorem}
Theorem \ref{thm:bvm_Lambda_outer} states that the induced distribution on the $i,j$-th elements of the matrices ${\mathbf{\tilde \Lambda}} {\mathbf{\tilde \Lambda}}$ and ${\mathbf{\tilde \Gamma}}_s {\mathbf{\tilde \Gamma}}_s^\top$, after appropriate centering by the point estimators, are asymptotically zero mean Gaussian distributions with variance $l_{0, jj'}^2(\rho_{\Lambda})$ and $l_{0, sjj'}^2(\rho_{\Gamma_s})$, where $l_{0, jj'}^2(\cdot)$ and $l_{0, sjj'}^2(\cdot)$ are defined in \eqref{eq:l_0_sq} and \eqref{eq:l_0_s_sq}. The next corollary provides an approximation of the credible intervals for elements of ${\mathbf{\Lambda}} {\mathbf{\Lambda}}^\top$ and ${\mathbf{\Gamma}}_s {\mathbf{\Gamma}}_s^\top$ when $p$ and the sample sizes are large. As a direct consequence of Theorem \ref{thm:bvm_Lambda_outer}, the asymptotic approximation to the $(1-\alpha)100\%$ equal-tail credible intervals from $\tilde \Pi$ for $\mathbf{\tilde \lambda}_j^\top \mathbf{\tilde \lambda}_{j'}$ and $\mathbf{\tilde \gamma}_{sj}^\top\mathbf{\tilde \gamma}_{sj'} $ are
    \begin{equation}
        \mathcal C_{jj'}(\rho_{\Lambda}) = \left[\mathbf{\mu}_{\lambda_j}^\top \mathbf{\mu}_{\lambda_{j'}}\pm z_{1- \alpha/2} \frac{l_{0jj'}(\rho_{\Lambda})}{\sqrt{n}}\right],
    \end{equation}
    and 
     \begin{equation}
        \mathcal C_{sjj'}(\rho_{\Gamma_s}) = \left[ \mathbf{\mu}_{\gamma_{sj}}^\top\mathbf{\mu}_{\gamma_{sj'}}\pm z_{1- \alpha/2} \frac{l_{0sjj'}(\rho_{\Gamma_s})}{\sqrt{n_s}}\right],
    \end{equation}
    respectively, where $z_{1-\alpha/2} = \Phi^{-1}(1- \alpha/2)$. 
Combining Theorems \ref{thm:clt_mu_Lambda_outer} and \ref{thm:bvm_Lambda_outer}, we can characterize the frequentist coverage of credible intervals from $\tilde \Pi$ as
\begin{equation}
\label{eq:coverage_ci}
    \begin{aligned}
         pr\left\{ \mathbf{\lambda}_{0j}^\top\mathbf{\lambda}_{0j'} \in \mathcal C_{jj'}(\rho_{\Lambda}) \right\} &= pr\left\{ \sqrt{n} \frac{ \left| \mathbf{\mu}_{\lambda_j}^\top \mathbf{\mu}_{\lambda_{j'} } -\mathbf{\lambda}_{0j}^\top \mathbf{\lambda}_{0j'} \right|
       }  {S_{0,jj}}  \leq z_{1-\alpha/2} \frac{l_{0jj'}(\rho_{\Lambda})}{S_{0,jj}} \right\}\\
       &\to q_{jj'}(\rho_{\Lambda}) = 2 \Phi\left\{z_{1-\alpha/2} \frac{l_{0jj'}(\rho_{\Lambda})}{S_{0,jj}}\right\}-1, \\
       pr\left\{ \mathbf{\gamma}_{0sj}^\top \mathbf{\gamma}_{0sj'} \in \mathcal C_{sjj'}(\rho_{\Gamma_s}) \right\} &= pr\left\{ \sqrt{n} \frac{ \left| \mathbf{\mu}_{\gamma_{sj}}^\top \mathbf{\mu}_{\gamma_{sj'} } \mathbf{\gamma}_{0sj}^\top \mathbf{\gamma}_{0sj'} \right|
       }  {S_{0,sjj}}  \leq z_{1-\alpha/2} \frac{l_{0sjj'}(\rho)}{S_{0,sjj}} \right\}\\
       &\to q_{sjj'}(\rho_{\Gamma_s}) = 2 \Phi\left\{z_{1-\alpha/2} \frac{l_{0sjj'}(\rho_{\Gamma_s})}{S_{0,sjj}}\right\}-1, 
    \end{aligned}     
    \end{equation}
    as $n_1, \dots, n_S, p \to \infty$.
Then, we can use \eqref{eq:coverage_ci} to tune the variance inflation terms. 
Let us define  \begin{equation*}\label{eq:b_0}
    b_{0jj'} = \begin{cases}
        \big( 1 +  \frac{||\mathbf{\lambda}_{0j}||^2 ||\mathbf{\lambda}_{0j'}||^2 + (\mathbf{\lambda}_{0j}^\top\mathbf{\lambda}_{0j'})^2 }{\sigma_{0j}^2 ||\mathbf{\lambda}_{0j'} ||^2 + \sigma_{0j'}^2||\mathbf{\lambda}_{0j} ||^2} \big)^{1/2}, \quad &\text{if } j \neq j'   \\   \big(1 + \frac{ |\mathbf{\lambda}_{0j} ||^2 }{ 2 \sigma_{0j}^2 }\big)^{1/2}, \quad &\text{otherwise,} 
              \end{cases}
\end{equation*}
and, for $s=1, \dots, S$, 
\begin{equation*}\label{eq:b_0s}
    b_{0sjj'} = \begin{cases}
    \begin{aligned}
        \big[ &1 +  (\sigma_{0j}^2 ||\mathbf{\gamma}_{0sj'}||^2 + \sigma_{0j'}^2||\mathbf{\gamma}_{0sj} ||^2)^{-1}\big(||\mathbf{\gamma}_{0sj}||^2 ||\mathbf{\gamma}_{0sj'}||^2 + (\mathbf{\gamma}_{0sj}^\top \mathbf{\gamma}_{0sj'})^2  \\ &+ ||\mathbf{\gamma}_{0sj} ||^2 ||\mathbf{\lambda}_{0j'}||^2 
      + ||\mathbf{\gamma}_{0sj} ||^2 ||\mathbf{\lambda}_{0j}||^2+  2\mathbf{\gamma}_{0sj}^\top\mathbf{\gamma}_{0sj'} \mathbf{\lambda}_{0j}^\top \mathbf{\lambda}_{0j'}\big) \big]^{1/2} 
    \end{aligned}
       , \quad &\text{if } j \neq j'   \\   \big(1 + \frac{ ||\mathbf{\gamma}_{0sj}||^2 +2||\mathbf{\lambda}_{0j}||^2 }{ 2\sigma_{0j}^2 }\big)^{1/2}, \quad &\text{otherwise,} 
              \end{cases}
\end{equation*}
and note that $\frac{l_{0jj'}(b_{0jj})}{S_{0,jj}}= 1$ and $\frac{l_{0sjj'}(b_{0sjj})}{S_{0,sjj}} =1$. Hence, setting
$\rho_{\Lambda} = b_{0jj'}$ and $\rho_{\Gamma_s} = b_{0sjj'}$, we have $ pr\left\{ \mathbf{\lambda}_{0j}^\top\mathbf{\lambda}_{0j'} \in \mathcal C_{jj'}(\rho_{\Lambda}) \right\} \to 1 - \alpha$ and $ pr\left\{ \mathbf{\gamma}_{0sj}^\top \mathbf{\gamma}_{0sj'} \in \mathcal C_{sjj'}(\rho_{\Gamma_s}) \right\} \to 1 - \alpha$. Clearly, this strategy is not feasible as values of the $b_{0jj'}$'s and $b_{0sjj'}$'s depend on the true parameters. However, we can replace them by consistent estimates. Indeed, if we set $\rho_{\Lambda} = b_{jj'}$ and $\rho_{\Gamma_s} = b_{sjj'}$, where $b_{jj'}$ and $b_{sjj'}$ are defined in \eqref{eq:b} and \eqref{eq:b_s}, we have 
$ pr\left\{ \mathbf{\lambda}_{0j}^\top\mathbf{\lambda}_{0j'} \in \mathcal C_{jj'}(\rho_{\Lambda}) \right\} \to 1 - \alpha$ and $ pr\left\{ \mathbf{\gamma}_{0sj}^\top \mathbf{\gamma}_{0sj'} \in \mathcal C_{sjj'}(\rho_{\Gamma_s}) \right\} \to 1 - \alpha$, since $\mathbf{\mu}_{\lambda_j}^\top \mathbf{\mu}_{\lambda_{j'}} \overset{pr}{\to} \mathbf{\lambda}_{0j}^\top \mathbf{\lambda}_{0j'}$, $\mathbf{\mu}_{\gamma_{sj}}^\top \mathbf{\mu}_{\gamma_{sj'}} \overset{pr}{\to} \mathbf{\gamma}_{0sj}^\top \mathbf{\gamma}_{0sj'}$ and $V_j \overset{pr}{\to} \sigma_{0j}^2$ for $j, j'=1, \dots, p$, by Lemma \ref{lemma:convergence} and \ref{lemma:convergence_Gamma_s} together with an application of Continuous Mapping Theorem \citep{billingsley}. Moreover, since $q_{jj'}(\cdot)$ and $q_{sjj'}(\cdot)$ are increasing functions, choosing $\rho_{\Lambda} = \max_{j,j'} b_{jj'}$ and $\rho_{\Gamma_s} = \max_{j,j'} b_{sjj'}$ guarantees entry-wise asymptotic valid coverage of credible intervals for ${\mathbf{\Lambda}} {\mathbf{\Lambda}}^\top$ and ${\mathbf{\Gamma}}_s {\mathbf{\Gamma}}_s^\top$ respectively. Alternatively, one could pick $\rho_{\Lambda}$ and $\rho_{\Gamma_s}$ by solving the non-linear equations 
\begin{equation*}
    \frac{1}{\binom{p}{2}} \sum_{1 \leq j \leq j' \leq p}q_{jj'}(\rho_{\Lambda}) = 1 - \alpha, \quad  \frac{1}{\binom{p}{2}} \sum_{1 \leq j \leq j' \leq p}q_{sjj'}(\rho_{\Gamma_s})= 1 - \alpha,
\end{equation*}
to obtain valid asymptotic coverage on average across elements of the low-rank components. In practice, we avoid solving those non-linear equations and approximate their solutions via the empirical mean of the $b_{jj'}$'s and $b_{sjj'}$'s as in \cite{fable}, and set $\rho_{\Lambda} = \binom{p}{2}^{-1} \sum_{1 \leq j \leq j' \leq p} b_{jj'}$ and $\rho_{\Gamma_s} =  \binom{p}{2}^{-1}\sum_{1 \leq j \leq j' \leq p}  b_{sjj'}$.

\section{Numerical experiments} \label{sec:numerical_experiments}
We illustrate the performance of our methodology in accuracy and quantification of uncertainty for ${\mathbf{\Lambda}}_0 {\mathbf{\Lambda}}_0^\top$ and the ${\mathbf{\Gamma}}_{0s}{\mathbf{\Gamma}}_{0s}^\top $'s, as well as computing time. 
We compare with two variational inference schemes, a stochastic variational inference algorithm (\texttt{SVI}) and a coordinate ascent algorithm (\texttt{CAVI}), of \citet{vi_bmsfa} {and the spectral estimator developed in \citet{Ando2017ClusteringHuge} (\texttt{SPECTRAL}), tailored to model \eqref{eq:fact_model}}. 
In the supplement, we consider a lower-dimensional example, where we also compare to the maximum likelihood estimate of \eqref{eq:fact_model} \citep{msfa} and a Bayesian estimate where posterior computation is carried out via a Gibbs sampler \citep{bmsfa}.
First, we consider an experiment where we simulate data from model \eqref{eq:fact_model}, generating the factor loadings as:
\begin{equation}\label{eq:generate_loadings} \big[  {\mathbf{\Lambda}_0} ~ {\mathbf{\Gamma}}_{01} ~  \cdots  ~ {\mathbf{\Gamma}}_{0S} \big] = L, \quad  [L']_{jl} \sim 0.5 \delta_0 + 0.5N(0, \sigma^2), \quad (j=1, \dots, p; l = 1, \dots, k_0 + {\textstyle \sum_s} q_s).  \end{equation}
We generate idiosyncratic variances from a uniform distribution supported on $[0.5, 5]$.  We consider the case where variances vary between outcomes but not between studies (homoscedastic), {while the heteroscedastic case where variances also vary across studies is considered in the supplement with minimal differences in conclusions.} 
{
Next, we consider an experiment where we generate loadings as follows
\begin{equation}
\begin{aligned}
&\mathbf{\Lambda}_0 = \Omega_0 + \mathbf{\bar \Lambda}_0\\
&\mathbf{ \Gamma}_{01} = \mathbf{\bar \Gamma}_{01}, \quad \mathbf{ \Gamma}_{0s} = \Omega_s + \mathbf{\bar \Gamma}_{0s}, \quad (s=2, \dots, S),\\
\end{aligned} \quad \text{where} \quad  \big[  {\mathbf{\bar \Lambda}_0} ~ {\mathbf{\bar \Gamma}}_{01} ~  \cdots  ~ {\mathbf{\bar \Gamma}}_{0S} \big] = L,
\end{equation}
with $L$ generated as in \eqref{eq:generate_loadings} and 
\begin{equation}
\begin{aligned}
 [ \Omega_0]_{jl} = [ \Omega_2]_{jl} = \dots = [ \Omega_S]_{jl}  &\sim N(0, 0.3^2), \quad (j=1, \dots, p; l=1,2), \quad\\
 [ \Omega_0]_{jl} &= 0, \quad (j=1, \dots, p; l=3, \dots,  k_0), \quad \\
 [ \Omega_s]_{jl} &= 0, \quad  (j=1, \dots, p; l=3, \dots,  q_s; s=2, \dots, S).
\end{aligned}
\end{equation}
}
{The first experiment targets a setting in which the $\mathbf \Gamma_{0s}$'s capture study-specific batch effects that are completely unrelated to a common underlying biological signal.
This is achieved by 
generating loading matrices independently of each other.
The second experiment captures a more realistic scenario 
with increased collinearity between $\mathbf{\Lambda_{0s}}$ and $\mathbf{\Gamma_{0s}}$ 
for some studies, 
making the separation of shared signal and study-specific variation more challenging.
This is achieved via a confounder $\Omega$ affecting a selected subset of the $\mathbf \Gamma_{0s}$'s, but not all. 
Incidentally, this allows the batch effects to bear some similarity to one another.
Both setups respect assumption \ref{assumption:li} almost surely.}

In all experiments, we take $S=5$ studies and choose study sample sizes and outcome dimensions as $(n_s, p) \in \mathbb (500, 1000) \times (500, 5000)$. For each configuration, we replicate the experiments 50 times. We evaluate estimation accuracy for ${\mathbf{\Lambda}}_0 {\mathbf{\Lambda}}_0^\top$ and ${\mathbf{\Gamma}}_{0s}{\mathbf{\Gamma}}_{0s}^\top $'s via the Frobenius norm of the difference of the estimate and true parameter rescaled by the norm of the true parameter and for $\mathbf{M}_s$ and $\mathbf{F}_s$ via the Procrustes or Frobenius error rescaled by the parameter size, i.e. for estimators $\widehat{\mathbf{\Lambda} \mathbf{\Lambda}^\top}$, $\widehat{{\mathbf{\Gamma}}_{s}{\mathbf{\Gamma}}_{s}^\top}$, $\mathbf{\hat M}_s$ and $\mathbf{\hat F}_s$, we compute
\begin{equation*}
\begin{aligned}\frac{||\widehat{\mathbf{\Lambda} \mathbf{\Lambda}^\top} -   {\mathbf{\Lambda}}_0 {\mathbf{\Lambda}}_0^\top||_F}{||{\mathbf{\Lambda}}_0 {\mathbf{\Lambda}}_0^\top||_F},& \quad \frac{||\widehat{{\mathbf{\Gamma}}_{s}{\mathbf{\Gamma}}_{s}^\top} - {\mathbf{\Gamma}}_{0s}{\mathbf{\Gamma}}_{0s}^\top||_F}{|| {\mathbf{\Gamma}}_{0s}{\mathbf{\Gamma}}_{0s}^\top ||_F},\\
    \min_{\mathbf{R} \in \mathbb R^{k_0 \times k_0} : \mathbf{R}^\top \mathbf{R} = \mathbf{I}_{k_0}} \frac{1}{\sqrt{n_s k_o}}||\mathbf{\hat M}_s \mathbf{R}_s- \mathbf{M}_{0s}||_F, &\quad   \min_{\mathbf{R} \in \mathbb R^{k_0 \times k_0} : \mathbf{R}^\top \mathbf{R} = \mathbf{I}_{k_0}} \frac{1}{\sqrt{n_s q_s}}||\mathbf{\hat F}_s \mathbf{R}_s- \mathbf{F}_{0s}||_F.
\end{aligned}
\end{equation*}

We evaluate uncertainty quantification via the average frequentist coverage of equal-tail $95\%$ credible intervals for randomly chosen $100 \times 100$ submatrices of covariance low-rank components. Additional details about the numerical studies are reported in the supplemental material. 

\begin{table}
\caption{Comparison of the methods in terms of estimation accuracy with independently generated loading matrices. Estimation errors have been multiplied by $10^2$. We report the mean and standard error over 50 replications.   \label{tab:accuracy_hom}}
\centering
{	\begin{tabular}{crrrr}
& \multicolumn{4}{c}{$p=500$, $n_s=500$}\\
Method & ${\mathbf{\Lambda}} {\mathbf{\Lambda}}^\top$ & $\mathbf{\Gamma}_s\mathbf{\Gamma}_s^\top$ & $\mathbf{M}_s$ & $\mathbf{F}_s$   \\
\hline
  \texttt{CAVI} & $60.74^{0.82}$ & $51.64^{0.70}$& $64.44^{0.91}$& $35.24^{0.83}$ \\ 
  \texttt{SVI} & $78.93^{0.15}$ & $65.09^{0.73}$& $80.35^{0.30}$& $62.88^{1.70}$ \\  
	 	 \texttt{SPECTRAL} & $17.68^{0.10}$ & $42.58^{0.31}$& $24.69^{0.13}$& $22.50^{0.11}$ \\ %
\texttt{BLAST} & $15.81^{0.07}$ & $37.12^{0.22}$& $23.54^{0.09}$& $22.49^{0.11}$ \\ 
   & \multicolumn{4}{c}{$p=500$, $n_s=1000$}\\
Method & ${\mathbf{\Lambda}} {\mathbf{\Lambda}}^\top$ & $\mathbf{\Gamma}_s\mathbf{\Gamma}_s^\top$ & $\mathbf{M}_s$ & $\mathbf{F}_s$   \\ 
\hline
	\texttt{CAVI} & $49.08^{0.40}$ & $39.39^{0.89}$& $53.27^{0.46}$& $30.62^{1.09}$ \\ 
    	\texttt{SVI} & $72.79^{0.17}$ & $53.98^{0.77}$& $75.88^{0.11}$& $48.54^{1.09}$ \\ 
	       \texttt{SPECTRAL} & $13.21^{0.07}$ & $34.90^{0.29}$& $24.47^{0.13}$& $21.82^{0.09}$ \\ %
\texttt{BLAST}& $11.75^{0.05}$ & $26.90^{0.17}$& $22.37^{0.08}$& $21.82^{0.09}$ \\ 
     & \multicolumn{4}{c}{$p=5000$, $n_s=500$}\\
Method & ${\mathbf{\Lambda}} {\mathbf{\Lambda}}^\top$ & $\mathbf{\Gamma}_s\mathbf{\Gamma}_s^\top$ & $\mathbf{M}_s$ & $\mathbf{F}_s$   \\
\hline
	\texttt{CAVI} & $73.61^{0.16}$ & $51.29^{0.42}$& $170.83^{1.81}$& $33.30^{0.41}$ \\ 
  		\texttt{SVI} &  $83.25^{0.05}$ & $58.65^{0.32}$& $118.31^{0.10}$& $43.30^{0.56}$ \\ 
	       \texttt{SPECTRAL} & $15.08^{0.05}$ & $33.44^{0.13}$& $8.74^{0.04}$& $8.67^{0.08}$ \\ %
\texttt{BLAST}&$14.75^{0.04}$ & $35.80^{0.17}$& $12.33^{0.11}$& $8.67^{0.08}$ \\ 
     & \multicolumn{4}{c}{$p=5000$, $n_s=1000$}\\
Method & ${\mathbf{\Lambda}} {\mathbf{\Lambda}}^\top$ & $\mathbf{\Gamma}_s\mathbf{\Gamma}_s^\top$ & $\mathbf{M}_s$ & $\mathbf{F}_s$   \\
\hline
	\texttt{CAVI} & $76.12^{0.10}$ & $39.19^{0.81}$& $152.43^{3.81}$& $24.60^{0.99}$ \\ 
  		\texttt{SVI} & $81.02^{0.25}$ & $45.34^{1.19}$& $157.85^{1.84}$& $31.05^{1.52}$ \\ 
	  \texttt{SPECTRAL} & $10.66^{0.03}$ & $24.15^{0.09}$& $8.24^{0.03}$& $7.74^{0.06}$ \\ %
      \texttt{BLAST}&$10.59^{0.08}$ & $26.16^{0.38}$& $14.15^{0.16}$& $7.81^{0.20}$\\ 
	\end{tabular}}
\end{table}

\begin{table}
\caption{Comparison of the methods in terms of estimation accuracy in the case with partially collinear loading matrices. Estimation errors have been multiplied by $10^2$. We report the mean and standard error over 50 replications.   \label{tab:accuracy_col}}
\centering
{	\begin{tabular}{crrrr}
& \multicolumn{4}{c}{$p=500$, $n_s=500$}\\
Method & ${\mathbf{\Lambda}} {\mathbf{\Lambda}}^\top$ & $\mathbf{\Gamma}_s\mathbf{\Gamma}_s^\top$ & $\mathbf{M}_s$ & $\mathbf{F}_s$   \\
\hline
  \texttt{CAVI} & $50.84^{0.87}$ & $57.46^{0.83}$& $62.00^{1.09}$ & $31.80^{0.98}$ \\ 
  \texttt{SVI} &  $70.39^{0.24}$ & $68.56^{0.59}$& $80.62^{0.34}$ & $60.51^{1.82}$ \\
	 	 \texttt{SPECTRAL} & $25.31^{0.25}$ & $63.74^{1.60}$& $31.01^{0.53}$ & $21.38^{0.12}$ \\
\texttt{BLAST} & $14.45^{0.06}$ & $33.76^{0.37}$& $23.00^{0.08}$ & $21.40^{0.12}$ \\
   & \multicolumn{4}{c}{$p=500$, $n_s=1000$}\\
Method & ${\mathbf{\Lambda}} {\mathbf{\Lambda}}^\top$ & $\mathbf{\Gamma}_s\mathbf{\Gamma}_s^\top$ & $\mathbf{M}_s$ & $\mathbf{F}_s$   \\ 
\hline
\texttt{CAVI} & $43.36^{0.62}$ & $48.34^{1.16}$& $54.02^{0.55}$ & $27.96^{1.08}$ \\ 
  \texttt{SVI} &  $64.78^{0.24}$ & $61.78^{0.71}$& $77.56^{0.17}$ & $46.71^{1.17}$ \\
	 	 \texttt{SPECTRAL} & $22.33^{0.22}$ & $59.86^{1.82}$& $30.85^{0.53}$ & $20.81^{0.11}$ \\
\texttt{BLAST} & $10.75^{0.05}$ & $24.38^{0.28}$& $21.75^{0.07}$ & $20.82^{0.11}$ \\
     & \multicolumn{4}{c}{$p=5000$, $n_s=500$}\\
Method & ${\mathbf{\Lambda}} {\mathbf{\Lambda}}^\top$ & $\mathbf{\Gamma}_s\mathbf{\Gamma}_s^\top$ & $\mathbf{M}_s$ & $\mathbf{F}_s$   \\
\hline
	\texttt{CAVI} & $63.15^{0.41}$ & $64.37^{0.94}$& $197.41^{3.31}$ & $31.12^{0.43}$ \\ 
  \texttt{SVI} &  $71.72^{0.09}$ & $67.90^{0.71}$& $121.09^{1.27}$ & $40.28^{0.71}$ \\
	 	 \texttt{SPECTRAL} & $22.36^{0.19}$ & $58.82^{1.80}$& $21.00^{0.81}$ & $8.46^{0.09}$ \\
\texttt{BLAST} & $13.59^{0.05}$ & $32.77^{0.33}$& $12.21^{0.11}$ & $8.48^{0.09}$ \\
     & \multicolumn{4}{c}{$p=5000$, $n_s=1000$}\\
Method & ${\mathbf{\Lambda}} {\mathbf{\Lambda}}^\top$ & $\mathbf{\Gamma}_s\mathbf{\Gamma}_s^\top$ & $\mathbf{M}_s$ & $\mathbf{F}_s$   \\
\hline
	\texttt{CAVI} & $64.80^{0.15}$ & $58.75^{1.44}$& $179.27^{3.34}$ & $22.31^{0.31}$ \\ 
  \texttt{SVI} & $69.41^{0.09}$ & $60.87^{1.18}$& $156.92^{0.59}$ & $27.89^{0.46}$ \\
	 	 \texttt{SPECTRAL} & $19.12^{0.16}$ & $54.71^{2.12}$& $20.89^{0.83}$ & $7.51^{0.06}$ \\
\texttt{BLAST} & $9.70^{0.04}$ & $23.64^{0.24}$& $9.79^{0.07}$ & $7.52^{0.06}$ \\
	\end{tabular}}
\end{table}

Tables \ref{tab:accuracy_hom} and \ref{tab:accuracy_col}  report a comparison in terms of estimation accuracy in the two experiments, respectively. 
{\texttt{BLAST} and \texttt{SPECTRAL} substantially outperform variational inference in estimating both low-rank variance components and latent factors. The comparison between \texttt{BLAST} and \texttt{SPECTRAL} is more nuanced. When loadings are generated independently (Table \ref{tab:accuracy_hom}), \texttt{BLAST} has a better estimation accuracy for covariance components, except for the study-specific ones when $p=5000$. Both methods have a comparable accuracy in estimating study-specific latent factors, while, for the shared ones, \texttt{BLAST} (\texttt{SPECTRAL}) has the best performance when $p=500$ ($p=5000$, respectively).
When loading matrices share more similarities (Table \ref{tab:accuracy_col}), \texttt{BLAST} has a much better performance than \texttt{SPECTRAL} in all the settings considered.  }

Tables \ref{tab:uq_hom} and \ref{tab:uq_col} provide strong support for our methodology in terms of providing well-calibrated credible intervals, while alternatives suffer from severe undercoverage.

Table \ref{tab:time_hom} reports a comparison in terms of running time in the first example, where \texttt{BLAST} is faster alternatives in almost all the scenarios considered.

\begin{table}
\caption{Comparison of the methods in terms of frequentist coverage of $95\%$ credible intervals with independently generated loading matrices. Coverage levels have been multiplied by $10^2$. We report the mean and standard error over 50 replications.   \label{tab:uq_hom}}
\centering
{	\begin{tabular}{crrrrr}
& \multicolumn{5}{c}{$p=500$}\\
& \multicolumn{2}{c}{$n_s=500$} & & \multicolumn{2}{c}{$n_s=1000$}\\
Method & ${\mathbf{\Lambda}} {\mathbf{\Lambda}}^\top$ & $\mathbf{\Gamma}_s\mathbf{\Gamma}_s^\top$ & & ${\mathbf{\Lambda}} {\mathbf{\Lambda}}^\top$ & $\mathbf{\Gamma}_s\mathbf{\Gamma}_s^\top$   \\
\hline
	\texttt{CAVI} & $27.76^{0.63}$ & $73.89^{0.60}$& &$26.75^{0.72}$ & $69.93^{0.91}$ \\ 
  		\texttt{SVI} &  $18.33^{0.28}$ & $62.92^{0.68}$& &$16.06^{0.25}$ & $56.88^{0.68}$ \\ 
	 \texttt{BLAST}& $92.52^{0.13}$ & $94.96^{0.10}$& &$92.06^{0.15}$ & $94.95^{0.08}$ \\ 

& \multicolumn{5}{c}{$p=5000$}\\
& \multicolumn{2}{c}{$n_s=500$} & & \multicolumn{2}{c}{$n_s=1000$}\\
Method & ${\mathbf{\Lambda}} {\mathbf{\Lambda}}^\top$ & $\mathbf{\Gamma}_s\mathbf{\Gamma}_s^\top$ & & ${\mathbf{\Lambda}} {\mathbf{\Lambda}}^\top$ & $\mathbf{\Gamma}_s\mathbf{\Gamma}_s^\top$   \\
\hline
	\texttt{CAVI} & $22.39^{0.30}$ & $73.52^{0.48}$& &$15.73^{0.83}$ & $70.65^{1.90}$ \\ 
  		\texttt{SVI} & $18.22^{0.26}$ & $70.97^{0.48}$& &$14.33^{0.76}$ & $66.36^{1.48}$ \\ 
	 \texttt{BLAST}&$94.09^{0.14}$ & $94.33^{0.09}$& &$93.71^{0.59}$ & $94.17^{0.33}$ \\ 
     
	\end{tabular}}
\end{table}

\begin{table}
\caption{Comparison of the methods in terms of frequentist coverage of $95\%$ credible intervals in the case with partially collinear loading matrices. Coverage level have been multiplied by $10^2$. We report the mean and standard error over 50 replications.   \label{tab:uq_col}}
\centering
{	\begin{tabular}{crrrrr}
& \multicolumn{5}{c}{$p=500$}\\
& \multicolumn{2}{c}{$n_s=500$} & & \multicolumn{2}{c}{$n_s=1000$}\\
Method & ${\mathbf{\Lambda}} {\mathbf{\Lambda}}^\top$ & $\mathbf{\Gamma}_s\mathbf{\Gamma}_s^\top$ & & ${\mathbf{\Lambda}} {\mathbf{\Lambda}}^\top$ & $\mathbf{\Gamma}_s\mathbf{\Gamma}_s^\top$   \\
\hline
	\texttt{CAVI} & $34.26^{0.70}$ & $59.50^{1.22}$ & & $31.05^{1.09}$ & $53.18^{1.60}$ \\ 
  		\texttt{SVI} & $24.94^{0.24}$ & $50.18^{0.98}$ & & $20.51^{0.21}$ & $42.66^{1.05}$  \\ 
	 \texttt{BLAST}& $92.95^{0.14}$ & $94.42^{0.11}$ & & $91.93^{0.17}$ & $94.50^{0.09}$ \\ 
& \multicolumn{5}{c}{$p=5000$}\\
& \multicolumn{2}{c}{$n_s=500$} & & \multicolumn{2}{c}{$n_s=1000$}\\
Method & ${\mathbf{\Lambda}} {\mathbf{\Lambda}}^\top$ & $\mathbf{\Gamma}_s\mathbf{\Gamma}_s^\top$ & & ${\mathbf{\Lambda}} {\mathbf{\Lambda}}^\top$ & $\mathbf{\Gamma}_s\mathbf{\Gamma}_s^\top$   \\
\hline
	\texttt{CAVI} & $28.59^{0.38}$ & $52.85^{1.51}$ & & $21.52^{0.34}$ & $45.81^{1.85}$ \\ 
  		\texttt{SVI} & $25.93^{0.32}$ & $52.32^{1.39}$ & & $20.76^{0.32}$ & $44.66^{1.65}$  \\ 
	 \texttt{BLAST}& $94.14^{0.13}$ & $94.03^{0.11}$ & & $94.17^{0.13}$ & $94.03^{0.09}$ \\
	\end{tabular}}
\end{table}

\begin{table}
\caption{Comparison of the methods in terms of running time in seconds. We report the mean and standard error over 50 replications.   \label{tab:time_hom}}
\centering
{	\begin{tabular}{crrrrr}
& \multicolumn{5}{c}{Time (s)}\\
& \multicolumn{2}{c}{$p=500$} & & \multicolumn{2}{c}{$p=5000$}\\
& $n_s = 500$ & $n_s = 1000$ & & $n_s = 500$ & $n_s = 1000$ \\
\hline
	\texttt{CAVI} & $923^{81}$ & $4293^{300}$ & &  $3768^{194}$ & $6053^{397}$ \\ 
  		\texttt{SVI} & $185^{28}$ & $358^{74}$ & &  $2491^{52}$ & $6591^{164}$ \\ 
              \texttt{SPECTRAL} & $49^{3}$ & $89^{6}$ & &  $2165^{32}$ & $3872^{151}$ \\ 
	 \texttt{BLAST} & $31^{1}$ & $92^{5}$ & &  $925^{28}$ & $2416^{37}$ \\ 
	\end{tabular}}
\end{table}

\section{Application}\label{sec:application}

We consider data sets from three studies that analyze gene expression among immune cells. Estimation of gene dependencies is a fundamental task in the development of cancer treatments \citep{tan}. Two studies are part of the
ImmGen project \citep{yoshida}. One is is the GSE109125 bulkRNAseq dataset, which contains data from 103 immunocyte populations \citep{yoshida}, while the other is the GSE37448  microarray dataset \citep{elpek}. 
Finally, the third study is the GSE15907 microarray dataset \citep{painter, desch}, which measures multiple ex vivo immune lineages, primarily from adult B6 male mice. 
The study sample sizes are 156, 628 and 146, respectively.

After preprocessing (described in Section \ref{sec:additional_details_application} in the supplementary material), we consider two experiments that retain 2846 and 7870 genes corresponding to the intersection of $25\%$ and $50\%$ of the genes with the largest variance in each study, respectively. Previous analyses of these data sets focused on a much smaller number of genes for computational feasibility \citep{sufa}. The computational efficiency of our procedure allows us to scale the analysis to thousands of gene while maintaining computational feasibility.

 To estimate latent dimensions, we apply the procedure described in Section \ref{subsec:estimation_k} to data sets with $p=2846$ genes, which estimates the shared latent dimension as $15$ and study-specific latent dimensions as $17$, $54$ and $21$, respectively. {In both experiments, \texttt{CAVI} runs into numerical issues and is omitted from the analysis.} 

For each experiment, we test the out-of-sample performance of competing methodologies. We randomly leave out $20\%$ of samples from each study and fit the model on the remaining $80\%$. We randomly divide the outcomes into two halves, and, for each left-out sample, we predict the first half of outcomes via their posterior conditional expected value given the second half. For each gene, we compute the mean squared error normalized by its empirical variance. We also compute the log-likelihood of the entire test-set. We report the out-of-sample accuracy in Table \ref{tab:oos_application}.
{For both scenarios, \texttt{SPECTRAL} has the lowest mean squared error with \texttt{BLAST} being a closed second. 
Table \ref{tab:coverage_application} reports the mean coverage of $95\%$ predictive intervals in the out-of-sample prediction task. Unsurprisingly, all three methods have the lowest average coverage for the GSE109125 study, where they also achieve the highest mean squared error. Overall, \texttt{BLAST} and \texttt{SVI} provide more accurate coverage than \texttt{SPECTRAL}. 
Finally, we also compute the log-likelihood in the test set using point estimates for the overall covariance matrix of each method and tested whether the out-of-sample log-likelihood of the \texttt{BLAST} estimate is larger than that of other methodologies by means of one-sided paired t-tests.  The results are reported in Table \ref{tab:log_lik_application}, with \texttt{BLAST} producing excellent out-of-sample results that outperform competitors in almost all cases. 
Additional plots and details on the analysis are reported in Section \ref{sec:additional_details_application} in the supplementary material.}

\begin{table}
\caption{Comparison of the methods in terms of out of sample accuracy. We report mean, $1^{st}$ and $3^{rd}$ quantiles for the mean squared error normalized by each gene empirical variance. \label{tab:oos_application}}
\centering
{	\begin{tabular}{crrrrrrrr}
& \multicolumn{8}{c}{$p=2846$}\\
& \multicolumn{2}{c}{ GSE15907}  & & \multicolumn{2}{c}{ GSE37448} & &  \multicolumn{2}{c}{GSE109125}\\
Method & Mean & ($Q_1$, $Q_3$) & & Mean & ($Q_1$, $Q_3$) & & Mean & ($Q_1$, $Q_3$)\\
\hline
  		\texttt{SVI} & 0.47& (0.31, 0.60)& &0.50 & (0.36, 0.62) & & 0.58& (0.41, 0.72)\\ 
      \texttt{SPECTRAL}& 0.16& (0.10, 0.20)& &0.16 & (0.10, 0.19) & & 0.27& (0.16. 0.32)\\ 
       \texttt{BLAST} & 0.18& (0.12, 0.23)& &0.17 & (0.12, 0.21) & & 0.28& (0.18, 0.33)\\ 
& \multicolumn{8}{c}{$p=7870$}\\
& \multicolumn{2}{c}{ GSE15907}  & & \multicolumn{2}{c}{ GSE37448} & &  \multicolumn{2}{c}{GSE109125}\\
Method & Mean & ($Q_1$, $Q_3$) & & Mean & ($Q_1$, $Q_3$) & & Mean & ($Q_1$, $Q_3$)\\
\hline
  		\texttt{SVI} & 0.51 & (0.34, 0.65)& &0.51 & (0.35, 0.64) & & 0.70& (0.52. 0.85)\\ 
      \texttt{SPECTRAL} & 0.18 & (0.11, 0.23)& & 0.17 & (0.10, 0.22) & & 0.35 & (0.21, 0.44)\\ 
       \texttt{BLAST} & 0.20 & (0.13, 0.26)& &0.18 & (0.11, 0.23) & & 0.36& (0.23, 0.45)\\ 
	\end{tabular}}
\end{table}

\begin{table}
\caption{Comparison of the methods in terms of mean coverage of $95\%$ predictive intervals.   \label{tab:coverage_application}}
\centering
{	\begin{tabular}{crrrrr}
& \multicolumn{5}{c}{$p=2846$}\\
Method & { GSE15907}  & &{ GSE37448} & &  {GSE109125}\\
\hline
  		\texttt{SVI}& 91.40\% & & 93.70\% & & 87.58\% \\ 
      \texttt{SPECTRAL}& 88.03\% & & 89.27\% & & 78.80\% \\  
       \texttt{BLAST} & 92.88\% & & 93.70\% & & 86.00\% \\  
& \multicolumn{5}{c}{$p=7870$}\\
Method & { GSE15907}  & &{ GSE37448} & &  {GSE109125}\\
\hline
  		\texttt{SVI}& 89.95\% & & 91.30\% & & 88.96\% \\ 
      \texttt{SPECTRAL}& 88.45\% & & 90.37\% & & 75.52\% \\  
       \texttt{BLAST} & 91.87\% & & 92.92\% & & 80.59\% \\  
	\end{tabular}}
\end{table}

\begin{table}
\caption{Comparison of the methods in terms of out-of-sample log-likelihood. We report the out-of-sample log-likelihood of different methods and p-values of the paired one-sided t-tests for a greater log-likelihood for \texttt{BLAST} estimate. \label{tab:log_lik_application}}
\centering
{	\begin{tabular}{crrrrrrrr}
& \multicolumn{8}{c}{$p=2846$}\\
& \multicolumn{2}{c}{ GSE15907}  & & \multicolumn{2}{c}{ GSE37448} & &  \multicolumn{2}{c}{GSE109125}\\
Method & log-lik & p-value & & log-lik & p-value & & log-lik & p-value\\
\hline
  		\texttt{SVI} & -71719.4 & $3.2 \times 10^{-15}$ &  & -330156.6& $1.2 \times 10^{-33}$  & &-90466.8 & $2.8 \times 10^{-4} $\\ 
      \texttt{SPECTRAL} & -38099.9 & 0.001&  &-165255.8 & 0.001& &-76092.9 & $1.6\times 10^{-7} $\\\ 
       \texttt{BLAST} & -35891.8 &  &  & -158336.1 & & &-71577.9  \\
& \multicolumn{8}{c}{$p=7870$}\\
& \multicolumn{2}{c}{ GSE15907}  & & \multicolumn{2}{c}{ GSE37448} & &  \multicolumn{2}{c}{GSE109125}\\
Method & log-lik & p-value & & log-lik & p-value & & log-lik & p-value\\
\hline
  		\texttt{SVI} & -240736.6 & $1.17 \times 10^{-14} $ &  &-1033932.0 &  $1.35 \times 10^{-40} $ & &-304317.1 & $0.24$\\ 
         \texttt{SPECTRAL} & -129383.6   & 0.020 &  & -501053.1& 0.011& & -308244.7& 0.001 \\ 
       \texttt{BLAST} & -124192.4 &  &  & -493709.2 &  & & -292967.4&  \\ 
	\end{tabular}}
\end{table}

\section{Discussion}\label{sec:conclusion}

We presented a method for computationally efficient inference with state-of-the-art accuracy in estimation and uncertainty quantification for high-dimensional multi-study factor analysis. {Considering the promising results in this paper, using \texttt{BLAST} to integrate high-dimensional data across studies or groups may yield improved results in various applied contexts. Some possibilities include obtaining more robust estimates of gene associations through integrating multiple genomic datasets or improving inferences on within- and between-study covariance in clinical data on patients.}  Additionally, several important directions for future work emerge from this article. Firstly, generalizing the method from the case of linear, additive, and Gaussian factor models to accommodate non-linear and non-Gaussian structure would be of substantial interest \citep{gplvm, zito2024compressive}. {Moreover, while we have been motivated by applications in which study subjects are reasonably assumed to be independent, generalizations to accommodate a variety of dependence structures including time series \citep{bai15, Ando2016PanelData, Ando2017ClusteringHuge}, spatially indexed data, and networks \citep{xie_eigen} would be useful.}

Other interesting directions include developing a supervised variant of this framework, targeting latent factors that are predictive of an outcome variable of interest \citep{hahn, ferrari}, or extensions to more complex and hierarchical scenarios, where the factor analytical component is included in a larger model. Deriving the methodologies along with their theoretical guarantees would be an important contribution to many applied tasks.

\bigskip

\section*{Acknowledgement}
This research was partially supported by the National Institutes of Health (grant ID R01ES035625), by the Office of Naval Research (Grant N00014-24-1-2626), and by the European Research Council under the European Union’s Horizon 2020 research and innovation programme (grant agreement No 856506).

\clearpage

\bibliographystyle{agsm}
\bibliography{references}

\clearpage

\newpage
\appendix

\appendix
\newpage
\begin{center}
{\large\bf SUPPLEMENTARY MATERIAL}
\end{center}

\section{Proofs of the theoretical results}
\label{sec:theory_proof}

\subsection{Preliminary results}

\begin{proposition}\label{prop:V_s_outer}
Let $\mathbf{V}_{0s} \in \mathbb R^{p \times k_s}$ and $\mathbf{V}_s \in \mathbb R^{p \times k_s} $ be the matrix of right singular vectors of $\mathbf{M}_{0s} {\mathbf{\Lambda}}_0^\top + \mathbf{F}_s {\mathbf{\Gamma}}_{0s}^\top $ and $\mathbf{Y}_s$ associated to their leading $k_s$ singular values. Then, under Assumption \ref{assumption:model}--\ref{assumption:sigma}, we have
\begin{equation*}
    pr\left\{\left|\left|\mathbf{V}_s \mathbf{V}_s^\top -\mathbf{V}_{0s}\mathbf{V}_{0s}^\top \right|\right| > G_1 (C_{sr, n_s} C_{sc, p})^2 \left(\frac{1}{n_s} + \frac{1}{p}\right)\right\} \to 0, \quad (n_s, p \uparrow \infty),
\end{equation*} 
for some finite constant $G_1 < \infty$. Replacing Assumptions \ref{assumption:distributions} and \ref{assumption:sigma} with Assumptions \ref{ass:bai-mean}--\ref{ass:sv_lf},
 with probability at least $1-o(1)$, we have 
\begin{equation*}
    pr\left\{\left|\left|\mathbf{V}_s \mathbf{V}_s^\top -\mathbf{V}_{0s}\mathbf{V}_{0s}^\top \right|\right| > G_1 \left(\frac{1}{\sqrt{n_s}} + \frac{1}{\sqrt{p}}\right)\right\} \to 0, \quad (n_s, p \uparrow \infty),
\end{equation*} 
 
\end{proposition}
\begin{proof}[Proof of Proposition \ref{prop:V_s_outer}]
The proof is related to the one for Proposition 3.5 in \citet{fable}.
Let $ \mathbf{\tilde M}_{0s} = [\mathbf{M}_{0s}~ \mathbf{F}_{0s}]$ and $\Tilde {\mathbf{\Lambda}}_{0s} = [{\mathbf{\Lambda}}_0~ {\mathbf{\Gamma}}_{0s}]$ and denote by $\mathbf{X}_{0s} = \mathbf{M}_{0s} {\mathbf{\Lambda}}_0^\top + \mathbf{F}_{0s} {\mathbf{\Gamma}}_{0s}^\top  =  \mathbf{\tilde M}_{0s} {\mathbf{\tilde \Lambda}}_{0s}^\top$ the true signal for the $s$-th study. In particular,
\begin{equation*}
\begin{aligned}
    s_{k_s}^2( \mathbf{\tilde M}_{0s}) s_{k_s}^2({\mathbf{\tilde \Lambda}}_s) ||\mathbf{V}_s^\perp \mathbf{V}_s^{\perp\top} -\mathbf{V}_{0s}^\perp \mathbf{V}_{0s}^{\perp\top}|| & \leq ||\mathbf{X}_{0s}^\top(\mathbf{V}_s^\perp \mathbf{V}_s^{\perp\top} -\mathbf{V}_{0s}^\perp \mathbf{V}_{0s}^{\perp\top}) \mathbf{X}_{0s}|| \\
    & \leq ||\mathbf{X}_{0s}^\top \mathbf{V}_s^\perp \mathbf{V}_s^{\perp\top}\mathbf{X}_{0s}- \mathbf{X}_{0s}^\top \mathbf{X}_{0s} ||\\
    &=  || (\mathbf{I}_n - \mathbf{V}_s^\perp \mathbf{V}_s^{\perp\top}) \mathbf{X}_{0s}||^2 \leq 4 ||\mathbf{E}_s||^2, 
\end{aligned} 
\end{equation*}
where the last inequality follows from Theorem 2 in \citet{luo_et_al}. Using Corollary 5.35 of \citet{vershynin_12}, we have $s_{k_s} ( \mathbf{\tilde M}_{0s}) \asymp \sqrt{n_s}$with probability at least $1-o(1)$, and Assumptions \ref{assumption:Lambda} and \ref{assumption:Gammas} imply $s_{k_s}({\mathbf{\tilde \Lambda}}_{0s}) \asymp \sqrt{p}$. The probabilistic bound on $||\mathbf{E}_s|| $ completes the proof. 
\end{proof}

\begin{proposition}\label{prop:recovery_P_0}
Let $\mathbf{\bar P}$ be the orthogonal projection onto the subspace spanned by the leading $k_0$ singular vectors of $\mathbf{\tilde P}$, where $\mathbf{\tilde P}$ is defined as in \eqref{eq:P_tilde}, and $\mathbf{V}_0 \in \mathbb R^{p \times k_0}$ be the matrix of left singular vectors of ${\mathbf{\Lambda}}$. Then, under Assumption \ref{assumption:model}--\ref{assumption:sigma}, we have
    \begin{equation*}
         pr\left\{\left|\left|\mathbf{\bar P} - \mathbf{V}_0\mathbf{V}_0^\top \right|\right| > G_2(C_{r, \{n_s\}_{s}^{S}} C_{c, p})^2 \left(\frac{1}{n_{\min}} + \frac{1}{p}\right)\right\} \to 0 \quad (n_1, \dots, n_S, p \uparrow \infty),
    \end{equation*}
where $G_2 < \infty$, and $n_{\min} = \min_{s=1, \dots, S}n_s$. Replacing Assumptions \ref{assumption:distributions} and \ref{assumption:sigma} with Assumptions \ref{ass:bai-mean}--\ref{ass:sv_lf},
 with probability at least $1-o(1)$, we have 
 \begin{equation*}
         pr\left\{\left|\left|\mathbf{\bar P} - \mathbf{V}_0\mathbf{V}_0^\top \right|\right| > G_2 \left(\frac{1}{\sqrt{n_{\min}}} + \frac{1}{\sqrt p}\right)\right\} \to 0 \quad (n_1, \dots, n_S, p \uparrow \infty),
    \end{equation*}
\end{proposition}
\begin{proof}[Proof of Proposition \ref{prop:recovery_P_0}]

Let $\mathbf{P}_0 = \mathbf{V}_0 \mathbf{V}_0^\top$, $\mathbf{A} = \frac{1}{S} \sum_{s=1}^S \mathbf{V}_s^\perp \mathbf{V}_s^{\perp \top}$, 
$\mathbf{\tilde E} = \mathbf{\tilde P}  - (\mathbf{P}_0 + \mathbf A)$, and $\epsilon$ be the maximum over $s \in \{1, \dots, S\}$ of the right hand side of the inequality in Proposition \ref{prop:V_s_outer}, and consider the large probability set where $\left|\left|\mathbf{V}_s \mathbf{V}_s^\top -\mathbf{V}_{0s}\mathbf{V}_{0s}^\top \right|\right| < C_1 {\epsilon}$ holds for any $s=1, \dots, S$ for some constant $C_1$, which has probability at least $1 - o(1)$ by Proposition \ref{prop:V_s_outer}. In particular, $\mathbf{\tilde P} = \mathbf{P}_0 + \mathbf{A} + \mathbf{\tilde E}$, where $\mathbf{P}_0$ is a $k_0$ rank matrix with unit singular values, $\mathbf{A}$ is a $\sum_s q_s$ rank matrix with singular values upper bounded by $1 - \delta$ by Assumption \ref{assumption:sv_A}, and  $|| \mathbf{\tilde E}|| \leq C_2 {\epsilon}$. 
As $n_{\min}$ and $p$ diverge, ${\epsilon} < \delta/(2C_2)$ eventually. Thus, for $n_{\min}$ and $p$ sufficiently large, we can apply Proposition 2 from \citet{vu_et_al_21}, which characterizes the matrix left singular vectors associated to the $k_0$ leading singular values of $\mathbf{\tilde P}$ (up to an orthogonal transformation) as 
\begin{equation}\label{eq:left_sv_tilde_P}
    \mathbf{\bar V} = (\mathbf{V}_0 -  \mathbf{V}_\perp  \mathbf{R})(\mathbf{I}_{k_0} +  \mathbf{R}^\top  \mathbf{R})^{-1},
\end{equation}
where $ \mathbf{V}_\perp$ is the matrix of singular vectors of $\mathbf{A}$ and $ \mathbf{R} \in \mathbb R^{\sum_{s} q_s \times k_0}$ such that $|| \mathbf{R}|| \asymp || \mathbf{\tilde E}||\lesssim {\epsilon}$. From \eqref{eq:left_sv_tilde_P}, we can get an expression for $\mathbf{\bar P} = \mathbf{\bar V} \mathbf{\bar V}^\top$ and applying the Woodbury identity, we obtain
\begin{equation*}
    ||\mathbf{\bar P} - \mathbf{P}_0|| \lesssim {\epsilon}.
\end{equation*}
\end{proof}

\begin{proposition}\label{prop:U_s_outer}
Denote by $\mathbf{U}_s^\perp \in \mathbb R^{n \times q_s}$ and $\mathbf{U}_{0s}^\perp \in \mathbb R^{n \times q_s}$ the matrices of left singular vectors of $\mathbf{Y}_s \mathbf{\bar Q} $ and $\mathbf{F}_s {\mathbf{\Gamma}}_s$, respectively, associated to their leading $q_s$ singular values.
If Assumption \ref{assumption:model}--\ref{assumption:sigma} hold and $n_{\min} \gtrsim \sqrt{n_s}$ for each $s=1, \dots, S$, where $n_{\min} = \min_{s=1, \dots, S}n_s$, then
\begin{equation*}
   pr\left\{\left|\left| \mathbf{U}_s^\perp \mathbf{U}_s^{\perp\top} -\mathbf{U}_{0s}^\perp \mathbf{U}_{0s}^{\perp\top} \right|\right| > G_3 (C_{sr, n_s} C_{sc, p})^2 \left(\frac{1}{n_{s}} + \frac{1}{p}\right)\right\} = 0\quad (n_1, \dots n_S, p \uparrow \infty),
\end{equation*}
where $G_3 < \infty$ is a finite constant. 
Replacing Assumptions \ref{assumption:distributions} and \ref{assumption:sigma} with Assumptions \ref{ass:bai-mean}--\ref{ass:sv_lf},
 with probability at least $1-o(1)$, we have
\begin{equation*}
   pr\left\{\left|\left| \mathbf{U}_s^\perp \mathbf{U}_s^{\perp\top} -\mathbf{U}_{0s}^\perp \mathbf{U}_{0s}^{\perp\top} \right|\right| > G_3 \left(\frac{1}{\sqrt{n_{s}}} + \frac{1}{\sqrt{p}}\right)\right\} = 0\quad (n_1, \dots n_S, p \uparrow \infty),
\end{equation*}
 
\end{proposition}
\begin{proof}[Proof of Proposition \ref{prop:U_s_outer}]
Recall  $ \mathbf{Y}_s \mathbf{\bar Q} = \mathbf{Y}_s (\mathbf{I}_p- \mathbf{V}_0 \mathbf{V}_0^\top) + \mathbf{Y}_s(\mathbf{V}_0\mathbf{V}_0^\top - \mathbf{\bar P})$. Next, we apply the same steps as in the proof for Proposition \ref{prop:V_s_outer}. 
In particular, 
\begin{equation*}
   s_{q_s}^2(\mathbf{F}_{0s}) s_{q_s}^2((\mathbf{I}_p - P_{{\mathbf{\Lambda}}}){\mathbf{\Gamma}}_s) ||\mathbf{U}_s^\perp \mathbf{U}_s^{\perp\top} -\mathbf{U}_{0s}^\perp \mathbf{U}_{0s}^{\perp\top}|| \leq 
   4 ||\mathbf{E}_s(\mathbf{I}_p- \mathbf{V}_0 \mathbf{V}_0^\top) + \mathbf{Y}_s(\mathbf{V}_0\mathbf{V}_0^\top - \mathbf{\bar P})||^2.
\end{equation*}
    Note that 
    \begin{equation*}
        \begin{aligned}
             || \mathbf{E}_s(\mathbf{I}_p- \mathbf{V}_0 \mathbf{V}_0^\top)|| &\leq|| \mathbf{E}_s||,\\
             ||\mathbf{Y}_s(\mathbf{V}_0\mathbf{V}_0^\top - \mathbf{\bar P})|| &\leq ||\mathbf{Y}_s|| ||\mathbf{V}_0\mathbf{V}_0^\top - \mathbf{\bar P}||,
        \end{aligned}
    \end{equation*}
    $s_{q_s}\left((\mathbf{I}_p - P_{{\mathbf{\Lambda}}}){\mathbf{\Gamma}}_s\right) \asymp \sqrt{p}$ by Assumption \ref{assumption:Gammas}, and $n_s = \mathcal O(n_{\min}^2)$, along with $s_{q_s}(\mathbf{F}_{0s}) \asymp \sqrt{n_s}$ with probability at least $1-o(1)$ by Corollary 5.35 of \citet{vershynin_12} under Assumption \ref{assumption:model} or by Assumption \ref{ass:lf_projection}. The probabilistic bounds on $|| \mathbf{E}_s||$ and $||\mathbf{Y}_s||$ in Lemmas \ref{lemma:E} and \ref{lemma:Y_s} prove the result. 
\end{proof}

\begin{proposition}\label{prop:recovery_M}
Denote by $\mathbf{U}^c \in \mathbb R^{n \times k_0}$ and $\mathbf{U}_0^c \in \mathbb R^{n \times k_0}$ the matrices of left singular vectors of $\mathbf{\hat Y}^c$ and $\mathbf{M}_0 {\mathbf{\Lambda}}_0^\top$, respectively, associated to the leading $k_0$ singular values.
 Then, if Assumptions \ref{assumption:model}--\ref{assumption:sigma} hold and  $n_{\min} \gtrsim \sqrt{n_s}$ for each $s=1, \dots, S$, where $n_{\min} = \min_{s=1, \dots, S}n_s$, we have
         \begin{equation*}
    pr\left\{\left|\left|\mathbf{U}^c \mathbf{U}^{c \top} -  \mathbf{U}_0^c \mathbf{U}_0^{c \top} \right|\right| > G_4 (C_{r, \{n_s\}_{s}^{S}} C_{c, p})^2 \left(\frac{1}{n} + \frac{1}{p}\right)\right\} \to 0 \quad (n_1, \dots, n_S, p \uparrow \infty),
\end{equation*}
where $G_4 < \infty$ is a finite constant. Moreover, Replacing Assumptions \ref{assumption:distributions} and \ref{assumption:sigma} with Assumptions \ref{ass:bai-mean}--\ref{ass:sv_lf},
 with probability at least $1-o(1)$, we have 
\begin{equation*}
    pr\left\{\left|\left|\mathbf{U}^c \mathbf{U}^{c \top} -  \mathbf{U}_0^c \mathbf{U}_0^{c \top} \right|\right| > G_4  \left(\frac{1}{\sqrt{n}} + \frac{1}{\sqrt{p}}\right)\right\} \to 0 \quad (n_1, \dots, n_S, p \uparrow \infty).
\end{equation*}
 
\end{proposition}
\begin{proof}[Proof of Proposition \ref{prop:recovery_M}]
We let 
\begin{equation}\label{eq:N_0_perp}
    \mathbf{N}_0^\perp= \begin{bmatrix}
        \mathbf{U}_{01}^\perp  & 0 & &  &  0 \\
        0 &  \mathbf{U}_{02}^\perp &   0  & \cdots&  0 \\
        \vdots & & & & \vdots\\
         0  & \cdots &&  0  &\mathbf{U}_{0S}^\perp 
    \end{bmatrix}
\end{equation}
and
\begin{equation}\label{eq:P_0_perp}
    \mathbf{P}_0^\perp = \mathbf{N}_0^\perp \mathbf{N}_0^{\perp \top};
\end{equation}
that is $ \mathbf{P}_0^\perp$ is a block diagonal matrix and the $s$-th block is given by $ \mathbf{U}_{0s}^\perp \mathbf{U}_{0s}^{\perp \top}$.
Note that $\mathbf{P}_0^{\perp 2} = \mathbf{P}_0^\perp$ and $\mathbf{P}_0^{\perp\top} = \mathbf{P}_0^\perp$, since all the blocks are orthogonal projection matrices, making $\mathbf{P}_0$ an orthogonal projection matrix
Similarly, we define $ \mathbf{Q}_0^{\perp}= \mathbf{I}_n - \mathbf{P}_0^\perp$, and $\Delta$ to be a block diagonal matrix and the $s$-th block is given $\mathbf{U}_{0s}^\perp \mathbf{U}_{0s}^{\perp \top}- \mathbf{U}_s^\perp \mathbf{U}_s^{\perp \top}$ for $s=1, \dots, S$,
\begin{equation}\label{eq:Delta}
    \Delta =  \begin{bmatrix}
        \mathbf{U}_{01}^\perp \mathbf{U}_{01}^{\perp \top}- \mathbf{U}_1^\perp \mathbf{U}_1^{\perp \top} & 0 & &  &  0 \\
        0 &  \mathbf{U}_{02}^\perp \mathbf{U}_{02}^{\perp \top}- \mathbf{U}_2^\perp \mathbf{U}_2^{\perp \top} &   0  & \cdots&  0 \\
        \vdots & & & & \vdots\\
         0  & \cdots &&  0  &\mathbf{U}_{0S}^\perp \mathbf{U}_{0S}^{\perp \top}- \mathbf{U}_S^\perp \mathbf{U}_S^{\perp \top}
    \end{bmatrix}.
\end{equation}

Since $\mathbf{\hat Y}^c =  \mathbf{M}_0{\mathbf{\Lambda}}_0^\top -\mathbf{P}_0^\perp \mathbf{M}_0{\mathbf{\Lambda}}_0^\top  +\mathbf{Q}_0^{\perp}   \mathbf{E} + \mathbf{\Delta  Y} $, we have
    \begin{equation*}
        \left|\left|\left( \mathbf{I}_n - \mathbf{U}^c \mathbf{U}^{c \top} \right)   \mathbf{M}_0{\mathbf{\Lambda}}_0^\top \right|\right|^2 \leq 4  \left|\left|-\mathbf{P}_0^\perp \mathbf{M}_0{\mathbf{\Lambda}}_0^\top +\mathbf{Q}_0^{\perp}  \mathbf E + \mathbf{\Delta  Y}  \right|\right|^2   
    \end{equation*}
by Theorem 2 in \citet{luo_et_al}. Moreover, we have 
$||\mathbf{P}_0^\perp \mathbf{M}_0{\mathbf{\Lambda}}_0^\top || \leq   ||\mathbf{P}_0^\perp \mathbf{M}_0|| ||{\mathbf{\Lambda}}_0|| \asymp \sqrt{p}$, with probability at least $1-o(1)$, since $||\mathbf{P}_0^\perp \mathbf{M}_0||\asymp 1$,  with probability at least $1-o(1)$, by Lemma \ref{lemma:Mt_P_M} under Assumptions \ref{assumption:model}--\ref{assumption:sigma}, or by Assumption \ref{ass:lf_projection}, and $||{\mathbf{\Lambda}}_0|| \asymp \sqrt{p}$ by Assumption \ref{assumption:Lambda}. The probabilistic bounds in Lemmas \ref{lemma:E} and \ref{lemma:Delta_y_j} with the fact that $s_{k_0}( \mathbf{M}_0) \asymp \sqrt{n}$  with probability at least $1 - o(1)$, by Corollary 5.35 of \citet{vershynin_12} under Assumption \ref{assumption:distributions}, or by Assumption \ref{ass:lf_projection}, and $s_{k_0}({\mathbf{\Lambda}}_0)  \asymp \sqrt{p}$ by Assumption \ref{assumption:Lambda} complete the proof. 
\end{proof}

\subsection{Proofs of the main results}

\begin{proof}[Proof of Theorem \ref{thm:factors_procrustes_error}]
    Let $\mathbf{U}_0 \mathbf{D}_0 \mathbf{V}_0^\top $ be the singular value decomposition of $\mathbf{M}_0$. Note that $\mathbf{U}_0 \mathbf{U}_0^\top = \mathbf{U}_0^c \mathbf{U}_0^{c\top} $ where $\mathbf{U}_0^c \in \mathbb R^{n \times k_0}$ is the matrix of left singular vectors of $\mathbf{M}_0 {\mathbf{\Lambda}}_0^\top $. Recall that, under Assumptions \ref{assumption:model}--\ref{assumption:sigma}, by Proposition \ref{prop:recovery_M}, we have $||\mathbf{U}^c \mathbf{U}^c - \mathbf{U}_0 \mathbf{U}_0^\top ||  \lesssim (C_{r, \{n_s\}_{s}^{S}} C_{c, p})^2 \big(\frac{1}{n} + \frac{1}{p}\big)$, with probability at least $1 - o(1)$.
    Then, by Davis-Kahan theorem \citep{davis_kahan} we have
    $\min_{\mathbf{R}: \mathbf{R}^\top \mathbf{R} = \mathbf{I}_{k_0}}||\mathbf{U}^c  - \mathbf{U}_0 \mathbf{R} || = ||\mathbf{U}^c  - \mathbf{U}_0\mathbf{\hat R} || \lesssim(C_{r, \{n_s\}_{s}^{S}} C_{c, p})^2 \big(\frac{1}{n} + \frac{1}{p}\big)$, with probability at least $1 - o(1)$, where $\mathbf{\hat R}$ is the orthogonal matrix achieving the minimum of the quantity on the left hand side. Recalling that $\mathbf{\hat M} = \sqrt{n} \mathbf{U}^c$ and letting $\mathbf{\tilde R} = \mathbf{V}_0\mathbf{ \hat R}$, we have \begin{equation*}
        \begin{aligned}
            ||\mathbf{\hat M} - \mathbf{M}_0 \mathbf{\tilde R}  || &= ||\sqrt{n} \mathbf{U}^c - \mathbf{U}_0 \mathbf{D}_0 \mathbf{\hat R} || \leq  ||\sqrt{n} (\mathbf{U}^c - \mathbf{U}_0\mathbf{\hat R} ) || + ||\sqrt{n} \mathbf{U}_0\mathbf{\hat R} - \mathbf{U}_0 \mathbf{D}_0\mathbf{ \hat R}  ||\\
            &\leq \sqrt{n}(C_{r, \{n_s\}_{s}^{S}} C_{c, p})^2 \big(\frac{1}{n} + \frac{1}{p}\big) + \max_{1 \leq l \leq k_0} |\sqrt{n} - d_{0l}|
        \end{aligned}
    \end{equation*}where $d_{0l}$ is the $l$-th largest singular value of $\mathbf{D}_0$. Moreover, by corollary 5.35 of \citet{vershynin_12}, we have $|d_{0l} - \sqrt{n}| \lesssim \sqrt{k_0}$ with probability at least $1-o(1)$. We conclude the proof of the first result by noting that $||\mathbf{\hat M}_s - \mathbf{M}_{0s} \mathbf{\tilde R}  || \leq ||\mathbf{\hat M} - \mathbf{M}_0 \mathbf{\tilde R}  ||$ and $\mathbf{\tilde R}^\top \mathbf{\tilde R} = \mathbf{I}_{k_0}$. The proof for the second result proceeds by applying Proposition \ref{prop:U_s_outer} followed by similar steps.  %
    \end{proof}

\begin{proof}[Proof of Theorem \ref{thm:posterior_contraction_Lambda_outer}]
First, we show $\mathbf{\mu}_\Lambda \mathbf{\mu}_\Lambda^\top$ is consistent for ${\mathbf{\Lambda}}_0 {\mathbf{\Lambda}}_0^\top$.
We have 
\begin{equation*}
    \begin{aligned}
          \mathbf{\mu}_\Lambda \mathbf{\mu}_\Lambda^\top &= \frac{n}{(n + \tau_{\Lambda}^{-2})^2} \left\{\mathbf{\hat Y}^{c\top} \mathbf{U}_0^c  \mathbf{U}_0^{c\top}\mathbf{\hat Y}^c +\mathbf{\hat Y}^{c\top} \left(\mathbf{U}^c  \mathbf{U}^{c\top} -  \mathbf{U}_0^c  \mathbf{U}_0^{c\top} \right)\mathbf{\hat Y}^c\right\}.
    \end{aligned}
\end{equation*}

Recall from the proof of Proposition \ref{prop:recovery_M} that $\mathbf{\hat Y}^c =  \mathbf{M}_0{\mathbf{\Lambda}}_0^\top -\mathbf{P}_0^\perp \mathbf{M}_0{\mathbf{\Lambda}}_0^\top  +\mathbf{Q}_0^{\perp} \mathbf E + \mathbf{\Delta  Y} $, where $\mathbf{P}_0^\perp$ and $\Delta$ are defined in \eqref{eq:P_0_perp} and \eqref{eq:Delta}, respectively, and $\mathbf{Q}_0^\perp = \mathbf{I}_p - \mathbf{P
}_0^\perp$. Then
\begin{equation*}
    \begin{aligned}
       \mathbf{\hat Y}^{c \top}\mathbf{U}_0^c  \mathbf{U}_0^{c\top}\mathbf{\hat Y}^c =&    n{\mathbf{\Lambda}}_0 {\mathbf{\Lambda}}_0^\top + {\mathbf{\Lambda}}_0 \left( \mathbf{M}_0^\top  \mathbf{M}_0 - n \mathbf{I}_{k_0} \right) {\mathbf{\Lambda}}_0^\top    \\   
        & +3 {\mathbf{\Lambda}}_0 \mathbf{M}_0^\top  \mathbf{P}_0^\perp  \mathbf{M}_0 {\mathbf{\Lambda}}_0^\top \\   
        & + 
        {\mathbf{\Lambda}}_0 \mathbf{M}_0^\top \mathbf{Q}_0^{\perp}   \mathbf{E} +   \mathbf{E}^\top \mathbf{Q}_0^{\perp}\mathbf{M}_0  {\mathbf{\Lambda}}_0^\top\\   
        & + 
        {\mathbf{\Lambda}}_0 \mathbf{M}_0^\top  \mathbf{\Delta  Y}  +  \mathbf{Y}^\top \mathbf{\Delta}^\top  \mathbf{M}_0{\mathbf{\Lambda}}_0^\top 
        \\   
        & +  {\mathbf{\Lambda}}_0 \mathbf{M}_0^\top  \mathbf{P}_0^\perp \mathbf{U}_0^c  \mathbf{U}_0^{c\top}\mathbf{Q}_0^{\perp}   \mathbf{E} +   \mathbf{E}^\top \mathbf{Q}_0^{\perp} \mathbf{U}_0^c  \mathbf{U}_0^{c\top}  \mathbf{P}_0^\perp \mathbf{M}_0  {\mathbf{\Lambda}}_0^\top  \\ 
        &+ {\mathbf{\Lambda}}_0 \mathbf{M}_0^\top  \mathbf{P}_0^\perp \mathbf{U}_0^c  \mathbf{U}_0^{c\top} \mathbf{\Delta  Y}  +  \mathbf{Y}^\top \mathbf{\Delta}^\top \mathbf{U}_0^c  \mathbf{U}_0^{c\top}  \mathbf{P}_0^\perp  \mathbf{M}_0{\mathbf{\Lambda}}_0^\top 
          \\     & +  \mathbf{E}^\top \mathbf{Q}_0^{\perp} \mathbf{U}_0^c  \mathbf{U}_0^{c\top}\mathbf{Q}_0^{\perp}   \mathbf{E} +   \mathbf{Y}^\top \mathbf{\Delta}^\top \mathbf{U}_0^c  \mathbf{U}_0^{c\top}   \mathbf{\Delta  Y}  
          \\    &
+  \mathbf{E}^\top \mathbf{Q}_0^{\perp} \mathbf{U}_0^c  \mathbf{U}_0^{c\top}   \mathbf{\Delta  Y}   +   \mathbf{Y}^\top \mathbf{\Delta}^\top \mathbf{U}_0^c  \mathbf{U}_0^{c\top} \mathbf{Q}_0^{\perp}   \mathbf{E}  . 
    \end{aligned}
\end{equation*}

We can bound each term, using the fact that, with probability at least $1- o(1)$, 
\begin{enumerate}
   \item $||\mathbf{M}_0|| \lesssim \sqrt{n}$, by Theorem 4.6.1 of \citet{vershynin2018hdp},
   \item $||   \mathbf{M}_0^\top  \mathbf{M}_0 - n \mathbf{I}_{k_0}  || \lesssim (\sqrt{n} + \sqrt{k}) \sqrt{\log n}$ by Lemma E.1 in \citet{fable},
   \item $|| \mathbf{U}_0^{\perp \top}  \mathbf{M}_0 || \lesssim \sqrt{k_0} + \sqrt{\sum_s q_s}$ by  Theorem 4.6.1 of \citet{vershynin2018hdp}, since $ \mathbf{U}_0^{\perp \top}  \mathbf{M}_0  \sim MN_{\sum_s q_s, k_0} (0,  \mathbf{I}_{\sum_s q_s}, \mathbf{I}_{k_0})$, where $ MN_{d_1, d_2}(\mu,  \mathbf{\Sigma}_r, \mathbf{\Sigma}_c)$ denotes a $d_1 \times d_2$ matrix normal distribution with mean $\mu$, within-column covariance $\mathbf{\Sigma}_c$ and within-row covariance $\mathbf{\Sigma}_r$,
    \item $||\mathbf{E}||\lesssim   C_{r, \{n_s\}_{s}^{S}} C_{c, p}  (\sqrt{n} + \sqrt{p})$, by Lemma \ref{lemma:E},
    \item $||\mathbf{\Delta Y} || \lesssim (C_{r, \{n_s\}_{s}^{S}} C_{c, p}   )^2 \big(\frac{\sqrt{p}}{\sqrt{n_{\min}}} + \frac{\sqrt{n_{\max}}}{\sqrt{p}}\big)$, by Lemma \ref{lemma:Delta_y_j},
    \end{enumerate}
    and $||{\mathbf{\Lambda}}_0|| \asymp \sqrt{p}$ by Assumption \ref{assumption:Lambda}, as well as, $||\mathbf{P}_0^\perp|| = ||\mathbf{Q}_0^\perp|| = ||\mathbf{U}_0^c|| = ||\mathbf{N}_0^\perp|| =1$,  where  $\mathbf{N}_0^\perp$ is defined in  \eqref{eq:N_0_perp}.
Note $\frac{n}{(n + \tau_{\Lambda}^{-2})} = \frac{1}{n} + \nu_n$ where $\nu_n \asymp \frac{1}{n^2}$ and $\nu_n ||{\mathbf{\Lambda}}_0|| \asymp \frac{p}{n^2}$.
Combining all of the above, we obtain
\begin{equation*}
    \begin{aligned}
        ||\mathbf{\mu}_\Lambda \mathbf{\mu}_\Lambda^\top  - {\mathbf{\Lambda}}_0 {\mathbf{\Lambda}}_0^\top || \lesssim (\sqrt{\log n} + C_{r, \{n_s\}_{s}^{S}} C_{c, p}) \frac{p}{\sqrt{n}}  + C_{r, \{n_s\}_{s}^{S}} C_{c, p}\sqrt{p}
    \end{aligned}
\end{equation*}
with probability at least $1 - o(1)$.
Finally, since $||{\mathbf{\Lambda}}_0|| \asymp \sqrt{p}$, with probability at least $1 - o(1)$,  
\begin{equation*}
\label{eq:consistency_mu_Lambda_outer}
    \begin{aligned}
                \frac{||\mathbf{\mu}_\Lambda \mathbf{\mu}_\Lambda^\top  - {\mathbf{\Lambda}}_0 {\mathbf{\Lambda}}_0^\top ||}{||{\mathbf{\Lambda}}_0 {\mathbf{\Lambda}}_0^\top ||} &\lesssim  (\sqrt{\log n} + C_{r, \{n_s\}_{s}^{S}} C_{c, p}) \frac{1}{\sqrt{n}}  + C_{r, \{n_s\}_{s}^{S}} C_{c, p}\frac{1}{\sqrt{p}}.
    \end{aligned}
\end{equation*}
{By bounding the size of the difference of each sample from the mean, we should also obtain posterior concentration. 
A sample ${\mathbf{\tilde \Lambda}}$ from $\tilde \Pi$ is given by
\begin{equation*}
   {\mathbf{\tilde \Lambda}} = \mathbf{\mu}_\Lambda + \mathbf{\tilde E}_{\Lambda}, \quad  \mathbf{\tilde E}_{\Lambda} = [ \mathbf{e}_{\lambda_1} ~ \cdots ~\mathbf{e}_{\lambda_p} ]^\top, \quad 
     \mathbf{e}_{\lambda_{j}} \sim N_{k_0}\left(0, \frac{\rho_{\Lambda}^2 \tilde \sigma_j^2}{n + \tau_{\Lambda}^{-2}} \mathbf{I}_{k_0}\right).
\end{equation*}
Thus, 
\begin{equation*}
\begin{aligned}
     \left|\left| {\mathbf{\tilde \Lambda}}  {\mathbf{\tilde \Lambda}} ^\top - \mathbf{\mu}_\Lambda \mathbf{\mu}_\Lambda^\top \right|\right| &\lesssim \left|\left| \mathbf{\tilde E}_{\Lambda} \mathbf{\tilde E}_{\Lambda}^\top \right|\right| + \left|\left| \mathbf{\mu}_\Lambda \mathbf{\tilde E}_{\Lambda}^\top \right|\right| \\
     &\lesssim  \frac{p}{n} \rho_{\Lambda}^2 \max_{j} \tilde \sigma_j^2 + \left( \frac{(np)^{1/2} + n^{1/2} + p^{1/2}}{n^{1/2}}\right) \frac{p^{1/2}}{n^{1/2}} \rho_{\Lambda}^2 \max_{j} \tilde \sigma_j^2 \\
  & \lesssim \frac{p}{n^{1/2}},
\end{aligned}
\end{equation*}
with probability at least $1 - o(1)$, since $\tilde \sigma_j^2  \lesssim 1$ with probability at least $1-o(1)$, by Lemma \ref{lemma:sigma_j_tilde}, determining the desired result.}

Next, we proceed by establishing consistency for ${\mathbf{\Gamma}}_s {\mathbf{\Gamma}}_s^\top$. 
As a first step, we bound  $||\mathbf{\hat M}_s \mathbf{\mu}_\Lambda^\top - \mathbf{M}_{0s} {\mathbf{\Lambda}}_{0}^\top||$. 
Recall that $\mathbf{\hat M} \mathbf{\mu}_\Lambda^\top = \frac{n}{n + \tau_{\Lambda}^{-2}} \mathbf{U}^c  \mathbf{U}^{c \top}\mathbf{\hat Y}^c$ and $\mathbf{ \hat M_s} \mathbf{\mu}_\Lambda^\top = \frac{n}{n + \tau_{\Lambda}^{-2}} \mathbf{U}_s^{c}  \mathbf{U}^{c \top}\mathbf{\hat Y}^c$, where $\mathbf{U}_s^{c} \in \mathbb R^{n_s \times k_0}$ is the block of $\mathbf{U}^c$ corresponding to the $s$-th study, that is $\mathbf{U}^c = \big[ \mathbf{U}_1^{c\top} ~ \dots ~ \mathbf{U}_S^{c\top}\big]^\top$.
We have the following decomposition
\begin{equation*}
    \begin{aligned}
       \mathbf{\hat M}_s \mathbf{\mu}_\Lambda^\top =  \frac{n}{n + \tau_{\Lambda}^{-2}}\left\{\mathbf{M}_{0s} {\mathbf{\Lambda}}_0^\top - \mathbf{U}_{0s}^c  \mathbf{U}_0^{c \top}\left( \mathbf{P}_0^\perp \mathbf{M}_0{\mathbf{\Lambda}}_0^\top +  \mathbf{Q}_0^{\perp} \mathbf E + \mathbf{\Delta  Y} \right) + \left(\mathbf{U}_{0s}^c \mathbf{U}_0^{c \top} - \mathbf{U}_s^{c}  \mathbf{U}^{c \top}\right)\mathbf{\hat Y}^c  \right\}. 
    \end{aligned}
\end{equation*}
Using the fact that, with probability at least $1- o(1)$, 
\begin{enumerate}
   \item $||\mathbf{M}_0|| \lesssim \sqrt{n}$, $||\mathbf{M}_{0s}|| \asymp \sqrt{n_s}$, and $||\mathbf{F}_{0s}|| \lesssim \sqrt{n_s}$ by Theorem 4.6.1 of \citet{vershynin2018hdp},
   \item $|| \mathbf P_0^{\perp }  \mathbf{M}_0 || \lesssim \sqrt{k_0} + \sqrt{\sum_s q_s}$ by Lemma 5,
    \item $||\mathbf{E}||\lesssim   C_{r, \{n_s\}_{s}^{S}} C_{c, p}  (\sqrt{n} + \sqrt{p})$ by Lemma \ref{lemma:E},
   \item $|| \mathbf{U}_{0s}^c  \mathbf{U}_0^{c \top} - \mathbf{U}^c  \mathbf{U}^{c \top} || \leq  ||  \mathbf{U}_0^c  \mathbf{U}_0^{c \top} - \mathbf{U}^c  \mathbf{U}^{c \top}|| \lesssim (C_{r, \{n_s\}_{s}^{S}} C_{c, p}   )^2 \big(\frac{1}{n} + \frac{1}{p}\big)$, with probability at least $1- o(1)$ by Proposition \ref{prop:recovery_M},
    \end{enumerate}
we obtain $  ||\mathbf{ \hat M_s} \mathbf{\mu}_\Lambda^\top  - \mathbf{M}_{0s} {\mathbf{\Lambda}}_0^\top|| \leq ||\mathbf{\hat M} \mathbf{\mu}_\Lambda^\top  -  \mathbf{M}_0{\mathbf{\Lambda}}_0^\top||\lesssim  (C_{r, \{n_s\}_{s}^{S}} C_{c, p}   )^2 (\sqrt{n} + \sqrt{p})$, with probability at least $1-o(1)$.

Hence, with probability at least $1-o(1)$, for a sample ${\mathbf{\tilde \Lambda}}$ from $\tilde \Pi$,  we have
\begin{equation*}
\begin{aligned}
   ||\mathbf{\hat M}_s\tilde{{\mathbf{\Lambda}}}^\top  -  \mathbf{M}_{0s}{\mathbf{\Lambda}}_0^\top|| \leq  ||\mathbf{\hat M} \mathbf{\mu}_\Lambda^\top  -  \mathbf{M}_0{\mathbf{\Lambda}}_0^\top|| + ||\mathbf{\hat M}\left( \mathbf{\mu}_\Lambda - {\mathbf{\tilde \Lambda}}\right)^\top || \lesssim (C_{r, \{n_s\}_{s}^{S}} C_{c, p}   )^2 (\sqrt{n} + \sqrt{p}),
\end{aligned}  
\end{equation*}
since
\begin{equation*}
    ||\mathbf{\hat M}\left( \mathbf{\mu}_\Lambda - {\mathbf{\tilde \Lambda}}\right)^\top || \leq ||\mathbf{\hat M} || ||  \mathbf{\mu}_\Lambda - {\mathbf{\tilde \Lambda}}|| \lesssim \sqrt{n} \frac{\sqrt{p}}{\sqrt{n}} \rho_{\Lambda} \max_{1 \leq j \leq p} \tilde \sigma_j \lesssim \sqrt{p},
\end{equation*}
and $ \max_{1 \leq j \leq p} \tilde \sigma_j  \lesssim 1$ by Lemma \ref{lemma:sigma_j_tilde} with probability $1-o(1)$, 
which implies $ ||\mathbf{\hat M}_s\tilde{{\mathbf{\Lambda}}}^\top  - \mathbf{M}_{0s} {\mathbf{\Lambda}}_0^\top|| \lesssim (C_{r, \{n_s\}_{s}^{S}} C_{c, p}   )^2 (\sqrt{n} + \sqrt{p})$, with probability at least $1-o(1)$, 
Hence, consider the outer product of the conditional mean of $\mathbf{\Gamma}_s$, $ \mathbf{\bar \Gamma}_s = E[   {\mathbf{\Gamma}}_s \mid \mathbf{Y}, \mathbf{\hat  F}_s,\mathbf{\hat M}_s, \mathbf{\tilde \Lambda}, \mathbf{\tilde \Sigma}, \rho_{\Lambda} ] $, given $\mathbf{ \hat M}_s$, $\mathbf{\tilde \Lambda}$ and $\mathbf{\tilde \Sigma}$, where $(\mathbf{\tilde \Lambda}, \mathbf{\tilde \Sigma})$ is a sample from $\tilde \Pi$,
\begin{equation*}
\begin{aligned}
        \bar {\mathbf{\Gamma}}_s \bar {\mathbf{\Gamma}}_s^\top &= \frac{n_s}{(n_s + \tau_{\Gamma_s}^{-2})^2} (\mathbf{Y}_s -\mathbf{\hat M}_s {\mathbf{\tilde \Lambda}}^\top)^\top \mathbf{U}_s^{\perp} \mathbf{U}_s^{\perp \top}(\mathbf{Y}_s -\mathbf{\hat M}_s {\mathbf{\tilde \Lambda}}^\top)\\
        &= \frac{n_s}{(n_s + \tau_{\Gamma_s}^{-2})^2} (\mathbf{Y}_s -\mathbf{\hat M}_s {\mathbf{\tilde \Lambda}}^\top)^\top \mathbf{U}_{0s}^\perp \mathbf{U}_{0s}^{\perp \top}(\mathbf{Y}_s -\mathbf{\hat M}_s {\mathbf{\tilde \Lambda}}^\top) \\
        &\quad + \frac{n_s}{(n_s + \tau_{\Gamma_s}^{-2})^2} (\mathbf{Y}_s -\mathbf{\hat M}_s {\mathbf{\tilde \Lambda}}^\top)^\top ( \mathbf{U}_s^{\perp} \mathbf{U}_s^{\perp \top} - \mathbf{U}_{0s}^\perp \mathbf{U}_{0s}^{\perp \top})(\mathbf{Y}_s -\mathbf{\hat M}_s {\mathbf{\tilde \Lambda}}^\top).
\end{aligned}
\end{equation*}
We analyze each term separately. 
Since $ \mathbf{Y}_s -\mathbf{\hat M}_s {\mathbf{\tilde \Lambda}}^\top = \mathbf{Y}_s - \mathbf{M}_{0s}  {\mathbf{\Lambda}}_0^\top + ( \mathbf{M}_{0s}  {\mathbf{\Lambda}}_0^\top -\mathbf{\hat M}_s \mathbf{\mu}_\Lambda^\top  +\mathbf{\hat M}_s \mathbf{\mu}_\Lambda^\top -\mathbf{\hat M}_s {\mathbf{\tilde \Lambda}}^\top  ) = \mathbf{F}_{0s} {\mathbf{\Gamma}}_{0s}^\top + \mathbf{E}_s +  ( \mathbf{M}_{0s}  {\mathbf{\Lambda}}_0^\top -\mathbf{\hat M}_s \mathbf{\mu}_\Lambda^\top ) + (\mathbf{ \hat M_s} \mathbf{\mu}_\Lambda^\top -\mathbf{\hat M}_s {\mathbf{\tilde \Lambda}}^\top  ) $. 
The first term can be decomposed as follows
\begin{equation*}
    \begin{aligned}
        \frac{n_s}{(n_s + \tau_{\Gamma_s}^{-2})^2} &(\mathbf{Y}_s -\mathbf{\hat M}_s {\mathbf{\tilde \Lambda}}^\top)^\top \mathbf{U}_{0s}^\perp \mathbf{U}_{0s}^{\perp \top}(\mathbf{Y}_s -\mathbf{\hat M}_s {\mathbf{\tilde \Lambda}}^\top) = \\
        &=  \frac{n_s}{(n_s + \tau_{\Gamma_s}^{-2})^2} {\mathbf{\Gamma}}_{0s} \mathbf{F}_{0s}^\top  \mathbf{F}_{0s}{\mathbf{\Gamma}}_{0s}^\top  +  \frac{n_s}{(n_s + \tau_{\Gamma_s}^{-2})^2} \left({\mathbf{\Gamma}}_{0s} \mathbf{F}_{0s}^\top \mathbf{E}_s + \mathbf{E}_s^\top \mathbf{F}_{0s}{\mathbf{\Gamma}}_{0s}^\top \right)\\
        & \quad + \frac{n_s}{(n_s + \tau_{\Gamma_s}^{-2})^2} \left\{{\mathbf{\Gamma}}_{0s} \mathbf{F}_{0s}^\top\left( \mathbf{M}_{0s}  {\mathbf{\Lambda}}_0^\top -\mathbf{\hat M}_s {\mathbf{\tilde \Lambda}}^\top\right)  + \left( \mathbf{M}_{0s}  {\mathbf{\Lambda}}_0^\top -\mathbf{\hat M}_s {\mathbf{\tilde \Lambda}}^\top\right)^\top \mathbf{F}_{0s}{\mathbf{\Gamma}}_{0s}^\top \right\}\\
         & \quad +\frac{n_s}{(n_s + \tau_{\Gamma_s}^{-2})^2} \left[ \left\{\mathbf{E}_s + \left( \mathbf{M}_{0s}  {\mathbf{\Lambda}}_0^\top -\mathbf{\hat M}_s {\mathbf{\tilde \Lambda}}^\top\right)\right\}^\top \mathbf{U}_{0s}^\perp \mathbf{U}_{0s}^{\perp \top} \left\{\mathbf{E}_s + \left( \mathbf{M}_{0s}  {\mathbf{\Lambda}}_0^\top -\mathbf{\hat M}_s {\mathbf{\tilde \Lambda}}^\top\right)\right\} \right].
    \end{aligned}
\end{equation*}
\begin{enumerate}
\item First, we decompose $\frac{n_s}{(n_s + \tau_{\Gamma_s}^{-2})^2} {\mathbf{\Gamma}}_{0s} \mathbf{F}_{0s}^\top  \mathbf{F}_{0s}{\mathbf{\Gamma}}_{0s}^\top = \frac{1}{n_s}{\mathbf{\Gamma}}_{0s} \mathbf{F}_{0s}^\top  \mathbf{F}_{0s}{\mathbf{\Gamma}}_{0s}^\top + \big(\frac{n_s}{(n_s + \tau_{\Gamma_s}^{-2})^2} -\frac{1}{n_s}\big){\mathbf{\Gamma}}_{0s} \mathbf{F}_{0s}^\top  \mathbf{F}_{0s}{\mathbf{\Gamma}}_{0s}^\top$, where $\big(\frac{n_s}{(n_s + \tau_{\Gamma_s}^{-2})^2} -\frac{1}{n_s}\big) \asymp \frac{1}{n_s^2}$. Then, note
\begin{equation*}
    || \big(\frac{n_s}{(n_s + \tau_{\Gamma_s}^{-2})^2} -\frac{1}{n_s}\big){\mathbf{\Gamma}}_{0s} \mathbf{F}_{0s}^\top  \mathbf{F}_{0s}{\mathbf{\Gamma}}_{0s}^\top  || \lesssim \frac{p}{n_s},
\end{equation*}
and
\begin{equation*}
    \frac{1}{n_s}{\mathbf{\Gamma}}_{0s} \mathbf{F}_{0s}^\top  \mathbf{F}_{0s}{\mathbf{\Gamma}}_{0s}^\top = {\mathbf{\Gamma}}_{0s} {\mathbf{\Gamma}}_{0s}^\top + {\mathbf{\Gamma}}_{0s} \big(\frac{1}{n_s}  \mathbf{F}_{0s}^\top  \mathbf{F}_{0s} - \mathbf{I}_{q_s}\big){\mathbf{\Gamma}}_{0s}^\top,
\end{equation*}
    with $|| {\mathbf{\Gamma}}_{0s} \big(\frac{1}{n_s}  \mathbf{F}_{0s}^\top  \mathbf{F}_{0s} - \mathbf{I}_{q_s}\big){\mathbf{\Gamma}}_{0s}^\top|| \leq ||{\mathbf{\Gamma}}_{0s}||^2 ||\frac{1}{n_s}  \mathbf{F}_{0s}^\top  \mathbf{F}_{0s} - \mathbf{I}_{q_s}|| \lesssim \frac{p}{\sqrt{n_s}}\sqrt{\log{n_s}}$,
    with probability at least $1-o(1)$, since 
    $||\frac{1}{n_s}  \mathbf{F}_{0s}^\top  \mathbf{F}_{0s} - \mathbf{I}_{q_s}|| \lesssim \frac{\sqrt{\log{n_s}}}{n_s}$ with probability at least $1-o(1)$ by Corollary E.1 of \citet{fable}.
    \item With probability at least $1-o(1)$, we have 
    \begin{equation*}
        \begin{aligned}
             \frac{n_s}{(n_s + \tau_{\Gamma_s}^{-2})^2} ||{\mathbf{\Gamma}}_{0s} \mathbf{F}_{0s}^\top \mathbf{E}_s + \mathbf{E}_s^\top \mathbf{F}_{0s}{\mathbf{\Gamma}}_{0s}^\top|| &\lesssim 
             \frac{1}{n_s} ||{\mathbf{\Gamma}}_{0s}|| ||\mathbf{F}_{0s} ||  || \mathbf{E}_s ||  
             \lesssim 
             \frac{1}{n_s} \sqrt{n_s p}C_{sr, n_s} C_{sc, p}   (\sqrt{n_s} + \sqrt{p})   \\
             &= C_{sr, n_s} C_{sc, p}  (\sqrt{p} + \frac{p}{\sqrt{n_s}}).
        \end{aligned}
    \end{equation*}
    \item With probability at least $1-o(1)$, we have  \begin{equation*}
        \begin{aligned}
             \frac{n_s}{(n_s + \tau_{\Gamma_s}^{-2})^2} &|| \left\{\mathbf{E}_s + \left( \mathbf{M}_{0s}  {\mathbf{\Lambda}}_0^\top -\mathbf{\hat M}_s {\mathbf{\tilde \Lambda}}^\top\right)\right\}^\top \mathbf{U}_{0s}^\perp \mathbf{U}_{0s}^{\perp \top} \left\{\mathbf{E}_s + \left( \mathbf{M}_{0s}  {\mathbf{\Lambda}}_0^\top -\mathbf{\hat M}_s {\mathbf{\tilde \Lambda}}^\top\right)\right\}||  \\
             & 
             \lesssim \frac{1}{n_s}(C_{r, \{n_s\}_{s}^{S}} C_{c, p}   )^4\left(\sqrt{n} + \sqrt{p}\right)^2 \asymp(C_{r, \{n_s\}_{s}^{S}} C_{c, p}   )^4 \big( \frac{n}{n_s} + \frac{p}{n_s} + \frac{\sqrt{np}}{n_s}\big),
        \end{aligned}
    \end{equation*}
\end{enumerate}
since, with probability at least $1-o(1)$, $||\mathbf{M}_{0s}  {\mathbf{\Lambda}}_0^\top -\mathbf{\hat M}_s {\mathbf{\tilde \Lambda}}^\top || \lesssim (C_{r, \{n_s\}_{s}^{S}} C_{c, p}   )^2(\sqrt{n} + \sqrt{p})$ and $||\mathbf{E}_s||\lesssim C_{sr, n_s} C_{sc, p} (\sqrt{n} + \sqrt{p})  $ by Lemma \ref{lemma:E}.
For the second term, since $||\mathbf{Y}_s -\mathbf{\hat M}_s {\mathbf{\tilde \Lambda}}^\top|| \leq ||\mathbf{F}_s {\mathbf{\Gamma}}_s^\top|| + ||\mathbf{E}_s|| +  ||\mathbf{\hat M}\tilde{{\mathbf{\Lambda}}}^\top  -  \mathbf{M}_0{\mathbf{\Lambda}}_0^\top|| \lesssim  \sqrt{n_s p} + (C_{r, \{n_s\}_{s}^{S}} C_{c, p}   )^2(\sqrt{n} + \sqrt{p})  \asymp \sqrt{n_s p}$
with probability at least $1-o(1)$, we have
\begin{equation*}
    \begin{aligned}
         \frac{n_s}{(n_s + \tau_{\Gamma_s}^{-2})^2} & ||(\mathbf{Y}_s -\mathbf{\hat M}_s {\mathbf{\tilde \Lambda}}^\top)^\top ( \mathbf{U}_s^{\perp} \mathbf{U}_s^{\perp \top} - \mathbf{U}_{0s}^\perp \mathbf{U}_{0s}^{\perp \top})(\mathbf{Y}_s -\mathbf{\hat M}_s {\mathbf{\tilde \Lambda}}^\top)|| \\
         & \lesssim \frac{1}{n_s}  ||\mathbf{Y}_s -\mathbf{\hat M}_s {\mathbf{\tilde \Lambda}}^\top||^2 || \mathbf{U}_s^{\perp} \mathbf{U}_s^{\perp \top} - \mathbf{U}_{0s}^\perp \mathbf{U}_{0s}^{\perp \top} || \\
         & \lesssim \frac{1}{n_s} n_s p (C_{r, \{n_s\}_{s}^{S}} C_{c, p}   )^2\left(\frac{1}{n_s} + \frac{1}{p}\right)    \\
         & \asymp (C_{r, \{n_s\}_{s}^{S}} C_{c, p}   )^2 \big(1 +\frac{p}{n_s} \big).
         \end{aligned}
\end{equation*}
Combining all of the above, with probability at least $1-o(1)$, we have
\begin{equation*}
    \begin{aligned}
        ||   \bar {\mathbf{\Gamma}}_s \bar {\mathbf{\Gamma}}_s^\top - {\mathbf{\Gamma}}_{0s}{\mathbf{\Gamma}}_{0s}^\top|| \lesssim (C_{sr, n_s} C_{sc, p}  )\big(\sqrt{p} + \frac{p}{\sqrt{n_s}}\big) + \frac{\sqrt{\log n_s} p}{ \sqrt{n_s}}.
    \end{aligned}
\end{equation*}
The assumption $||{\mathbf{\Gamma}}_{0s}|| \asymp \sqrt{p}$, implies
\begin{equation*}
    \begin{aligned}
        \frac{ ||   \bar {\mathbf{\Gamma}}_s \bar {\mathbf{\Gamma}}_s^\top - {\mathbf{\Gamma}}_{0s}{\mathbf{\Gamma}}_{0s}^\top||}{ ||  {\mathbf{\Gamma}}_{0s}{\mathbf{\Gamma}}_{0s}^\top||} \lesssim 
       (\sqrt{\log n_s} + C_{sr, n_s} C_{sc, p}) \frac{p}{\sqrt{n}}  + C_{sr, n_s} C_{sc, p}\sqrt{p}
    \end{aligned}
\end{equation*}
with probability at least $1-o(1)$.
Given the corresponding sample for ${\mathbf{\Lambda}} $, a sample for ${\mathbf{\tilde \Gamma}}_s$ from $\tilde \Pi$ is given by
\begin{equation*}
   {\mathbf{\tilde \Gamma}}_s = \bar {\mathbf{\Gamma}}_s +\mathbf{ \mathbf{\tilde E}}_{\Gamma_s}, \quad \mathbf{ \mathbf{\tilde E}}_{\Gamma_s} =[\mathbf{e}_{\gamma_{s1}} ~ \cdots ~ \mathbf{e}_{\gamma_{sp}} ]^\top, \quad 
      \mathbf{e}_{\gamma_{sj}} \sim N_{q_s}\left(0, \frac{\rho_{\Gamma_s}^2 \tilde \sigma_j^2}{n_s + \tau_{\Gamma_s}^{-2}} \mathbf{I}_{q_s}\right),
\end{equation*}
where ${\mathbf{\Gamma}}_s =  \frac{\sqrt{n_s}}{n_s + \tau_{\Gamma_s}^{-2}} (\mathbf{Y}_s -\mathbf{\hat M}_s {\mathbf{\tilde \Lambda}}^\top)^\top \mathbf{U}_s^\perp$ and 
\begin{equation*}
    \begin{aligned}
        ||\bar {\mathbf{\Gamma}}_s  || & \asymp \frac{1}{\sqrt{n_s}} ||\mathbf{Y}_s -\mathbf{\hat M}_s {\mathbf{\tilde \Lambda}}^\top || \lesssim \frac{1}{\sqrt{n_s}} (||\mathbf{Y}_s - \mathbf{M}_{0s} {\mathbf{\Lambda}}_{0}^\top || + ||  \mathbf{M}_{0s} {\mathbf{\Lambda}}_{0}^\top  -\mathbf{\hat M}_s {\mathbf{\tilde \Lambda}}^\top ||)\\
        & \lesssim \frac{1}{\sqrt{n_s}} \sqrt{n_s p}  \asymp \sqrt{p}.
    \end{aligned}
\end{equation*}
Thus, 
\begin{equation*}
\begin{aligned}
     \left|\left|{\mathbf{\tilde \Gamma}}_s {\mathbf{\tilde \Gamma}}_s ^\top - \bar {\mathbf{\Gamma}}_s \bar {\mathbf{\Gamma}}_s^\top \right|\right| &\lesssim \left|\left| \mathbf{ \mathbf{\tilde E}}_{\Gamma_s}\mathbf{ \mathbf{\tilde E}}_{\Gamma_s}^\top \right|\right| + \left|\left|  \bar {\mathbf{\Gamma}}_s\mathbf{ \mathbf{\tilde E}}_{\Gamma_s}^\top \right|\right| \\
     &\lesssim  \frac{p}{n_s} \rho_{\Gamma_s}^2 \max_{j} \tilde \sigma_j^2 + \frac{p}{\sqrt{n_s}} \rho_{\Gamma_s}^2 \max_{j} \tilde \sigma_j^2  \\
  & \lesssim \frac{p}{n_s^{1/2}},
\end{aligned}
\end{equation*}
since $\max_{j} \tilde \sigma_j^2 \asymp 1$, with $\tilde \Pi$ probability at least $1 - o(1)$, by Lemma \ref{lemma:sigma_j_tilde}, determining the desired result.
\end{proof}

\begin{proof}[Proof of Theorem \ref{thm:clt_mu_Lambda_outer}]
 Follows as a special case of Theorem \ref{thm:clt_mu_Lambda_outer_general}
\end{proof}
{
\begin{theorem}\label{thm:clt_mu_Lambda_outer_general}
    Suppose Assumptions \ref{assumption:model}--\ref{assumption:sigma_lb} hold, the residual error variance matrices $\{\mathbf \Sigma_{0sc}\}_s$ are equal across studies, $n_s = \mathcal O(n_{\min}^2)$, where $n_{\min} = \min_{s=1, \dots, S} n_s$, for all $s=1, \dots, S$, $(C_{r, \{n_s\}_{s}^{S}} C_{c, p})^2\sqrt{n}/p = o(1)$ and $\log^2 p / \sqrt{n_{\min}}$. For $1 \leq j \leq j' \leq p$, let 
    \begin{equation}\label{eq:S_0_sq}
        \mathcal S_{0jj'}^2 = \begin{cases}
             \sigma_{0j}^2 ||\mathbf{\lambda}_{0j'}||^2 + \sigma_{0j'}^2 ||\mathbf{\lambda}_{0j}||^2 +  \sigma_{0jj'}\mathbf{\lambda}_{0j}^\top \mathbf{\lambda}_{0j'} + ||\mathbf{\lambda}_{0j}||^2 ||\mathbf{\lambda}_{0j'}||^2 +  (\mathbf{\lambda}_{0j}^\top \mathbf{\lambda}_{0j'})^2 \quad &\text{if } j \neq j',\\
       2 ||\mathbf{\lambda}_{0j}||^4 +4 \sigma_{0j}^2||\mathbf{\lambda}_{0j}||^2 
\quad &\text{otherwise,}
        \end{cases}
    \end{equation}
and
   \begin{equation}\label{eq:S_0_s_sq}
        \mathcal S_{0sjj'}^2 = \begin{cases}
        \begin{aligned}
               &\sigma_{0j}^2 ||\mathbf{\gamma}_{0sj'}||^2 + \sigma_{0j'}^2 ||\mathbf{\gamma}_{0sj}||^2 + \sigma_{0jj'} \mathbf{\gamma}_{0sj}^\top \mathbf{\gamma}_{0sj'} +||\mathbf{\gamma}_{0sj}||^2 ||\mathbf{\gamma}_{0sj'}||^2 \quad  \\ & +  (\mathbf{\gamma}_{0sj}^\top \mathbf{\gamma}_{0sj'})^2  + ||\mathbf{\gamma}_{0sj}||^2 ||\mathbf{\lambda}_{0j}||^2 
             + ||\mathbf{\gamma}_{0sj'}||^2 ||\mathbf{\lambda}_{0j}||^2 +  2\mathbf{\gamma}_{0sj}^\top \mathbf{\gamma}_{0sj'}\mathbf{\lambda}_{0j}^\top \mathbf{\lambda}_{0j'}
        \end{aligned} \quad   &\text{if } j \neq j',
          \\
       2 ||\mathbf{\gamma}_{0sj}||^4 + 4||\mathbf{\gamma}_{0sj}||^2||\mathbf{\lambda}_{0j}||^2 +4 \sigma_{0j}^2 ||\mathbf{\gamma}_{0sj}||^2 
 &\text{otherwise,}
        \end{cases}
    \end{equation}
    where $\sigma_{0jj'}$ denotes the covariance between the $j$-th and $j'$-th variables. 
    Then, as $n_1, \dots, n_s, p \to \infty$, we have
    \begin{equation}\label{eq:mu_lambda_clt}
    \begin{aligned}
        &\frac{\sqrt{n}}{S_{0, jj'}} \left(\mathbf{\mu}_{\lambda_j}^\top \mathbf{\mu}_{\lambda_{j'}}-   \mathbf{\lambda}_{0j}^\top\mathbf{\lambda}_{0j'}\right) \Longrightarrow N(0, 1),\\
        & \frac{\sqrt{n_s}}{S_{0, sjj'}} \left(\mathbf{\mu}_{\gamma_{sj}}^\top \mathbf{\mu}_{\gamma_{sj'}} - \mathbf{\gamma}_{0sj}^\top \mathbf{\gamma}_{0sj'}\right) \Longrightarrow N(0, 1), \quad (s=1, \dots, S). 
    \end{aligned} 
    \end{equation}
\end{theorem}
}
\begin{proof}[Proof of Theorem \ref{thm:clt_mu_Lambda_outer_general}]
We start by proving the result for $\mathbf{\mu}_{\lambda_j}^\top \mathbf{\mu}_{\lambda_{j'}}$. Consider 
\begin{equation*}
\begin{aligned}
   \sqrt{n} \mathbf{\mu}_{\lambda_j}^\top \mathbf{\mu}_{\lambda_{j'}}  &= \sqrt{n} \frac{n}{(n + \tau_{\Lambda}^{-1})^2}\mathbf{\hat y}^{c (j)}  \mathbf{U}^c \mathbf{U}^{c \top}\mathbf{\hat y}^{c (j')} \\
    &=  \sqrt{n} \frac{n}{(n + \tau_{\Lambda}^{-1})^2} \left\{\mathbf{\hat y}^{c (j)}  \mathbf{U}_0^c \mathbf{U}_0^{c \top}\mathbf{\hat y}^{c (j')} +\mathbf{\hat y}^{c (j)}  \left( \mathbf{U}^c \mathbf{U}^{c \top} - \mathbf{U}_0^c \mathbf{U}_0^{c \top}\right)\mathbf{\hat y}^{c (j')} \right\}.
\end{aligned}     
\end{equation*}
 Let $ \sqrt{n} \frac{n}{(n + \tau_{\Lambda}^{-1})^2} = \frac{1}{\sqrt{n}} + \nu_n$, where $\nu_n \asymp 1/n^{3/2}$,
 $\mathbf{\hat y}^{c (j)} =  \mathbf{M}_{0} \mathbf{\lambda}_{0j} +  \mathbf{e}^{(j)} + \mathbf{r}^{(j)}$, $ \mathbf{r}^{(j)} =\mathbf{\Delta} y^{(j)}   - \mathbf{P}_0^\perp ( \mathbf{M}_{0} \mathbf{\lambda}_{0j} + \mathbf{e}^{(j)} )$, and $\Delta$ and $\mathbf{P}_0^\perp$ are defined in \eqref{eq:Delta} and \eqref{eq:P_0_perp}, respectively.  First, we decompose $ \mathbf{\hat y}^{c (j) \top} \mathbf{U}_0^c \mathbf{U}_0^{c \top} \mathbf{\hat y}^{c (j')}$ as 
\begin{equation*} 
\begin{aligned}
   \mathbf{\hat y}^{c (j) \top} \mathbf{U}_0^c \mathbf{U}_0^{c \top} \mathbf{\hat y}^{c (j')} =& \mathbf{\lambda}_j^\top  \mathbf{M}_0^\top   \mathbf{M}_0 \mathbf{\lambda}_{j'} +  \mathbf{\lambda}_j^\top  \mathbf{M}_0^\top   \mathbf{e}^{(j')} + \mathbf{\lambda}_{j'}^\top  \mathbf{M}_0^\top   \mathbf{e}^{(j)}\\\
    & + \mathbf{\lambda}_j^\top  \mathbf{M}_0^{c\top}  \mathbf{U}_0^c  \mathbf{U}_0^{c \top}\mathbf{r}^{(j')} + \mathbf{\lambda}_{j'}^\top  \mathbf{M}_0^{c\top}  \mathbf{U}_0^c  \mathbf{U}_0^{c \top}\mathbf{r}^{(j)} \\
    & + \mathbf{e}^{(j)\top}  \mathbf{U}_0^c  \mathbf{U}_0^{c \top}\mathbf{r}^{(j')} +  \mathbf{e}^{(j')\top}  \mathbf{U}_0^c  \mathbf{U}_0^{c \top}\mathbf{r}^{(j)} \\
    & + \mathbf{r}^{(j)\top} \mathbf{U}_0^c \mathbf{U}_0^{c \top}  \mathbf{r}^{(j')} + \mathbf{e}^{(j)\top} \mathbf{U}_0^c \mathbf{U}_0^{c \top}  \mathbf{e}^{(j')}. 
\end{aligned}\end{equation*}

Note that $\mathbf{\lambda}_j^\top  \mathbf{M}_0^\top   \mathbf{M}_0 \mathbf{\lambda}_{j'} = \sum_{i=1}^n (\mathbf{\lambda}_{0j}^\top \mathbf{\eta}_i)(\mathbf{\lambda}_{0j'}^\top \mathbf{\eta}_i)$, with $ E[(\mathbf{\lambda}_{0j}^\top \mathbf{\eta}_i)(\mathbf{\lambda}_{0j'}^\top \mathbf{\eta}_i)] = \mathbf{\lambda}_{0j}^\top \mathbf{\lambda}_{0j'}$, and\\ $\mathbb V[(\mathbf{\lambda}_{0j}^\top \mathbf{\eta}_i)(\mathbf{\lambda}_{0j'}^\top \mathbf{\eta}_i)] = \xi_{jj'}^2$, where $$\xi_{jj'}^2 = \begin{cases}
    (\mathbf{\lambda}_{0j}^\top \mathbf{\lambda}_{0j'})^2 + ||\mathbf{\lambda}_{0j}||^2|| \mathbf{\lambda}_{0j'}||^2, \quad &\text{if } j \neq j',\\
    2||\mathbf{\lambda}_{0j}||^4 &\text{otherwise.} 
\end{cases}$$\\
In addition, the $(\mathbf{\lambda}_{0j}^\top \mathbf{\eta}_i)(\mathbf{\lambda}_{0j'}^\top \mathbf{\eta}_i)$'s are independent of each other.
Therefore, an application of the central limit theorem gives
\begin{equation*}
   \frac{1}{\sqrt{n}} \mathbf{\lambda}_j^\top  \mathbf{M}_0^\top   \mathbf{M}_0 \mathbf{\lambda}_{j'} \Longrightarrow N(\mathbf{\lambda}_{0j}^\top \mathbf{\lambda}_{0j'}, \xi_{jj'}^2).
\end{equation*}
Next, for $j \neq j'$, let
\begin{equation*}
\begin{aligned}
     l_{jj'}^2( \mathbf{M}_0 ) &= \sigma_{0j}^2 \mathbf{\lambda}_{0j'}^\top \frac{ \mathbf{M}_0^\top  \mathbf{M}_0}{n} \mathbf{\lambda}_{0j'} + \sigma_{0j'}^2 \mathbf{\lambda}_{0j}^\top \frac{ \mathbf{M}_0^\top  \mathbf{M}_0}{n}\mathbf{\lambda}_{0j}  + \sigma_{0jj'} \mathbf{\lambda}_{0j}^\top \frac{ \mathbf{M}_0^\top  \mathbf{M}_0}{n}\mathbf{\lambda}_{0j'}, \\
    l_{0jj'}^2 &= \sigma_{0j}^2||\mathbf{\lambda}_{0j'}||^2  + \sigma_{0j'}^2||\mathbf{\lambda}_{0j}||^2 +\sigma_{0jj'}\mathbf{\lambda}_{0j}^\top \mathbf{\lambda}_{0j'}, 
\end{aligned}
\end{equation*}
and note
\begin{equation*}
\begin{aligned}
     \frac{1}{\sqrt{n}}\left(\mathbf{\lambda}_{0j}  \mathbf{M}_0^\top   \mathbf{e}^{(j')} + \mathbf{\lambda}_{0j'}  \mathbf{M}_0^\top   \mathbf{e}^{(j)}\right) \mid  \mathbf{M}_0 &\sim N(0, l_{jj'}^2( \mathbf{M}_0))  \overset{d}{=}l_{jj'}( \mathbf{M}_0) z_{jj'}, \\ z_{jj'}&\sim N(0, 1), \quad z_{jj'} \perp  \mathbf{M}_0,
\end{aligned}
\end{equation*}
where $ \overset{d}{=}$ denotes equality in distribution.  Hence, 
\begin{equation*}
    \frac{1}{\sqrt{n}}\left(\mathbf{\lambda}_{0j}  \mathbf{M}_0^\top   \mathbf{e}^{(j')} + \mathbf{\lambda}_{0j'}  \mathbf{M}_0^\top   \mathbf{e}^{(j)}\right) =  l_{0jj'} z_{jj'} + \left(l_{jj'}( \mathbf{M}_0) -l_{0jj'} \right)z_{jj'}.
\end{equation*}
Since $|| \mathbf{M}_0|| \asymp \sqrt{n}$, we have $|l_{jj'}( \mathbf{M}_0 )  + l_{0jj'}| \asymp 1$ with probability at least $1 - o(1)$, and $l_{jj'}( \mathbf{M}_0 )  + l_{0jj'} > 0$.
\begin{equation*}
    \begin{aligned}
         l_{0jj'} z_{jj'} + \left(l_{jj'}( \mathbf{M}_0) -l_{0jj'} \right)z_{jj'} & =  l_{0jj'} z_{jj'} + \frac{\left(l_{jj'}^2( \mathbf{M}_0) -l_{0jj'}^2 \right)}{l_{jj'}^2( \mathbf{M}_0 )  + l_{0jj'}^2}z_{jj'}\\
         &\Longrightarrow N(0, l_{0jj'}^2),    
         \end{aligned}
\end{equation*}
since $\left|\left|\frac{ \mathbf{M}_0^\top  \mathbf{M}_0}{n} - \mathbf{I}_{k_0} \right| \right| \lesssim \frac{\sqrt{\log n}}{\sqrt{n}}$ with probability $1 - o(1)$ by Lemma E.1 of \citet{fable}, making $\frac{\left(l_{jj'}^2( \mathbf{M}_0) -l_{0jj'}^2 \right)}{l_{jj'}^2( \mathbf{M}_0 )  + l_{0jj'}^2}z_{jj'} \lesssim \frac{\sqrt{\log n}}{\sqrt{n}}$ with probability at least $1 - o(1)$. 
For $j=j'$, we can show with similar steps
\begin{equation*}
     \frac{1}{\sqrt{n}}2\mathbf{\lambda}_{0j}^\top  \mathbf{M}_0^\top   \mathbf{e}^{(j)}\Longrightarrow N(0, 4 \sigma_{0j}^2 ||\mathbf{\lambda}_{0j}||^2).
\end{equation*}
Combining the results above and using the fact that elements of $\mathbf{M}_0$ are uncorrelated with $\mathbf{e}^{(j)}$ and $\mathbf{e}^{(j')}$, we have
\begin{equation*}
     \frac{1}{\sqrt{n}} \left(\mathbf{\lambda}_{0j}^\top  \mathbf{M}_0^\top   \mathbf{M}_0 \mathbf{\lambda}_{0j'} + \mathbf{\lambda}_{0j}^\top  \mathbf{M}_0^\top   \mathbf{e}^{(j')} + \mathbf{\lambda}_{0j'}  \mathbf{M}_0^\top   \mathbf{e}^{(j)}\right) \Longrightarrow N(\mathbf{\lambda}_{0j}^\top \mathbf{\lambda}_{0j'}, \xi_{jj'}^2 + l_{0jj'}^2).
\end{equation*}

In the rest of the proof, we show the remaining terms can be suitably bounded by sequences decreasing to 0.
In the following, we use the fact that, with probability at least $1-o(1)$,
\begin{enumerate}
    \item  $  \max_{j=1, \dots, p} ||\mathbf{U}_0^{c \top} \mathbf e^{(j)} || \lesssim  \log p$ by Lemma \ref{lemma:U_c_e_j},
    \item $ \max_{j=1, \dots, p}|| \mathbf{M}_0^{\top}  \mathbf{P}_0^\perp  \mathbf{e}^{(j)}|| \lesssim \log p$ by Lemma \ref{lemma:Mt_P_M},
    \item $||\mathbf M_0|| \asymp \sqrt{n}$ by Theorem 4.6.1 of \citet{vershynin2018hdp},
    \item $|| \mathbf{M}_0^{\top}  \mathbf{P}_0^\perp  \mathbf{M}_0^{\top}||\asymp 1$ by Lemma \ref{lemma:Mt_P_M},
    \item $ \max_{j=1, \dots, p}||\mathbf{\Delta y}^{ (j)}|| \asymp  C_{\sigma} C_{r, \{n_s\}_{s}^{S}}
\big(\frac{1}{\sqrt{n_{\min}}} + \frac{\sqrt{n_{\max}}}{p}\big) $  by Lemma \ref{lemma:Delta_y_j},
\item $|| \mathbf{r}^{(j')}|| \asymp  C_{r, \{n_s\}_{s}^{S}}
(\frac{1}{\sqrt{n_{\min}}} + \frac{\sqrt{n_{\max}}}{p}) + \log p$ with probability at least $1 - o(1)$ by Lemma \ref{lemma:r_j},
\item $ \max_{j=1, \dots, p}||\mathbf{e}^{(j)}||  \lesssim  \sqrt{\log p}$ by Corollary 5.35 of \citet{vershynin_12},
\item $ ||  \mathbf{U}^c \mathbf{U}^{c \top} - \mathbf{U}_0^c \mathbf{U}_0^{c \top}|| \lesssim (C_{r, \{n_s\}_{s}^{S}} C_{c, p}   )^2 \big(\frac{1}{n} + \frac{1}{p}\big)$ by Proposition \ref{prop:recovery_M}.
\end{enumerate}
along with $\max_{j, \dots, p}||\mathbf{\lambda}_{0j}|\lesssim  ||{\mathbf{\Lambda}}||_{\infty}\sqrt{k_0} \asymp 1$ by Assumption \ref{assumption:Lambda}. 
\begin{enumerate}
    \item With probability at least $1 - o(1)$, we have
    \begin{equation*}
        \begin{aligned}
            \sqrt{n} \frac{n}{(n + \tau_{\Lambda}^{-1})^2} \left| \mathbf{e}^{(j)\top}  \mathbf{U}_0^c  \mathbf{U}_0^{c \top}  \mathbf{e}^{(j')}\right| \asymp \frac{1}{\sqrt{n}} ||\mathbf{U}_0^{c \top}  \mathbf{e}^{(j)}|| ||\mathbf{U}_0^{c \top}  \mathbf{e}^{(j')}|| \lesssim \frac{\log p}{\sqrt{n}}.
        \end{aligned}
    \end{equation*}
   \item With probability at least $1 - o(1)$, we have \begin{equation*}
        \begin{aligned}
  \sqrt{n} \frac{n}{(n + \tau_{\Lambda}^{-1})^2} &\left|\mathbf{\lambda}_{0j}^\top  \mathbf{M}_0^{c\top}  \mathbf{U}_0^c  \mathbf{U}_0^{c \top}\mathbf{r}^{(j')}\right|  =  \sqrt{n} \frac{n}{(n + \tau_{\Lambda}^{-1})^2} \left|\mathbf{\lambda}_{0j}^\top  \mathbf{M}_0^{c\top} \mathbf{r}^{(j')} \right|\\
  & \lesssim \frac{1}{\sqrt{n}} \left[|| \mathbf{\lambda}_j|| \left\{|| \mathbf{\lambda}_{0j'}|| || \mathbf{M}_0^{\top}  \mathbf{P}_0^\perp  \mathbf{M}_0^{\top}|| + || \mathbf{M}_0^{\top}  \mathbf{P}_0^\perp  \mathbf{e}^{(j')}|| + || \mathbf{M}_0|| ||\mathbf{ \Delta} y^{c (j')}|| \right\}   \right] \\
  & \lesssim \frac{ \log p }{\sqrt{n}}+  C_{\sigma} C_{r, \{n_s\}_{s}^{S}}
\big(\frac{1}{\sqrt{n_{\min}}} + \frac{\sqrt{n_{\max}}}{p}\big).
        \end{aligned}
    \end{equation*}
    Similarly, $ \sqrt{n} \frac{n}{(n + \tau_{\Lambda}^{-1})^2} \big|\mathbf{\lambda}_{0j'}  \mathbf{M}_0^{c\top}  \mathbf{U}_0^c  \mathbf{U}_0^{c \top}\mathbf{r}^{(j)} \big| \lesssim  \frac{ \log p }{\sqrt{n}}+  C_{\sigma} C_{r, \{n_s\}_{s}^{S}}
\big(\frac{1}{\sqrt{n_{\min}}} + \frac{\sqrt{n_{\max}}}{p}\big)$ with probability at least $1 - o(1)$.
    \item With probability at least $1 - o(1)$, we have
    \begin{equation*}
        \begin{aligned}
             \sqrt{n} \frac{n}{(n + \tau_{\Lambda}^{-1})^2} \left| \mathbf{e}^{(j)\top}  \mathbf{U}_0^c  \mathbf{U}_0^{c \top}\mathbf{r}^{(j')} \right| & \leq \sqrt{n} \frac{n}{(n + \tau_{\Lambda}^{-1})^2} || \mathbf{e}^{(j)\top}  \mathbf{U}_0^c|||| \mathbf{U}_0^{c \top}\mathbf{r}^{(j')} || \\
  & \lesssim \frac{1}{\sqrt{n}}|| \mathbf{e}^{(j)\top}  \mathbf{U}_0^c||  || \mathbf{U}_0^c|| ||\mathbf{r}^{(j')}|| \\
  & \lesssim \frac{1}{\sqrt{n}} C_{\sigma} C_{r, \{n_s\}_{s}^{S}}
\big(\frac{1}{\sqrt{n_{\min}}}  + \frac{\sqrt{n_{\max}}}{p}\big)\log p + \frac{\log p^2}{\sqrt{n}}.
        \end{aligned}
    \end{equation*}
    Similarly, $ \sqrt{n} \frac{n}{(n + \tau_{\Lambda}^{-1})^2} \left| \mathbf{e}^{(j')\top}  \mathbf{U}_0^c  \mathbf{U}_0^{c \top}\mathbf{r}^{(j)} \right| \lesssim \frac{1}{\sqrt{n}} C_{\sigma} C_{r, \{n_s\}_{s}^{S}}
\big(\frac{1}{\sqrt{n_{\min}}}  + \frac{\sqrt{n_{\max}}}{p}\big)\log p + \frac{\log p^2}{\sqrt{n}}$ with probability at least $1 - o(1)$.

     \item With probability at least $1 - o(1)$, we have
     \begin{equation*}
         \begin{aligned}
             \sqrt{n} \frac{n} {(n + \tau_{\Lambda}^{-1})^2} \left| \mathbf{r}^{(j)\top}  \mathbf{U}_0^c  \mathbf{U}_0^{c \top}\mathbf{r}^{(j')} \right| & \leq \sqrt{n} \frac{n}{(n + \tau_{\Lambda}^{-1})^2} || \mathbf{r}^{(j)\top} || || \mathbf{U}_0^c \mathbf{U}_0^{c \top} || || \mathbf{r}^{(j')} || \\
  & \lesssim \frac{1}{\sqrt{n} }\big\{C_{r, \{n_s\}_{s}^{S}}
(\frac{1}{\sqrt{n_{\min}}} + \frac{\sqrt{n_{\max}}}{p}) + \log p\big\}.
         \end{aligned}
     \end{equation*}

    \item With probability at least $1 - o(1)$, we have 
    \begin{equation*}
        \begin{aligned}
            \nu_n \big(\mathbf{\lambda}_{0j}^\top  \mathbf{M}_0^\top   \mathbf{M}_0 \mathbf{\lambda}_{0j'} +& \mathbf{\lambda}_{0j}^\top  \mathbf{M}_0^\top   \mathbf{e}^{(j')} + \mathbf{\lambda}_{0j'}^\top  \mathbf{M}_0^\top   \mathbf{e}^{(j)}\big)  \\
            & \lesssim \nu_n \big(||\mathbf{\lambda}_{0j}|| ||\mathbf{\lambda}_{0j'}|| || \mathbf{M}_0||^2 + ||\mathbf{\lambda}_{0j}|| || \mathbf{M}_0|| || \mathbf{e}^{(j')} || \\
            & \quad \quad \quad +  ||\mathbf{\lambda}_{0j'}|| || \mathbf{M}_0|| || \mathbf{e}^{(j)}||\big) \\
            & \lesssim \frac{1}{n^{3/2}}( n + \log p) \asymp \frac{1}{\sqrt{n}} + \frac{\log p}{n}.
        \end{aligned}
    \end{equation*}
    \item With probability at least $1 - o(1)$, we have
    \begin{equation*}
        \begin{aligned}
            \sqrt{n} \frac{n}{(n + \tau_{\Lambda}^{-1})^2} \big|\mathbf{\hat y}^{c (j)}  \big( \mathbf{U}^c \mathbf{U}^{c \top} -& \mathbf{U}_0^c \mathbf{U}_0^{c \top}\big) \mathbf{\hat y}^{c (j')} \big| \\ 
            & \lesssim   \sqrt{n} \frac{n}{(n + \tau_{\Lambda}^{-1})^2} ||\mathbf{\hat y}^{c (j)} || ||\mathbf{\hat y}^{c (j')} || ||  \mathbf{U}^c \mathbf{U}^{c \top} - \mathbf{U}_0^c \mathbf{U}_0^{c \top}||\\
           & \lesssim \frac{1}{\sqrt{n}} n (C_{r, \{n_s\}_{s}^{S}} C_{c, p}   )^2 \big(\frac{1}{n} + \frac{1}{p}\big) \lesssim  (C_{r, \{n_s\}_{s}^{S}} C_{c, p}   )^2 \big(\frac{1}{\sqrt{n}} + \frac{\sqrt{n}}{p}\big).
        \end{aligned}
    \end{equation*}

\end{enumerate}
To complete the proof, we apply Lemma F.2 from \citet{fable}.


Next, we move to $\mathbf{\mu}_{\Gamma_s} \mathbf{\mu}_{\Gamma_s}^\top$. Recall the posterior mean for $\mathbf{\gamma}_{sj}$
\begin{equation*}
    \mathbf{\mu}_{\gamma_{sj}} =  \frac{1}{n_s + \tau_{\Gamma_s}^{-2}} \mathbf{\hat  F_s}^\top (\mathbf{y}_{sj} -\mathbf{\hat M}_s \mathbf{\mu}_{\lambda_j}).
\end{equation*}
 Consider the singular value decomposition of $(\mathbf{I}_n - \mathbf{P}_0^\perp) \mathbf{M}_0 {\mathbf{\Lambda}}^\top = \mathbf{U}' \mathbf{D}' \mathbf{V}'^\top $. Following similar steps to those in the proof of Proposition \ref{prop:recovery_M}, we can show that $||\mathbf{U}^c \mathbf{U}^{c \top} - \mathbf{U}'\mathbf{U}'^\top|| \lesssim (C_{r, \{n_s\}_{s}^{S}} C_{c, p}   )^2 \big(\frac{1}{n} + \frac{1}{p} \big)$ with probability at least $1-o(1)$. Moreover, define $\mathbf{U}_s' \in \mathbb R^{n_s \times k_0}$ to be the block of $\mathbf{U}' \in R^{n \times k_0}$ corresponding to the $s$-th study, that is $\mathbf{U}' = \big[ \mathbf{U}_1^{'\top} ~ \cdots ~ \mathbf{U}_S^{'\top}\big]$. 
Therefore, 
 \begin{equation*}
    \begin{aligned}
       \mathbf{\hat M}_s \mathbf{\mu}_{\lambda_j} =  \frac{n}{n + \tau_{\Lambda}^{-2}} \mathbf{U}_s^{c} \mathbf{U}^{c\top}\mathbf{\hat y}^{c (j)}  = \frac{n}{n + \tau_{\Lambda}^{-2}} \left\{\mathbf{U}_s' \mathbf{U}^{'\top}\mathbf{\hat y}^{c (j)}   + ( \mathbf{U}_s^{c} \mathbf{U}^{c\top} - \mathbf{U}_s' \mathbf{U}^{'\top})\mathbf{\hat y}^{c (j)} \right\}  . 
    \end{aligned}
\end{equation*}
where 
\begin{equation*}
    \begin{aligned}
        \mathbf{U}_s' \mathbf{U}^{'\top}\mathbf{\hat y}^{c (j)}  &= \mathbf{U}_s' \mathbf{U}^{'\top}\left\{(\mathbf{I}_n - \mathbf{P}_0^\perp)( \mathbf{M}_{0} \mathbf{\lambda}_{0j} +  \mathbf{e}^{(j)})  + \mathbf{\Delta} y^{(j)}  \right\}\\
        &= (\mathbf{I}_{n_s} - \mathbf{U}_{0s}^\perp \mathbf{U}_{0s}^{\perp \top}) \mathbf{M}_{0s} \mathbf{\lambda}_{0j} +  \mathbf{U}_s' \mathbf{U}^{'\top} \left\{(\mathbf{I}_n - \mathbf{P}_0^\perp) \mathbf{e}^{(j)}  - \mathbf{\Delta} y^{(j)}\right\}.
    \end{aligned}
\end{equation*}
Hence, 
\begin{equation*}
    \begin{aligned}
        \mathbf{M}_{0s} \mathbf{\lambda}_{0j} -\mathbf{\hat M}_s \mathbf{\mu}_{\lambda_j} =\frac{n}{n + \tau_{\Lambda}^{-2}} \big[ & \mathbf{U}_{0s}^\perp \mathbf{U}_{0s}^{\perp \top} \mathbf{M}_{0s} \mathbf{\lambda}_{0j} - \mathbf{U}_s' \mathbf{U}^{'\top} \left\{(\mathbf{I}_n - \mathbf{P}_0^\perp) \mathbf{e}^{(j)}  - \mathbf{\Delta} y^{(j)}\right\} \\ &  +( \mathbf{U}_s^{c} \mathbf{U}^{c\top} - \mathbf{U}_s' \mathbf{U}^{'\top})\mathbf{\hat y}^{c (j)} \big]+ \frac{\tau_{\Lambda}^{-2}}{n_s + \tau_{\Lambda}^{-2}} \mathbf{M}_{0s} \mathbf{\lambda}_{0j} 
    \end{aligned}
\end{equation*}
and
\begin{equation*}
\begin{aligned}
    \mathbf{y}_{sj} -\mathbf{\hat M}_s \mathbf{\mu}_{\lambda_j} = \mathbf{y}_{sj} -  \mathbf{M}_{0s} &\mathbf{\lambda}_{0j} + (\mathbf{M}_{0s} \mathbf{\lambda}_{0j} -\mathbf{\hat M}_s \mathbf{\mu}_{\lambda_j}) \\
    \quad =  -  \frac{n}{n + \tau_{\Lambda}^{-2}} \big[& \mathbf{U}_{0s}^\perp \mathbf{U}_{0s}^{\perp \top} \mathbf{M}_{0s} \mathbf{\lambda}_{0j} - \mathbf{U}_s' \mathbf{U}^{'\top} \left\{(\mathbf{I}_n - \mathbf{P}_0^\perp) \mathbf{e}^{(j)}  - \mathbf{\Delta} y^{(j)}\right\} \\
    &+ ( \mathbf{U}_s^{c} \mathbf{U}^{c\top} - \mathbf{U}_s' \mathbf{U}^{'\top})\mathbf{\hat y}^{c (j)} \big]
    + \mathbf{F}_{0s}\mathbf{\gamma}_{0sj} + \mathbf{e}_s^{(j)} + \frac{\tau_{\Lambda}^{-2}}{n_s + \tau_{\Lambda}^{-2}} \mathbf{M}_{0s} \mathbf{\lambda}_{0j},
\end{aligned} 
\end{equation*}

\begin{equation*}
    \begin{aligned}
       \mathbf{\mu}_{\gamma_{sj}} &=  \frac{1}{n_s + \tau_{\Gamma_s}^{-2}} \mathbf{\hat  F_s}^\top \big[\mathbf{F}_{0s}\mathbf{\gamma}_{0sj} + \mathbf{e}_s^{(j)}  + \frac{1}{n + \tau_{\Lambda}^{-2}}\left(n\mathbf{U}_{0s}^\perp \mathbf{U}_{0s}^{\perp \top} - \tau_{\Lambda}^{-2} \mathbf{I}_{n_s} \right)\mathbf{M}_{0s} \mathbf{\lambda}_{0j} \\ & 
     \quad \quad \quad \quad \quad \quad \quad  +  \mathbf{U}_s' \mathbf{U}^{'\top} \left\{(\mathbf{I}_n - \mathbf{P}_0^\perp)  \mathbf{e}^{(j)} + \mathbf{\Delta} y^{(j)}\right\} + ( \mathbf{U}_s^{c} \mathbf{U}^{c\top} - \mathbf{U}_s' \mathbf{U}^{'\top})\mathbf{\hat y}^{c (j)}\big]\\
       &=  \frac{\sqrt{n_s}}{n_s + \tau_{\Gamma_s}^{-2}}  \mathbf{U}_s^{\perp \top} \big[\mathbf{F}_{0s}\mathbf{\gamma}_{0sj} + \mathbf{e}_s^{(j)}  + \frac{1}{n + \tau_{\Lambda}^{-2}}\left(n\mathbf{U}_{0s}^\perp \mathbf{U}_{0s}^{\perp \top} - \tau_{\Lambda}^{-2} \mathbf{I}_{n_s} \right)\mathbf{M}_{0s}  \mathbf{\lambda}_{0j}\\ & 
     \quad \quad \quad \quad \quad \quad \quad   + \mathbf{U}_s' \mathbf{U}^{'\top} \left\{(\mathbf{I}_n - \mathbf{P}_0^\perp)  \mathbf{e}^{(j)} + \mathbf{\Delta} y^{(j)}\right\}
       + ( \mathbf{U}_s^{c} \mathbf{U}^{c\top} - \mathbf{U}_s' \mathbf{U}^{'\top})\mathbf{\hat y}^{c (j)}\big]
    \end{aligned}
\end{equation*}
Therefore, we have 
\begin{equation*}
    \begin{aligned}
       \sqrt{n_s}  \mathbf{\mu}_{\gamma_{sj}}^\top \mathbf{\mu}_{\gamma_{sj'}}& = \frac{n_s^{3/2}}{(n_s + \tau_{\Gamma_s}^{-2})^2} \left[\mathbf{F}_{0s}\mathbf{\gamma}_{0sj'} + \mathbf{e}_s^{(j')}  + \frac{1}{n + \tau_{\Lambda}^{-2}}\left(n\mathbf{U}_{0s}^\perp \mathbf{U}_{0s}^{\perp \top} - \tau_{\Lambda}^{-2} \mathbf{I}_{n_s} \right)\mathbf{M}_{0s}  \mathbf{\lambda}_{0j'}\right.\\  & \left. \quad \quad \quad \quad \quad \quad \quad + 
       \mathbf{U}_s' \mathbf{U}^{'\top} \left\{(\mathbf{I}_n - \mathbf{P}_0^\perp)  \mathbf{e}^{(j')} +\mathbf{\Delta} y^{(j')}\right\}
       + ( \mathbf{U}_s^{c} \mathbf{U}^{c\top} - \mathbf{U}_s' \mathbf{U}^{'\top})\mathbf{\hat y}^{c (j')}\right] \\  &  \quad \quad \quad \quad \quad \quad 
 \mathbf{U}_{0s}^{\perp } \mathbf{U}_{0s}^{\perp \top} \left[\mathbf{F}_{0s}\mathbf{\gamma}_{0sj} + \mathbf{e}_s^{(j)}  + \frac{1}{n + \tau_{\Lambda}^{-2}}\left(n\mathbf{U}_{0s}^\perp \mathbf{U}_{0s}^{\perp \top} - \tau_{\Lambda}^{-2} \mathbf{I}_{n_s} \right)\mathbf{M}_{0s}  \mathbf{\lambda}_{0j}\right.\\
       & \left. \quad \quad \quad \quad \quad \quad \quad \quad  + \mathbf{U}_s' \mathbf{U}^{'\top} \left\{(\mathbf{I}_n - \mathbf{P}_0^\perp)  \mathbf{e}^{(j)} + \mathbf{\Delta} y^{(j)}\right\} 
       + ( \mathbf{U}_s^{c} \mathbf{U}^{c\top} - \mathbf{U}_s' \mathbf{U}^{'\top})\mathbf{\hat y}^{c (j)}\right] \\
         &  + \frac{n_s^{3/2}}{(n_s + \tau_{\Gamma_s}^{-2})^2}\left[\mathbf{F}_{0s}\mathbf{\gamma}_{0sj'} + \mathbf{e}_s^{(j')}  + \frac{1}{n + \tau_{\Lambda}^{-2}}\left(n\mathbf{U}_{0s}^\perp \mathbf{U}_{0s}^{\perp \top} - \tau_{\Lambda}^{-2} \mathbf{I}_{n_s} \right)\mathbf{M}_{0s}  \mathbf{\lambda}_{0j'}  \right.\\  & \left. \quad \quad \quad \quad \quad \quad \quad + \mathbf{U}_s' \mathbf{U}^{'\top} \left\{(\mathbf{I}_n - \mathbf{P}_0^\perp)  \mathbf{e}^{(j')} +\mathbf{\Delta} y^{(j')}\right\} 
       + ( \mathbf{U}_s^{c} \mathbf{U}^{c\top} - \mathbf{U}_s' \mathbf{U}^{'\top})\mathbf{\hat y}^{c (j')}\right] \\
       & \quad \quad \quad \quad \quad \quad \quad \quad  \left( \mathbf{U}_s^{\perp} \mathbf{U}_s^{\perp \top}- \mathbf{U}_{0s}^{\perp } \mathbf{U}_{0s}^{\perp \top}\right) \\
       & \quad \quad \quad \quad \quad \quad \left[\mathbf{F}_{0s}\mathbf{\gamma}_{0sj} + \mathbf{e}_s^{(j)}   + \frac{1}{n + \tau_{\Lambda}^{-2}}\left(n\mathbf{U}_{0s}^\perp \mathbf{U}_{0s}^{\perp \top} - \tau_{\Lambda}^{-2} \mathbf{I}_{n_s} \right)\mathbf{M}_{0s}  \mathbf{\lambda}_{0j}\right.\\
       & \left. \quad \quad \quad \quad \quad \quad \quad \quad + \mathbf{U}_s' \mathbf{U}^{'\top} \left\{(\mathbf{I}_n - \mathbf{P}_0^\perp)  \mathbf{e}^{(j)} + \mathbf{\Delta} y^{(j)}\right\} \right]
    \end{aligned}
\end{equation*}
First, note the following
\begin{equation*}
\begin{aligned}
    &\mathbf{\gamma}_{0sj'}^\top \mathbf{F}_{0s}^\top \mathbf{U}_{0s}^{\perp } \mathbf{U}_{0s}^{\perp \top}  \mathbf{F}_{0s}\mathbf{\gamma}_{0sj} =  \mathbf{\gamma}_{0sj'}^\top \mathbf{F}_{0s}^\top \mathbf{F}_{0s}\mathbf{\gamma}_{0sj} = \sum_{i=1}^{n_s} (\mathbf{\gamma}_{0sj'}^\top \mathbf{\phi}_{0si})(\mathbf{\gamma}_{0sj}^\top \mathbf{\phi}_{0si}),\\
     &\mathbf{\gamma}_{0sj'}^\top \mathbf{F}_{0s}^\top \mathbf{U}_{0s}^{\perp } \mathbf{U}_{0s}^{\perp \top}  \mathbf{e}_s^{(j)} =  \mathbf{\gamma}_{0sj'}^\top \mathbf{F}_{0s}^\top  \mathbf{e}_s^{(j)} = \sum_{i=1}^{n_s} (\mathbf{\gamma}_{0sj'}^\top \mathbf{\phi}_{0si})e_{sij},\\
      &\mathbf{\gamma}_{0sj'}^\top \mathbf{F}_{0s}^\top \mathbf{U}_{0s}^{\perp } \mathbf{U}_{0s}^{\perp \top} \mathbf{U}_{0s}^\perp \mathbf{U}_{0s}^{\perp \top}\mathbf{M}_{0s} \mathbf{\lambda}_{0j} =  \mathbf{\gamma}_{0sj'}^\top \mathbf{F}_{0s}^\top  \mathbf{M}_{0s} \mathbf{\lambda}_{0j}= \sum_{i=1}^{n_s} (\mathbf{\gamma}_{0sj'}^\top \mathbf{\phi}_{0si})(\mathbf{\lambda}_{0j}^\top \mathbf{\eta}_{0si}).\\
\end{aligned}    
\end{equation*}
Moreover, we have
\begin{equation*}
    \begin{aligned}
        E[(\mathbf{\gamma}_{0sj'}^\top \mathbf{\phi}_{0si})(\mathbf{\gamma}_{0sj}^\top \mathbf{\phi}_{0si})] & = \mathbf{\gamma}_{0sj'}^\top\mathbf{\gamma}_{0sj},\\
        V[(\mathbf{\gamma}_{0sj'}^\top \mathbf{\phi}_{0si})(\mathbf{\gamma}_{0sj}^\top \mathbf{\phi}_{0si})] &= \begin{cases}
    (\mathbf{\gamma}_{0sj}^\top \mathbf{\gamma}_{0sj'})^2 + ||\mathbf{\gamma}_{0sj}||^2|| \mathbf{\gamma}_{0sj'}||^2, \quad &\text{if } j \neq j',\\
        2||\mathbf{\gamma}_{0sj}||^4 &\text{otherwise,} 
    \end{cases}\\
    E[(\mathbf{\gamma}_{0sj'}^\top \mathbf{\phi}_{0si})e_{sij}]& = 0\\
      V[(\mathbf{\gamma}_{0sj'}^\top \mathbf{\phi}_{0si})e_{sij}]& = \begin{cases}
\sigma_{0j}^2 ||\mathbf{\gamma}_{0sj'}||^2 + \sigma_{0j'}^2 ||\mathbf{\gamma}_{0sj}||^2 + \sigma_{0jj'} \mathbf{\gamma}_{0sj}^\top \mathbf{\gamma}_{0sj'} , \quad &\text{if } j \neq j',\\ 
4 \sigma_{0j}^2 ||\mathbf{\gamma}_{0sj}||^2  
 &\text{otherwise,} \\
\end{cases}\\
E[(\mathbf{\gamma}_{0sj'}^\top \mathbf{\phi}_{0si})(\mathbf{\lambda}_{0j}^\top \mathbf{\eta}_{0si})]& = 0\\
      V[(\mathbf{\gamma}_{0sj'}^\top \mathbf{\phi}_{0si})(\mathbf{\lambda}_{0j}^\top \mathbf{\eta}_{0si})]& = \begin{cases}
||\mathbf{\gamma}_{0sj}||^2 ||\mathbf{\lambda}_{0j}||^2 
             + ||\mathbf{\gamma}_{0sj'}||^2 ||\mathbf{\lambda}_{0j}||^2 +  2\mathbf{\gamma}_{0sj}^\top \mathbf{\gamma}_{0sj'}\mathbf{\lambda}_{0j}^\top \mathbf{\lambda}_{0j'}, \quad &\text{if } j \neq j',\\ 
4||\mathbf{\gamma}_{0sj}||^2||\mathbf{\lambda}_{0j}||^2 
 &\text{otherwise,} \\
\end{cases}
    \end{aligned}
\end{equation*}
and \begin{equation*}
    \begin{aligned}
        cov((\mathbf{\gamma}_{0sj'}^\top \mathbf{\phi}_{0si}), (\mathbf{\gamma}_{0sj'}^\top \mathbf{\phi}_{0si})e_{sij}) &= cov((\mathbf{\gamma}_{0sj'}^\top \mathbf{\phi}_{0si}), (\mathbf{\gamma}_{0sj'}^\top \mathbf{\phi}_{0si})(\mathbf{\lambda}_{0j}^\top \mathbf{\eta}_{0si}) ) \\
        &= cov( (\mathbf{\gamma}_{0sj'}^\top \mathbf{\phi}_{0si})e_{sij}, (\mathbf{\gamma}_{0sj'}^\top \mathbf{\phi}_{0si})(\mathbf{\lambda}_{0j}^\top \mathbf{\eta}_{0si})) = 0,
    \end{aligned}
\end{equation*} for all $i=1, \dots, n_s$ and $j,j' = 1, \dots, p$. 
Next, we bound the remaining terms using the fact that, with probability at least $1-o(1)$,
\begin{enumerate}
    \item $||\mathbf{U}_{0s}^{\perp \top} \mathbf{M}_{0s}|| \leq 
    ||\mathbf P_0^{\perp } \mathbf M_0|| \lesssim 1$ by Lemma \ref{lemma:Mt_P_M},
    \item $ \max_{j=1, \dots, p}||  \mathbf{U}_0^{\perp  \top}\mathbf{e}_s^{(j)}|| \lesssim  \max_{j=1, \dots, p}||   \mathbf{P}_0^\perp  \mathbf{e}^{(j)}|| \lesssim \log p$ by Lemma \ref{lemma:Mt_P_M},
    \item $ \max_{j=1, \dots p}|| \mathbf{U}_0^{c \top}\mathbf{e}^{(j)}|| \lesssim \log $ by Lemma \ref{lemma:U_c_e_j},
    \item $|| \mathbf{U}_{0s}^{c\top }\mathbf{F}_{0s} || \lesssim 1$ by Lemma \ref{lemma:U_c_F},
\end{enumerate}
along with $\max_{j, \dots, p}||\mathbf{\lambda}_{0j}|\lesssim  ||{\mathbf{\Lambda}}||_{\infty}\sqrt{k_0} \asymp 1$ by Assumption \ref{assumption:Lambda}. 
With probability at least $1-o(1)$, we have
\begin{equation*}
    \begin{aligned}
       \frac{n_s^{3/2}n^2}{(n_s + \tau_{\Gamma_s}^{-2})^2 (n + \tau_{\Lambda}^{-2})^2} &\big| \mathbf{\lambda}_{0j'}^\top \mathbf{M}_{0s}^\top \mathbf{U}_{0s}^\perp \mathbf{U}_{0s}^{\perp \top} \mathbf{U}_{0s}^\perp \mathbf{U}_{0s}^{\perp \top} \mathbf{U}_{0s}^\perp \mathbf{U}_{0s}^{\perp \top}\mathbf{M}_{0s} \mathbf{\lambda}_{0j}\big|  \\
       & = \frac{n_s^{3/2}n^2}{(n_s + \tau_{\Gamma_s}^{-2})^2 (n + \tau_{\Lambda}^{-2})^2}   \big| \mathbf{\lambda}_{0j'}^\top \mathbf{M}_{0s}^\top \mathbf{U}_{0s}^\perp \mathbf{U}_{0s}^{\perp \top}\mathbf{M}_{0s} \mathbf{\lambda}_{0j}\big|  \\
        & \asymp \frac{1}{\sqrt{n_s}}\big| \mathbf{\lambda}_{0j'}^\top \mathbf{M}_{0s}^\top \mathbf{U}_{0s}^\perp \mathbf{U}_{0s}^{\perp \top}\mathbf{M}_{0s} \mathbf{\lambda}_{0j} \big|  \\
        & \asymp \frac{1}{\sqrt{n_s}} ||\mathbf{\lambda}_{0j}|| ||\mathbf{\lambda}_{0j'}|| ||\mathbf{U}_{0s}^{\perp \top} \mathbf{M}_{0s}||^2  \\ & \asymp \frac{1}{\sqrt{n_s}},\\
         \end{aligned}
\end{equation*}
\begin{equation*}
    \begin{aligned}
      \frac{n_s^{3/2}n}{(n_s + \tau_{\Gamma_s}^{-2})^2 (n + \tau_{\Lambda}^{-2})}& \big| \mathbf{\lambda}_{0j'}^\top \mathbf{M}_{0s}^\top \mathbf{U}_{0s}^\perp \mathbf{U}_{0s}^{\perp \top} \mathbf{U}_{0s}^\perp \mathbf{U}_{0s}^{\perp \top}\mathbf{e}_s^{(j)}\big| \\
      & = \frac{n_s^{3/2}n}{(n_s + \tau_{\Gamma_s}^{-2})^2 (n + \tau_{\Lambda}^{-2})}\big| \mathbf{\lambda}_{0j'}^\top \mathbf{M}_{0s}^\top \mathbf{U}_{0s}^\perp \mathbf{U}_{0s}^{\perp \top}\mathbf{e}_s^{(j)}\big|  \\
        & \asymp \frac{1}{\sqrt{n_s}}\big| \mathbf{\lambda}_{0j'}^\top \mathbf{M}_{0s}^\top \mathbf{U}_{0s}^\perp \mathbf{U}_{0s}^{\perp \top}\mathbf{e}_s^{(j)} \big|  \\
        & \asymp \frac{1}{\sqrt{n_s}}  ||\mathbf{\lambda}_{0j'}|| ||\mathbf{U}_{0s}^{\perp \top} \mathbf{M}_{0s}||  ||\mathbf{U}_{0s}^{\perp \top}\mathbf{e}_s^{(j)}||  
         \\
        & \lesssim \frac{\log p}{\sqrt{n_s}},
    \end{aligned}
\end{equation*}
similarly, $ \frac{n_s^{3/2}}{(n_s + \tau_{\Gamma_s}^{-2})^2}   \frac{n^2}{(n + \tau_{\Lambda}^{-2})^2} \big| \mathbf{\lambda}_{0j'}^\top \mathbf{M}_{0s}^\top \mathbf{U}_{0s}^\perp \mathbf{U}_{0s}^{\perp \top} \mathbf{U}_{0s}^\perp \mathbf{U}_{0s}^{\perp \top}\mathbf{e}_s^{(j')}\big| \lesssim \frac{1}{\sqrt{n_s}}$,
and
\begin{equation*}
    \begin{aligned}
        \frac{n_s^{3/2}}{(n_s + \tau_{\Gamma_s}^{-2})^2}   \big|\mathbf{e}_s^{(j')} \mathbf{U}_{0s}^\perp \mathbf{U}_{0s}^{\perp \top}\mathbf{e}_s^{(j)}\big|  &  \asymp \frac{1}{\sqrt{n_s}}\big| \mathbf{e}_s^{(j')} \mathbf{U}_{0s}^\perp \mathbf{U}_{0s}^{\perp \top}\mathbf{e}_s^{(j)}  \big|  \\
        & \asymp \frac{1}{\sqrt{n_s}}  ||\mathbf{U}_{0s}^{\perp \top}\mathbf{e}_s^{(j)}||  ||\mathbf{U}_{0s}^{\perp \top}\mathbf{e}_s^{(j')}||  
         \\
        & \lesssim \frac{\log^2 p}{\sqrt{n_s}}.
    \end{aligned}
\end{equation*}
Combining all of the above, we get
\begin{equation*}
    \begin{aligned}
        \frac{n_s^{3/2}}{(n_s + \tau_{\Gamma_s}^{-2})^2} &\left(\mathbf{F}_{0s}\mathbf{\gamma}_{0sj} + \mathbf{e}_s^{(j)}  +  \frac{n}{n + \tau_{\Lambda}^{-2}} \mathbf{U}_{0s}^\perp \mathbf{U}_{0s}^{\perp \top}\mathbf{M}_{0s} \mathbf{\lambda}_{0j} \right)^\top \mathbf{U}_{0s}^\perp \mathbf{U}_{0s}^{\perp \top} \\
        &\left(\mathbf{F}_{0s}\mathbf{\gamma}_{0sj'} + \mathbf{e}_s^{(j')}  + \frac{n}{n + \tau_{\Lambda}^{-2}}\mathbf{U}_{0s}^\perp \mathbf{U}_{0s}^{\perp \top}\mathbf{M}_{0s} \mathbf{\lambda}_{0j'} \right)  \Longrightarrow N(\mathbf{\lambda}_{0j}^\top \mathbf{\lambda}_{0j'}, S_{0, sjj'}^2)
    \end{aligned}
\end{equation*}
Finally, we show that all the remaining terms decrease to $0$ at the appropriate rate, using the fact that, with probability at least $1-o(1)$, 
\begin{enumerate}
\item $||\mathbf{M}_0|| \lesssim \sqrt{n}$, $||\mathbf{M}_{0s}|| \asymp \sqrt{n_s}$, and $||\mathbf{F}_{0s}|| \lesssim \sqrt{n_s}$ by Theorem 4.6.1 of \citet{vershynin2018hdp},
 \item $||\mathbf{U}_{0s}^{\perp \top} \mathbf{M}_{0s}|| \leq 
    ||\mathbf P_0^{\perp } \mathbf M_0|| \lesssim 1$, by Lemma \ref{lemma:Mt_P_M},
    \item $ \max_{j=1, \dots, p}||  \mathbf{P}_0^{\perp  }\mathbf{e}_s^{(j)}||  =  \max_{j=1, \dots, p}||  \mathbf{U}_0^{\perp  \top}\mathbf{e}_s^{(j)}|| \lesssim  \max_{j=1, \dots, p}||   \mathbf{P}_0^\perp  \mathbf{e}^{(j)}|| \lesssim \log p$ by Lemma \ref{lemma:Mt_P_M},
    \item $ \max_{j=1, \dots p}|| \mathbf{U}_0^{c \top}\mathbf{e}^{(j)}|| \lesssim \log $ by Lemma \ref{lemma:U_c_e_j},
    \item $|| \mathbf{U}_{0s}^{c\top }\mathbf{F}_{0s} || \lesssim 1$ by Lemma \ref{lemma:U_c_F}.
    \item  $ \max_{j=1, \dots p}|| \mathbf{U}_0^{c \top}\mathbf{e}^{(j)}|| \lesssim \log $ by Lemma \ref{lemma:U_c_e_j},
    \item $ \max_{j=1, \dots, p}||\mathbf{\Delta y}^{ (j)}|| \asymp  C_{\sigma} C_{r, \{n_s\}_{s}^{S}}
\big(\frac{1}{\sqrt{n_{\min}}} + \frac{\sqrt{n_{\max}}}{p}\big) $  by Lemma \ref{lemma:Delta_y_j},
\item $||    \mathbf{U}_s^{\perp} \mathbf{U}_s^{\perp \top}- \mathbf{U}_{0s}^{\perp } \mathbf{U}_{0s}^{\perp \top}   || \lesssim (C_{sr, n_s} C_{sc, p}  
)^2 \big(\frac{1}{p} + \frac{1}{n}$\big) by Proposition \ref{prop:U_s_outer}. 
\end{enumerate}
along with $\max_{j, \dots, p}||\mathbf{\lambda}_{0j}|\lesssim  ||{\mathbf{\Lambda}}||_{\infty}\sqrt{k_0} \asymp 1$ and $\max_{j, \dots, p}||\mathbf{\gamma}_{0sj}|\lesssim  ||{\mathbf{\Gamma}_s}||_{\infty}\sqrt{q_s} \asymp 1$ by Assumptions \ref{assumption:Lambda} and \ref{assumption:Gammas}, respectively. \begin{enumerate}
\item With probability at least $1-o(1)$, we have 
\begin{equation*}
    \begin{aligned}
 & \left| \sqrt{n_s} \frac{n_s}{(n_s + \tau_{\Gamma_s}^{-2})^2} \mathbf{\gamma}_{0sj'}^\top \mathbf{F}_{0s}^\top \mathbf{U}_{0s}^{\perp } \mathbf{U}_{0s}^{\perp \top} \left[-  \frac{\tau_{\Gamma_s}^{-2}}{n_s + \tau_{\Gamma_s}^{-2}} \mathbf{M}_{0s} \mathbf{\lambda}_{0j} + \mathbf{U}_{0s}^c \mathbf{U}_0^{c\top}\left\{(\mathbf{I}_n - \mathbf{P}_0^\perp)  \mathbf{e}^{(j)} + \mathbf{\Delta} y^{(j)}\right\} \right. \right. \\
 &\left. \left.  + \frac{n}{n + \tau_{\Lambda}^{-2}} ( \mathbf{U}_s^{c} \mathbf{U}^{c\top} -  \mathbf{U}_s' \mathbf{U}^{'\top} )\mathbf{\hat y}^{c (j)} \right]  \right| \\
  & \quad \quad  =\left| \sqrt{n_s} \frac{n_s}{(n_s + \tau_{\Gamma_s}^{-2})^2}\mathbf{\gamma}_{0sj'}^\top \mathbf{F}_{0s}^\top \left[-  \frac{\tau_{\Gamma_s}^{-2}}{n_s + \tau_{\Gamma_s}^{-2}} \mathbf{M}_{0s} \mathbf{\lambda}_{0j} + + \mathbf{U}_{0s}^c \mathbf{U}_0^{c\top}\left\{(\mathbf{I}_n - \mathbf{P}_0^\perp)  \mathbf{e}^{(j)} + \mathbf{\Delta} y^{(j)}\right\} \right. \right. \\
 &\left. \left. \quad \quad \quad  + \frac{n}{n + \tau_{\Lambda}^{-2}} ( \mathbf{U}_s^{c} \mathbf{U}^{c\top} -  \mathbf{U}_s' \mathbf{U}^{'\top} )\mathbf{\hat y}^{c (j)} \right]  \right| \\
  &\quad \quad   \lesssim \frac{1}{n_s^{3/2}} ||\mathbf{\gamma}_{0sj'}^\top \mathbf{F}_{0s}^\top \mathbf{M}_{0s} \mathbf{\lambda}_{0j} || + \frac{1}{\sqrt{n_s}} ||\mathbf{\gamma}_{0sj'}^\top \mathbf{F}_{0s}^\top \mathbf{U}_{0s}^c \mathbf{U}_0^{c\top} \mathbf{e}^{(j)}  ||  \\
  &\quad \quad  \quad  +\frac{1}{\sqrt{n_s}} ||\mathbf{\gamma}_{0sj'}^\top \mathbf{F}_{0s}^\top \mathbf{U}_{0s}^c \mathbf{U}_0^{c\top} \mathbf{\Delta} y^{(j)} || + \frac{1}{\sqrt{n_s}} ||\mathbf{\gamma}_{0sj'}^\top \mathbf{F}_{0s}^\top \mathbf{U}_{0s}^c \mathbf{U}_0^{c\top} \mathbf{\Delta} y^{(j)} ||  \\
  &\quad \quad  \quad  +   \frac{1}{\sqrt{n_s}} ||  \mathbf{\gamma}_{0sj'}^\top \mathbf{F}_{0s}^\top  ( \mathbf{U}_s^{c} \mathbf{U}^{c\top} -  \mathbf{U}_s' \mathbf{U}^{'\top} )\mathbf{\hat y}^{c (j)} ||\\
  & \quad \quad   \lesssim \frac{1}{n_s^{3/2}} ||\mathbf{\gamma}_{0sj'}|| || \mathbf{F}_{0s}|| || \mathbf{M}_{0s}|| || \mathbf{\lambda}_{0j} || + \frac{1}{\sqrt{n_s}} ||\mathbf{\gamma}_{0sj'}|| || \mathbf{F}_{0s} \mathbf{U}_{0s}^c|| || \mathbf{U}_0^{c\top} \mathbf{e}^{(j)}  ||  \\
 &\quad \quad  \quad  +\frac{1}{\sqrt{n_s}} ||\mathbf{\gamma}_{0sj'}|| || \mathbf{F}_{0s}  \mathbf{U}_{0s}^c \mathbf{U}_0^{c\top}|| ||\mathbf{\Delta} y^{(j)} || + \frac{1}{\sqrt{n_s}} ||\mathbf{\gamma}_{0sj'}|| || \mathbf{F}_{0s}^\top \mathbf{U}_{0s}^c \mathbf{U}_0^{c\top}|| || \mathbf{\Delta} y^{(j)} ||  \\
  &\quad \quad  \quad  +   \frac{1}{\sqrt{n_s}} ||  \mathbf{\gamma}_{0sj'}^\top|| || \mathbf{F}_{0s}|| || \mathbf{U}_s^{c} \mathbf{U}^{c\top} -  \mathbf{U}_s' \mathbf{U}^{'\top} || || \mathbf{\hat y}^{c (j)} ||\\ 
  & \quad \quad \lesssim    \frac{1}{\sqrt{n_s}} n_s (C_{r, \{n_s\}_{s}^{S}} C_{c, p}   )^2 \big( \frac{1}{n_s} + \frac{1}{p}\big) \asymp (C_{r, \{n_s\}_{s}^{S}} \big(\frac{1}{\sqrt{n_s}} + \frac{\sqrt{n_s}}{p} \big).
    \end{aligned}
\end{equation*}

Similarly, with probability at least $1- o(1)$,
\begin{equation*}
\begin{aligned}
      & \left| \sqrt{n_s} \frac{n_s}{(n_s + \tau_{\Gamma_s}^{-2})^2} \mathbf{\gamma}_{0sj}^\top \mathbf{F}_{0s}^\top \mathbf{U}_{0s}^{\perp } \mathbf{U}_{0s}^{\perp \top}\left[-  \frac{\tau_{\Gamma_s}^{-2}}{n_s + \tau_{\Gamma_s}^{-2}} \mathbf{M}_{0s} \mathbf{\lambda}_{0j'} + \mathbf{U}_{0s}^c \mathbf{U}_0^{c\top}\left\{(\mathbf{I}_n - \mathbf{P}_0^\perp)  \mathbf{e}^{(j')}+\mathbf{\Delta} y^{(j')}\right\} \right. \right. \\
 &\left. \left. \quad \quad + \frac{n}{n + \tau_{\Lambda}^{-2}} ( \mathbf{U}_s^{c} \mathbf{U}^{c\top} -  \mathbf{U}_s' \mathbf{U}^{'\top} )\mathbf{\hat y}^{c (j')}\right]  \right| \lesssim 
    (C_{r, \{n_s\}_{s}^{S}} \big(\frac{1}{\sqrt{n_s}} + \frac{\sqrt{n_s}}{p} \big).
\end{aligned}
\end{equation*}

\item  With probability at least $1-o(1)$, we have 
\begin{equation*}
    \begin{aligned}
         \left| \sqrt{n_s} \frac{n_s}{(n_s + \tau_{\Gamma_s}^{-2})^2}\mathbf{e}_s^{(j) \top} \mathbf{U}_{0s}^{\perp } \mathbf{U}_{0s}^{\perp \top}  \mathbf{e}_s^{(j)} \right| &\asymp 
         \frac{1}{\sqrt{n_s}} ||\mathbf{U}_{0s}^{\perp \top}  \mathbf{e}_s^{(j)}||||\mathbf{U}_{0s}^{\perp \top}  \mathbf{e}_s^{(j')}||\\
         &\lesssim \frac{\log^2 p}{\sqrt{n_s}}.
    \end{aligned}
\end{equation*}

\item  With probability at least $1-o(1)$, we have 
\begin{equation*}
    \begin{aligned}
 & \left|  \frac{n_s^{3/2}}{(n_s + \tau_{\Gamma_s}^{-2})^2}\mathbf{e}_s^{(j) \top} \mathbf{U}_{0s}^{\perp } \mathbf{U}_{0s}^{\perp \top} \left[-  \frac{\tau_{\Gamma_s}^{-2}}{n_s + \tau_{\Gamma_s}^{-2}} \mathbf{M}_{0s} \mathbf{\lambda}_{0j} +  \mathbf{U}_{0s}^c \mathbf{U}_0^{c\top}\left\{(\mathbf{I}_n - \mathbf{P}_0^\perp)  \mathbf{e}^{(j)} + \mathbf{\Delta} y^{(j)}\right\} \right. \right. \\
 &\left. \left.  \quad  \quad  \quad  \quad  \quad  \quad  \quad  \quad  \quad  \quad  \quad+ \frac{n}{n + \tau_{\Lambda}^{-2}} ( \mathbf{U}_s^{c} \mathbf{U}^{c\top} -  \mathbf{U}_s' \mathbf{U}^{'\top} )\mathbf{\hat y}^{c (j)} \right]  \right| \\
  &\quad \quad  \quad \lesssim \frac{1}{n_s^{3/2}} ||\mathbf{e}_s^{(j')\top} \mathbf{U}_{0s}^{\perp } \mathbf{U}_{0s}^{\perp \top} \mathbf{M}_{0s} \mathbf{\lambda}_{0j} || + \frac{1}{\sqrt{n_s}} ||\mathbf{e}_s^{(j')\top} \mathbf{U}_{0s}^{\perp } \mathbf{U}_{0s}^{\perp \top} \mathbf{U}_{0s}^c \mathbf{U}_0^{c\top} \mathbf{e}^{(j)}  || \\
  &\quad \quad  \quad \quad + \frac{1}{\sqrt{n_s}} ||\mathbf{e}_s^{(j')\top} \mathbf{U}_{0s}^{\perp } \mathbf{U}_{0s}^{\perp \top}\mathbf{U}_{0s}^c \mathbf{U}_0^{c\top}  \mathbf{P}_0^\perp  \mathbf{e}^{(j)}|| \\
  &\quad \quad  \quad \quad + \frac{1}{\sqrt{n_s}} ||\mathbf{e}_s^{(j')\top} \mathbf{U}_{0s}^{\perp } \mathbf{U}_{0s}^{\perp \top} \mathbf{U}_{0s}^c \mathbf{U}_0^{c\top} \mathbf{\Delta} y^{(j)} || \\
  &\quad \quad  \quad \quad +  \frac{1}{\sqrt{n_s}}\frac{n}{n + \tau_{\Lambda}^{-2}} ||\mathbf{e}_s^{(j')\top} \mathbf{U}_{0s}^{\perp } \mathbf{U}_{0s}^{\perp \top} ( \mathbf{U}_s^{c} \mathbf{U}^{c\top} -  \mathbf{U}_s' \mathbf{U}^{'\top} )\mathbf{\hat y}^{c (j)}  ||\\
  &\quad \quad  \quad \lesssim \frac{1}{n_s^{3/2}} ||\mathbf{e}_s^{(j')\top} \mathbf{U}_{0s}^{\perp }|| || \mathbf{U}_{0s}^{\perp \top} \mathbf{M}_{0s}|| ||\mathbf{\lambda}_{0j} || \\
  &\quad \quad  \quad \quad  + \frac{1}{\sqrt{n_s}} ||\mathbf{e}_s^{(j')\top} \mathbf{U}_{0s}^{\perp }|| || \mathbf{U}_{0s}^{\perp \top}|| || \mathbf{U}_{0s}^c|| || \mathbf{U}_0^{c\top} \mathbf{e}^{(j)}  ||\\
  &\quad \quad  \quad \quad+  \frac{1}{\sqrt{n_s}} ||\mathbf{e}_s^{(j')\top} \mathbf{U}_{0s}^{\perp }|| || \mathbf{U}_{0s}^{\perp \top}|| ||\mathbf{U}_{0s}^c \mathbf{U}_0^{c\top}  || ||\mathbf{P}_0^\perp (\mathbf{M}_{0} \mathbf{\lambda}_{0j} + \mathbf{e}^{(j)})|| \\
  &\quad \quad  \quad \quad + \frac{1}{\sqrt{n_s}} ||\mathbf{e}_s^{(j')\top} \mathbf{U}_{0s}^{\perp }|| || \mathbf{U}_{0s}^{\perp \top}|| || \mathbf{U}_{0s}^c \mathbf{U}_0^{c\top}|| || \mathbf{\Delta} y^{(j)} ||  \\
  & \quad \quad  \quad \quad+   \frac{1}{\sqrt{n_s}} ||\mathbf{e}_s^{(j')\top} \mathbf{U}_{0s}^{\perp } ||  ||\mathbf{U}_{0s}^{\perp} || || \mathbf{U}_s^{c} \mathbf{U}^{c\top} -  \mathbf{U}_s' \mathbf{U}^{'\top} || ||\mathbf{\hat y}^{c (j)}  ||\\
  & \quad \quad  \quad \lesssim \frac{\log^2 p}{\sqrt{n_s}}.
    \end{aligned}
\end{equation*}

Similarly, with probability at least $1- o(1)$, \begin{equation*}
    \begin{aligned}
 & \left|  \frac{n_s^{3/2}}{(n_s + \tau_{\Gamma_s}^{-2})^2} \mathbf{e}_s^{(j) \top} \mathbf{U}_{0s}^{\perp } \mathbf{U}_{0s}^{\perp \top} \left[-  \frac{\tau_{\Gamma_s}^{-2}}{n_s + \tau_{\Gamma_s}^{-2}} \mathbf{M}_{0s} \mathbf{\lambda}_{0j'} +  \mathbf{U}_{0s}^c \mathbf{U}_0^{c\top}\left\{(\mathbf{I}_n - \mathbf{P}_0^\perp)  \mathbf{e}^{(j')} +\mathbf{\Delta} y^{(j')}\right\} \right. \right. \\
 &\left. \left. \quad  \quad  \quad  \quad  \quad  \quad  \quad  \quad  \quad  \quad \quad+ \frac{n}{n + \tau_{\Lambda}^{-2}} ( \mathbf{U}_s^{c} \mathbf{U}^{c\top} -  \mathbf{U}_s' \mathbf{U}^{'\top} )\mathbf{\hat y}^{c (j')}\right]  \right|  \lesssim \frac{\log^2 p}{\sqrt{n_s}}.
    \end{aligned}
\end{equation*}
.

\item  With probability at least $1-o(1)$, we have 
 \begin{equation*}
    \begin{aligned}
 & \frac{n_s^{3/2}}{(n_s + \tau_{\Gamma_s}^{-2})^2} \bigg| \big[-  \frac{\tau_{\Gamma_s}^{-2}}{n_s + \tau_{\Gamma_s}^{-2}} \mathbf{M}_{0s} \mathbf{\lambda}_{0j'} +  \mathbf{U}_{0s}^c \mathbf{U}_0^{c\top}\left\{(\mathbf{I}_n - \mathbf{P}_0^\perp)  \mathbf{e}^{(j')} +\mathbf{\Delta} y^{(j')}\right\} \\
 & \quad \quad \quad \quad \quad \quad  + \frac{n}{n + \tau_{\Lambda}^{-2}} ( \mathbf{U}_s^{c} \mathbf{U}^{c\top} -  \mathbf{U}_s' \mathbf{U}^{'\top} )\mathbf{\hat y}^{c (j')}\big]^\top \mathbf{U}_{0s}^{\perp } \mathbf{U}_{0s}^{\perp \top} \\ &\quad \quad \quad \quad \quad  \big[-  \frac{\tau_{\Gamma_s}^{-2}}{n_s + \tau_{\Gamma_s}^{-2}} \mathbf{M}_{0s} \mathbf{\lambda}_{0j} +  \mathbf{U}_{0s}^c \mathbf{U}_0^{c\top}\left\{(\mathbf{I}_n - \mathbf{P}_0^\perp)  \mathbf{e}^{(j)}+ \mathbf{\Delta} y^{(j)}\right\} \\
 & \quad \quad \quad \quad \quad \quad+ \frac{n}{n + \tau_{\Lambda}^{-2}} ( \mathbf{U}_s^{c} \mathbf{U}^{c\top} -  \mathbf{U}_s' \mathbf{U}^{'\top} )\mathbf{\hat y}^{c (j)} \big] \bigg| \\ 
  &\quad    \lesssim \frac{1}{\sqrt{n_s}} \big\{\frac{1}{n_s}|| \mathbf{M}_{0s}|| || \mathbf{\lambda}_{0j'}|| +  ||\mathbf{U}_{0s}^c|| || \mathbf{U}_0^{c\top} \mathbf{e}^{(j')}|| + ||\mathbf{U}_{0s}^c \mathbf{U}_0^{c\top}|| ( || \mathbf{P}_0^\perp  \mathbf{e}^{(j')}||  + ||\Delta y^{(j')}||)\\
  & \quad \quad  \quad \quad\quad + ||\mathbf{U}_s^{c} \mathbf{U}^{c\top} -  \mathbf{U}_s' \mathbf{U}^{'\top}|| ||\mathbf{\hat y}^{c (j')}||  \big\}  ||  \mathbf{U}_{0s}^{\perp } \mathbf{U}_{0s}^{\perp \top} ||\big\{ \frac{1}{n_s}|| \mathbf{M}_{0s}|| || \mathbf{\lambda}_{0j}||\\
  & \quad \quad  \quad \quad \quad  +  ||\mathbf{U}_{0s}^c|| || \mathbf{U}_0^{c\top} \mathbf{e}^{(j')}||+ ||\mathbf{U}_{0s}^c \mathbf{U}_0^{c\top}|| ( || \mathbf{P}_0^\perp  \mathbf{e}^{(j')}||+ ||\mathbf{\Delta} y^{(j)}||) \\
  & \quad \quad  \quad \quad \quad   + ||\mathbf{U}_s^{c} \mathbf{U}^{c\top} -  \mathbf{U}_s' \mathbf{U}^{'\top}|| ||\mathbf{\hat y}^{c (j)}||\big\} \\
  & \quad \quad  \lesssim  \frac{\log^2 p}{\sqrt{n_s}} + (C_{r, \{n_s\}_{s}^{S}} C_{c, p}   )^2 \big( \frac{1}{\sqrt{n_s}} + \frac{\sqrt{n_s}}{p} \big).
  \\
    \end{aligned}
\end{equation*}

\item  With probability at least $1-o(1)$, we have 
\begin{equation*}
    \begin{aligned}
   & \frac{n_s^{3/2}}{(n_s + \tau_{\Gamma_s}^{-2})^2} \bigg| \big[\mathbf{F}_{0s}\mathbf{\gamma}_{0sj'} + \mathbf{e}_s^{(j')}  + \frac{1}{n + \tau_{\Lambda}^{-2}}\left(n\mathbf{U}_{0s}^\perp \mathbf{U}_{0s}^{\perp \top} - \tau_{\Lambda}^{-2} I \right)\mathbf{M}_{0} \mathbf{\lambda}_{0j'} \\
    &\quad \quad \quad \quad \quad+ \mathbf{U}_{0s}^c \mathbf{U}_0^{c\top}\left\{(\mathbf{I}_n - \mathbf{P}_0^\perp)  \mathbf{e}^{(j')} +\mathbf{\Delta} y^{(j')}\right\}\big]^\top\left( \mathbf{U}_s^{\perp} \mathbf{U}_s^{\perp \top}- \mathbf{U}_{0s}^{\perp } \mathbf{U}_{0s}^{\perp \top}\right) \\
           & \quad \quad \quad\quad \Big[\mathbf{F}_{0s}\mathbf{\gamma}_{0sj} + \mathbf{e}_s^{(j)}  + \frac{1}{n + \tau_{\Lambda}^{-2}}\left(n\mathbf{U}_{0s}^\perp \mathbf{U}_{0s}^{\perp \top} - \tau_{\Lambda}^{-2} I \right)\mathbf{M}_{0} \mathbf{\lambda}_{0j}  \\ 
            & \quad \quad \quad \quad \quad  + \mathbf{U}_{0s}^c \mathbf{U}_0^{c\top}\left\{(\mathbf{I}_n - \mathbf{P}_0^\perp)  \mathbf{e}^{(j)} + \mathbf{\Delta} y^{(j)}\right\} \Big]  \bigg|\\
           & \quad \lesssim \frac{1}{\sqrt{n_s}} \big(|| \mathbf{F}_{0s}|| ||\mathbf{\gamma}_{0sj'}|| + ||\mathbf{e}_s^{(j')}||  + ||\mathbf{U}_{0s}|| || \mathbf{U}_{0s}^{\perp \top}\mathbf{M}_{0} || + \frac{1}{n}|| \mathbf{M}_0|| ||\mathbf{\lambda}_{0j'}|| \\
           &\quad \quad \quad \quad \quad + ||\mathbf{U}_{0s}^c|| ||\mathbf{U}_0^{c\top}\mathbf{e}^{(j')} ||   +  ||\mathbf{U}_{0s}^c \mathbf{U}_0^{c\top}|| ||\mathbf{P}_0^\perp  \mathbf{e}^{(j')}||  + ||\Delta y^{(j')} || \big)\\
           &\quad \quad \quad \quad ||    \mathbf{U}_s^{\perp} \mathbf{U}_s^{\perp \top}- \mathbf{U}_{0s}^{\perp } \mathbf{U}_{0s}^{\perp \top}   || \\ & \quad \quad \quad \quad \big(|| \mathbf{F}_{0s}|| ||\mathbf{\gamma}_{0sj}|| + ||\mathbf{e}_s^{(j)}||   + ||\mathbf{U}_{0s}|| || \mathbf{U}_{0s}^{\perp \top}\mathbf{M}_{0} || + \frac{1}{n}|| \mathbf{M}_0|| ||\mathbf{\lambda}_{0j}|| \\
           & \quad \quad \quad \quad \quad   + ||\mathbf{U}_{0s}^c|| ||\mathbf{U}_0^{c\top}\mathbf{e}^{(j)} || +  ||\mathbf{U}_{0s}^c \mathbf{U}_0^{c\top}|| ||\mathbf{P}_0^\perp  \mathbf{e}^{(j)}||  + ||\mathbf{\Delta} y^{(j)} || \big)\\
  &  \quad\lesssim (C_{sr, n_s} C_{sc, p}  
)^2 \big(\frac{1}{\sqrt{n_s}} + \frac{\sqrt{n_s}}{p}\big). 
    \end{aligned}
\end{equation*}
\end{enumerate}
 To complete the proof, we apply Lemma F.2 from \citet{fable}.
\end{proof}

\begin{proof}[Proof of Theorem \ref{thm:bvm_Lambda_outer}]
First, we show the result for $ \mathbf{\tilde \lambda}_j^\top  \mathbf{\tilde \lambda}_{j'}$.  A sample $\mathbf{\tilde \lambda}_j$ from $\tilde \Pi$ for $ \mathbf{\lambda}_j$ is given by
\begin{equation*}
   \mathbf{\tilde \lambda}_j= \mathbf{\mu}_{\lambda_j} +  \frac{\rho_{\Lambda} \tilde \sigma_j}{\sqrt{n + \tau_{\Lambda}^{-2}}} \varepsilon_{\mathbf{\lambda}_{j}} , \quad        \varepsilon_{\mathbf{\lambda}_{j}} \sim N_{k_0}\left(0, \mathbf{I}_{k_0}\right), \quad (j= 1, \dots, p).
\end{equation*}
Thus, for $j \neq j'$,
\begin{equation*}
     \mathbf{\tilde \lambda}_j^\top  \mathbf{\tilde \lambda}_{j'} = \mathbf{\mu}_{\lambda_j}^\top \mathbf{\mu}_{\lambda_{j'}}  + \frac{\rho_{\Lambda} \tilde \sigma_{j'}}{\sqrt{n + \tau_{\Lambda}^{-2}}} \mathbf{\mu}_{\lambda_j}^\top \varepsilon_{\mathbf{\lambda}_{j'}}  + \frac{\rho_{\Lambda} \tilde \sigma_{j}}{\sqrt{n + \tau_{\Lambda}^{-2}} }\mathbf{\mu}_{\lambda_{j'}}^\top \varepsilon_{\mathbf{\lambda}_{j}} + \rho_{\Lambda}^2 \frac{\tilde \sigma_{j}\tilde \sigma_{j'}}{n + \tau_{\Lambda}^{-2}} \varepsilon_{\mathbf{\lambda}_{j}}^\top \varepsilon_{\mathbf{\lambda}_{j'}}.
\end{equation*}
We examine each term separately.
Let $$l^2(\rho_{\Lambda}, \tilde \sigma_j, \mathbf{\mu}_{\lambda_j}, \tilde \sigma_{j'}, \mathbf{\mu}_{\lambda_{j'}} ) = \frac{\rho_{\Lambda}^2n}{n + \tau_{\Lambda}^{-2}}\left(\tilde \sigma_{j'}^2 || \mathbf{\mu}_{\lambda_j}||^2  +\tilde \sigma_{j}^2 || \mathbf{\mu}_{\lambda_{j'}}||^2 \right),$$
and note 
$$
\mathbf{\tilde \lambda}_j^\top  \mathbf{\tilde \lambda}_{j'} = \mathbf{\mu}_{\lambda_j}^\top \mathbf{\mu}_{\lambda_{j'}}  +  l_{0, jj'}(\rho_{\Lambda}) z_{jj'} + \{l(\rho_{\Lambda}, \tilde \sigma_j, \mathbf{\mu}_{\lambda_j}, \tilde \sigma_{j'}, \mathbf{\mu}_{\lambda_{j'}} ) - l_{0, jj'}(\rho_{\Lambda})\}z_{jj'} +    \rho_{\Lambda}^2 \frac{\tilde \sigma_{j}\tilde \sigma_{j'}}{n + \tau_{\Lambda}^{-2}} \varepsilon_{\mathbf{\lambda}_{j}}^\top \varepsilon_{\mathbf{\lambda}_{j'}},
$$
where $z_{jj'} \sim N(0,1)$. 
\begin{enumerate}
    \item With probability at least $1 - o(1)$, we have 
    \begin{equation*}
        \sqrt{n} \rho_{\Lambda}^2 \frac{\tilde \sigma_{j}\tilde \sigma_{j'}}{n + \tau_{\Lambda}^{-2}}| \varepsilon_{\mathbf{\lambda}_{j}}^\top \varepsilon_{\mathbf{\lambda}_{j'}} |\lesssim  \frac{1}{\sqrt{n}},
    \end{equation*}
    since $\tilde \sigma_u = \mathcal O_{pr} (1)$ for $u=j,j'$ by Lemma \ref{lemma:sigma_j_tilde} and $||\varepsilon_{\mathbf{\lambda}_{j}}|| \asymp || \varepsilon_{\mathbf{\lambda}_{j'}}|| \asymp 1$ with $\tilde \Pi$ probability at least $1 - o(1)$ by Corollary 5.35 of \citet{vershynin_12}.
    \item Since $||\mathbf{\mu}_{\lambda_u}||^2 \to ||\mathbf{\lambda}_{0u}||$ and $\tilde \sigma_u^2 \to \sigma_{0u}^2$ in probability for $u=j,j'$ by Lemma \ref{lemma:convergence}, and $\frac{n}{n + \tau_{\Lambda}^{-2}} \asymp 1$ as $n\to\infty$, we have $l(\rho_{\Lambda}, \tilde \sigma_j, \mathbf{\mu}_{\lambda_j}, \tilde \sigma_{j'}, \mathbf{\mu}_{\lambda_{j'}} ) - l_{0, jj'}(\rho_{\Lambda}) \to 0$ in probability for any finite $\rho_{\Lambda}$.
\end{enumerate}
For $j=j'$, we have 
\begin{equation*}
         ||\mathbf{\tilde \lambda}_j||^2  = ||\mathbf{\mu}_{\lambda_j}||^2 + 2 \frac{\rho_{\Lambda} \tilde \sigma_{j'}}{\sqrt{n + \tau_{\Lambda}^{-2}}} \mathbf{\mu}_{\lambda_j}^\top  \varepsilon_{\mathbf{\lambda}_{j}}+  \frac{ \rho_{\Lambda}^2 \tilde \sigma_{j}^2}{n + \tau_{\Lambda}^{-2}} || \varepsilon_{\mathbf{\lambda}_{j'}}||^2.  
\end{equation*}
Similarly as above, we have $2 \frac{\rho_{\Lambda} \tilde \sigma_{j'} \sqrt{n}}{\sqrt{n + \tau_{\Lambda}^{-2}}} \mathbf{\mu}_{\lambda_j}^\top  \varepsilon_{\mathbf{\lambda}_{j}}\Longrightarrow N(0, 4  \rho_{\Lambda}^2 \sigma_{0j}^2 ||\mathbf{\lambda}_{0j}||^2)$ and $\rho_{\Lambda}^2 \frac{\tilde \sigma_{j}^2 \sqrt{n}}{n + \tau_{\Lambda}^{-2}} || \varepsilon_{\mathbf{\lambda}_{j'}}||^2 \lesssim \frac{1}{\sqrt{n}}$ with probability at least $1 - o(1)$ for all finite $\rho_{\Lambda}$.

Consider the following representation of a sample for $\mathbf{\gamma}_{sj}$,
\begin{equation*}
    \begin{aligned}
        \mathbf{\tilde \gamma}_{sj} &= \mathbf{\mu}_{\gamma_{sj}} + \frac{\rho_{\Lambda} \tilde \sigma_j}{\sqrt{n + \tau_{\Lambda}^{-2}} (n_s + \tau_{\Gamma_s}^{-2})} \mathbf{\hat  F_s}^\top\mathbf{\hat M}_s\varepsilon_{\mathbf{\lambda}_j}  +   \frac{\rho_{\Gamma_s} \tilde \sigma_j}{ \sqrt{ n_s + \tau_{\Gamma_s}^{-2}}}  \varepsilon_{\mathbf{\gamma}_{sj}}, \\
        \varepsilon_{\mathbf{\lambda}_j} &\sim N_{k_0}(0, \mathbf{I}_{k_0}), \quad \varepsilon_{\mathbf{\gamma}_{sj}} \sim N_{q_s}(0, \mathbf{I}_{q_s}).
    \end{aligned}
\end{equation*}
Following the same steps as in the Proof for Theorem \ref{thm:bvm_Lambda_outer} and using Lemma \ref{lemma:convergence_Gamma_s}, we can show 
\begin{equation*}
    \sqrt{n} \bigg(\frac{\rho_{\Gamma_s} \tilde \sigma_j}{ \sqrt{ n_s + \tau_{\Gamma_s}^{-2}}}  \varepsilon_{\mathbf{\gamma}_{sj}}^\top \mathbf{\mu}_{\gamma_{sj'}} + \frac{\rho_{\Gamma_s} \tilde \sigma_{j'}}{ \sqrt{ n_s + \tau_{\Gamma_s}^{-2}}}  \varepsilon_{\mathbf{\gamma}_{sj'}}^\top \mathbf{\mu}_{\gamma_{sj}}  \bigg) \Longrightarrow N(0, l_{0, sjj'}^2(\rho_{\Gamma_s}) ).
\end{equation*}
We complete the proof by showing the remaining terms in $\sqrt{n} \mathbf{\tilde \gamma}_{sj}^\top \mathbf{\tilde \gamma}_{sj'} $ converge to $0$ at the appropriate rate. 
\begin{enumerate}
    
    \item With probability at least $1 - o(1)$, we have
    \begin{equation*}
        \begin{aligned}
         \big| \frac{ \sqrt{n_s} \rho_{\Lambda} \tilde \sigma_j }{\sqrt{n + \tau_{\Lambda}^{-2}}  (n_s + \tau_{\Gamma_s}^{-2})} \mathbf{\mu}_{\gamma_{sj}}^\top   \mathbf{\hat  F_s}^\top\mathbf{\hat M}_s\varepsilon_{\mathbf{\lambda}_{j'}} \big| \asymp \big| \mathbf{\mu}_{\gamma_{sj}}^\top    \mathbf{U}_s^{\perp \top} \mathbf{U}_s^{c}\varepsilon_{\mathbf{\lambda}_{j'}} \big| \lesssim (C_{r, \{n_s\}_{s}^{S}} C_{c, p})^2 \big(\frac{1}{n_s} + \frac{1}{p}\big),
        \end{aligned}
    \end{equation*}
    since $||\mathbf{U}_s^{\perp \top} \mathbf{U}_s^{c} || \lesssim (C_{r, \{n_s\}_{s}^{S}} C_{c, p})^2 \big(\frac{1}{n_s} + \frac{1}{p}\big)$ with probability at least $1 - o(1)$ as we show in the following. In particular, from the proof of Theorem \ref{thm:clt_mu_Lambda_outer}, we have $||\mathbf{U}^c \mathbf{U}^{c \top} - \mathbf{U}' \mathbf{U}^{'\top} || \lesssim (C_{r, \{n_s\}_{s}^{S}} C_{c, p})^2 \big(\frac{1}{n} + \frac{1}{p}\big)$ with probability at least $1 - o(1)$, where $\mathbf{U}'$ is the matrix of left singular vectors of $(\mathbf{I}_n - \mathbf{P}_0^\perp) \mathbf{M}_0 {\mathbf{\Lambda}}_0^\top$. Then, by Davis-Kahan theorem \citep{davis_kahan} we have
    $\min_{\mathbf{R}: \mathbf{R}^\top \mathbf{R} = \mathbf{I}_{k_0}}||\mathbf{U}^c  - \mathbf{U}' \mathbf{R} || = ||\mathbf{U}^c  - \mathbf{U}' \mathbf{\hat R} || \lesssim (C_{r, \{n_s\}_{s}^{S}} C_{c, p})^2 \big(\frac{1}{n} + \frac{1}{p}\big)$, with probability at least $1 - o(1)$, where $\mathbf{\hat R}$ is the orthogonal matrix achieving the minimum of the quantity on the left hand side. Similarly, we have  $\min_{\mathbf{R}_s: \mathbf{R}_s^\top \mathbf{R}_s = \mathbf{I}_{q_s}}||\mathbf{U}_s^\perp  - \mathbf{U}_{0s}^\perp R || = ||\mathbf{U}_s^\perp  - \mathbf{U}_{0s}^\perp  \mathbf{\hat R}_s || \lesssim (C_{r, \{n_s\}_{s}^{S}} C_{c, p})^2 \big(\frac{1}{n_s} + \frac{1}{p}\big)$, with probability at least $1 - o(1)$. Therefore, with probability at least $1 - o(1)$, we have
    $||\mathbf{U}_s^{\perp \top} \mathbf{U}_s^{c}|| = ||(\mathbf{U}_{0s}^\perp  \mathbf{\hat R}_s  + \mathbf{U}_s^{\perp} - \mathbf{U}_{0s}^\perp  \mathbf{\hat R}_s)^\top (\mathbf{U}' \mathbf{\hat R} +\mathbf{U}^c - \mathbf{U}' \mathbf{\hat R})|| \lesssim(C_{r, \{n_s\}_{s}^{S}} C_{c, p})^2 \big(\frac{1}{n_s} + \frac{1}{p}\big)$, since $\mathbf{U}_s^{\perp \top} \mathbf{U}' = 0$. Analogously, with probability at least $1 - o(1)$, we have
\begin{equation*}
        \begin{aligned}
         \left| \frac{ \sqrt{n_s} \rho_{\Lambda} \tilde \sigma_{j'} }{\sqrt{n + \tau_{\Lambda}^{-2}}  (n_s + \tau_{\Gamma_s}^{-2})} \mathbf{\mu}_{\gamma_{sj}}^\top   \mathbf{\hat  F_s}^\top\mathbf{\hat M}_s \mathbf{\eta}_{j} \right| \asymp \big| \mathbf{\mu}_{\gamma_{sj}}^\top    \mathbf{U}_s^{\perp \top} \mathbf{U}_s^{c}\varepsilon_{\mathbf{\lambda}_{j'}} \big| \lesssim (C_{r, \{n_s\}_{s}^{S}} C_{c, p})^2 \big(\frac{1}{n_s} + \frac{1}{p}\big).
        \end{aligned}
    \end{equation*}
    
    \item  With probability at least $1 - o(1)$, we have
    \begin{equation*}
        \begin{aligned}
         \left| \frac{ \sqrt{n_s} \rho_{\Lambda}^2 \tilde \sigma_j \tilde \sigma_{j'}}{(n + \tau_{\Lambda}^{-2})  (n_s + \tau_{\Gamma_s}^{-2})^2} \varepsilon_{\mathbf{\lambda}_j}^\top\mathbf{\hat M}_s^\top \mathbf{\hat  F_s}   \mathbf{\hat  F_s}^\top\mathbf{\hat M}_s\varepsilon_{\mathbf{\lambda}_{j'}} \right| \asymp \frac{1}{\sqrt{n_s}}\big|   \varepsilon_{\mathbf{\lambda}_j}^\top \mathbf{U}_s^{c\top} \mathbf U_c^\perp \mathbf U_c^{\perp \top} \mathbf{U}_s^{c}\varepsilon_{\mathbf{\lambda}_{j'}}\big|  \lesssim \frac{1}{\sqrt{n_s}}.
        \end{aligned}
    \end{equation*}
    \item With probability at least $1 - o(1)$, we have
    \begin{equation*}
        \begin{aligned}
           \left|  \frac{\rho_{\Lambda}\rho_{\Gamma_s} \tilde \sigma_j \tilde \sigma_{j'} \sqrt{n_s}}{\sqrt{n + \tau_{\Lambda}^{-2}} (n_s + \tau_{\Gamma_s}^{-2})^{3/2}}\varepsilon_j^\top\mathbf{\hat M}_s^\top \mathbf{\hat  F_s} \varepsilon_{\mathbf{\gamma}_{sj'}}  \right| & \lesssim  \frac{1}{\sqrt{n_s}}  \big|   \varepsilon_{\mathbf{\lambda}_j}^\top \mathbf{U}_s^{c\top} \mathbf U_c^{\perp}  \varepsilon_{\mathbf{\gamma}_{sj'}}\big|  \lesssim\frac{1}{\sqrt{n_s}},
        \end{aligned}
    \end{equation*}
and
\begin{equation*}
   \left|  \frac{\rho_{\Lambda} \rho_{\Gamma_s} \tilde \sigma_j \tilde \sigma_{j'} \sqrt{n_s}}{\sqrt{n + \tau_{\Lambda}^{-2}} (n_s + \tau_{\Gamma_s}^{-2})^{3/2}}\varepsilon_{j'}^\top\mathbf{\hat M}_s^\top \mathbf{\hat  F_s} \varepsilon_{\mathbf{\gamma}_{sj}}  \right|\lesssim \frac{1}{\sqrt{n_s}}.
\end{equation*}
    \item With probability at least $1 - o(1)$, we have
    \begin{equation*}
        \begin{aligned}
           \left| \frac{\rho_{\Gamma_s}^2 \tilde \sigma_j \tilde \sigma_{j'}  \sqrt{n_s}}{n_s + \tau_{\Gamma_s}^{-2}}  \varepsilon_{\mathbf{\gamma}_{sj}}^\top  \varepsilon_{\mathbf{\gamma}_{sj'}}    \right| & \lesssim    \frac{1}{\sqrt{n_s}}.
        \end{aligned}
    \end{equation*}
\end{enumerate}
\end{proof}

\subsection{Consistency results under moment assumptions}

{In this section, we extend the results reported in Theorems \ref{thm:factors_procrustes_error} and \ref{thm:posterior_contraction_Lambda_outer}, under more general conditions. In particular, instead of the distributional assumptions on the latent factors and residual noises, we make the following moment based assumptions. \begin{assumption}[Mean-zero noise]\label{ass:bai-mean}
For the idiosyncratic noise, we have 
\begin{equation*}
E(\mathbf{e}_{si})=0
\quad (i=1,\dots,n_s; s=1, \dots, S).
\end{equation*}
\end{assumption}
\begin{assumption}[Weak cross-sectional dependence]\label{ass:bai-cross}
The total covariance across the $p$ columns of the noise is uniformly bounded:
\begin{equation*}
\frac{1}{p}
\sum_{j=1}^p \sum_{j'=1}^p
\big|\operatorname{Cov}({e}_{sij}, \mathbf{e}_{sij'})\big|
\;\le\; C_1, \quad (i=1,\dots,n_s; s=1, \dots, S),
\end{equation*}
for some universal constant $C_1>0$.
\end{assumption}
\begin{assumption}[Weak across-units dependence]\label{ass:bai-time}
The total covariance of the noise across the samples in each study is uniformly bounded:
\begin{equation*}
\frac{1}{n_s}
\sum_{i=1}^{n_s} \sum_{i'=1}^{n_s}
\big|\operatorname{Cov}(e_{sij}, e_{si'j})\big|
\;\le\; C_2, \quad (j=1,\dots,p; s=1, \dots, S),
\end{equation*}
for some universal constant $C_2>0$.
\end{assumption}}

{
\begin{assumption}[Norm of residual errors]\label{ass:bai-norm_error}
    We assume that, with probability at least $1-o(1)$, 
    \begin{equation*}
        \begin{aligned}
            &\frac{1}{\sqrt{n_s} p}||\mathbf E_{s} \mathbf{e}_{si} ||_F \lesssim \frac{1}{\sqrt{n_s}} + \frac{1}{\sqrt{p}}, \quad (i=1, \dots, n_s)\\
             &\frac{1}{n_s \sqrt{p}}||\mathbf{e}_{s}^{(j)\top}\mathbf E_{s}||_F \lesssim \frac{1}{\sqrt{n_s}} + \frac{1}{\sqrt{p}}, \quad (j=1, \dots, p),
        \end{aligned} \quad (s=1, \dots, S).
    \end{equation*}
\end{assumption}
\begin{assumption}\label{ass:bai-error_cross}
 We assume 
\begin{equation*}
\begin{aligned}
    \frac{1}{\sqrt{p}} E \bigg( ||\sum_{j=1}^p \lambda_{0j} \mathbf e_{si}^{(j)}||^2 + ||\sum_{j=1}^p \gamma_{0sj} \mathbf e_{si}^{(j)}||^2
 \bigg) &\le M\\
|| \mathbf e_{si}^\top \mathbf E_s^\top \mathbf M_{0s}|| + || \mathbf e_{si}^\top \mathbf E_s^\top \mathbf F_{0s}||&\lesssim n_s + p
\end{aligned}
\qquad  (i=1, \dots, n_s; s=1, \dots, S),
\end{equation*}
and
\begin{equation*}
    \begin{aligned}
    \frac{1}{\sqrt{p}} E \bigg( ||\sum_{i=1}^{n_s} \eta_{si} \mathbf e_{si}^{(j)}||^2 + ||\sum_{i=1}^{n_s} \phi_{si} \mathbf e_{si}^{(j)}||^2
 \bigg) &\le M\\
|| \mathbf e_{s}^{(j)\top} \mathbf E_s \mathbf \Lambda_0|| + || \mathbf e_{s}^{(j)\top} \mathbf E_s \mathbf \Gamma_{0s}||&\lesssim n_s + p
\end{aligned}
\qquad  (j=1, \dots, p; s=1, \dots, S).
\end{equation*}
\end{assumption}
}

{
\begin{assumption}[Near orthogonality between latent factors]\label{ass:lf_projection}
We assume the latent factors $\mathbf{F}_{0s}$ and $\mathbf{M}_{0s}$ to be nearly orthogonal, that is, with probability at least $1-o(1)$, we have
    \begin{equation*}
       || \mathbf{P}_{\mathbf{F}_{0s}}\mathbf{M}_{0s} || \lesssim 1, \ \quad (s=1, \dots, S),
    \end{equation*}
    where $\mathbf{P}_{\mathbf{F}_{0s}} = \mathbf{F}_{0s} \big(\mathbf{F}_{0s}^{\top}\mathbf{F}_{0s}\big)^{-1}\mathbf{F}_{0s}^\top$ is the orthogonal projection matrix onto the column space of $\mathbf{F}_{0s}$.
\end{assumption}
\begin{assumption}[Singular values of latent factor matrices]\label{ass:sv_lf}
We assume the following conditions for the singular values of the latent factor matrices. With probability at least $1-o(1)$, we have
\begin{equation*}
    \big| s_l(\mathbf M_0) - \sqrt{n}\big| \lesssim 1, \quad(l=1, \dots, k_0), \quad   \big| s_l(\mathbf F_{0s}) - \sqrt{n_s}\big| \lesssim 1, \quad(l=1, \dots, q_s; s=1, \dots, S).
\end{equation*}
Moreover, we assume that with probability at least $1-o(1)$, 
\begin{equation*}
    \frac{1}{n}\big| \big| \mathbf M_0^\top \mathbf M_0 - n \mathbf I_{k_0} \big| \big|  \lesssim  \frac{\sqrt{\log n}}{\sqrt{n}}, \quad    \frac{1}{n_s}\big| \big| \mathbf F_{0s}^\top \mathbf F_{0s} - n_s \mathbf I_{q_s} \big| \big| \lesssim \frac{\sqrt{\log n_s}}{\sqrt{n_s}}, \quad(s=1, \dots, S).
\end{equation*}
\end{assumption}
}
{
\begin{theorem}[Recovery of latent factors under moment-based assumptions]\label{thm:factors_procrustes_error_additional}
   Suppose Assumptions \ref{assumption:model}, \ref{assumption:li}--\ref{assumption:sv_A} and \ref{ass:bai-mean}--\ref{ass:sv_lf} hold and $n_s = \mathcal O(n_{\min}^2)$, where $n_{\min} = \min_{s=1, \dots, S} n_s$, for all $s=1, \dots, S$, then, as $n_1, \dots, n_s, p \to \infty$,
    with probability at least $1-o(1)$,
    \begin{equation*}
        \begin{aligned}
           \min_{\mathbf{R}_s \in \mathbb R^{k_0 \times k_0} : \mathbf{R}_s^\top \mathbf{R}_s = \mathbf{I}_{k_0}} \frac{1}{\sqrt{n_s}}||\mathbf{\hat M}_s \mathbf{R}_s- \mathbf{M}_{0s}||& \lesssim \frac{1}{\sqrt{n_s}} + \sqrt{\frac{n}{p n_s}},\\
             \min_{\mathbf{R}_s \in \mathbb R^{q_s \times q_s} : \mathbf{R}_s^\top \mathbf{R}_s = \mathbf{I}_{q_s}} \frac{1}{\sqrt{n_s}}|| \mathbf{\hat  F_s} \mathbf{R}_s- \mathbf{F}_{0s}||& \lesssim \frac{1}{\sqrt{n_{\min}}} + \frac{1}{\sqrt{p}}.
        \end{aligned}
    \end{equation*}
\end{theorem}
\begin{proof}[Proof of Theorem \ref{thm:factors_procrustes_error_additional}]
We apply the same steps as in the proof for Theorem \ref{thm:factors_procrustes_error}.  Under Assumptions \ref{assumption:li}--\ref{assumption:sv_A} and \ref{ass:bai-mean}--\ref{ass:lf_projection}, with probability at least $1-o(1)$, we have  $ ||\mathbf{U}^c \mathbf{U}^c - \mathbf{U}_0 \mathbf{U}_0^\top ||  \lesssim  \frac{1}{\sqrt{n}} + \frac{1}{\sqrt{p}}$, by Proposition \ref{prop:recovery_M}, where $\mathbf{U}_0$ denotes the matrix of left singular vectors of $\mathbf M_0$. Then, applying Davis-Kahan theorem \citep{davis_kahan} and using Assumption \ref{ass:sv_lf}, we obtain the first result. The proof for the second result proceeds by applying Proposition \ref{prop:U_s_outer} followed by similar steps. 
\end{proof}}
{
\begin{theorem}[Consistency under moment-based assumptions]\label{thm:consistency_supp}
    Suppose Assumptions \ref{assumption:model}, \ref{assumption:li}--\ref{assumption:sv_A} and \ref{ass:bai-mean}--\ref{ass:sv_lf} hold and $n_s = \mathcal O(n_{\min}^2)$, where $n_{\min} = \min_{s=1, \dots, S} n_s$, for all $s=1, \dots, S$, then, as $n_1, \dots, n_s, p \to \infty$,
    with probability at least $1-o(1)$,
    \begin{equation}
    \begin{aligned}
          \frac{ \left|\left| \mathbf{\mu}_\Lambda \mathbf{\mu}_\Lambda^\top - {\mathbf{\Lambda}}_0 {\mathbf{\Lambda}}_0^\top  \right|\right|}{ \left|\left|  {\mathbf{\Lambda}}_0 {\mathbf{\Lambda}}_0^\top  \right|\right|} &\lesssim  \frac{\sqrt{\log n}}{\sqrt{n}}  + \frac{1}{\sqrt{p}},\\
          \frac{ \left|\left| \mathbf{\mu}_{\Gamma_s} \mathbf{\mu}_{\Gamma_s}^\top - {\mathbf{\Gamma}}_{0s} {\mathbf{\Gamma}}_{0s}^\top  \right|\right|}{ \left|\left| {\mathbf{\Gamma}}_{0s} {\mathbf{\Gamma}}_{0s}^\top  \right|\right|} &\lesssim \frac{\sqrt{\log n_s}}{\sqrt{n_s}}  + \frac{1}{\sqrt{p}}, \quad (s=1, \dots, S).
    \end{aligned}
    \end{equation}
\end{theorem}
\begin{proof}[Proof of Theorem \ref{thm:consistency_supp}]
The proof mimics the one of Theorem \ref{thm:posterior_contraction_Lambda_outer} under the novel assumptions. In particular, we follow the same steps as in Theorem \ref{thm:posterior_contraction_Lambda_outer}, and we prove the first results, using the fact that, with probability at least $1- o(1)$, 
\begin{enumerate}
   \item $||\mathbf{M}_0|| \asymp \sqrt{n}$ by Assumption \ref{ass:sv_lf},
   \item $||   \mathbf{M}_0^\top  \mathbf{M}_0 - n \mathbf{I}_{k_0}  || \lesssim \sqrt{n \log n}$ by Assumption \ref{ass:sv_lf},
   \item $||\mathbf P_0^{\perp}  \mathbf{M}_0 || \lesssim 1$, which is implied by Assumption \ref{ass:lf_projection},
    \item $||\mathbf{E}||\lesssim   \sqrt{n} p^{1/4}+ n^{1/4}\sqrt{p}$ by Lemma \ref{lemma:E},
    \item $||\mathbf{\Delta Y} || \lesssim \sqrt{n} + \sqrt{p}$ by Lemma \ref{lemma:Delta_y_j},
    \item $|| \mathbf E^\top \mathbf M_0|| \lesssim \sqrt{n p}$ by Equation (4) of \citet{Bai2020SimplerProofs},
    \end{enumerate}
    and $||{\mathbf{\Lambda}}_0|| \asymp \sqrt{p}$ by Assumption \ref{assumption:Lambda}, as well as, $||\mathbf{P}_0^\perp|| = ||\mathbf{Q}_0^\perp|| = ||\mathbf{U}_0^c|| = ||\mathbf{N}_0^\perp|| =1$,  where  $\mathbf{N}_0^\perp$ is defined in  \eqref{eq:N_0_perp}.
    Next, use the same bounds, along with $|| \mathbf{U}_{0s}^c  \mathbf{U}_0^{c \top} - \mathbf{U}^c  \mathbf{U}^{c \top} || \leq  ||  \mathbf{U}_0^c  \mathbf{U}_0^{c \top} - \mathbf{U}^c  \mathbf{U}^{c \top}|| \lesssim \frac{1}{\sqrt{n}} + \frac{1}{\sqrt{p}}$, with probability at least $1-o(1)$, by Proposition \ref{prop:recovery_M}, to claim, with probability at least $1-o(1)$,
$$
||\mathbf{\hat M} \mathbf{\mu}_\Lambda^\top  -  \mathbf{M}_0{\mathbf{\Lambda}}_0^\top|| \lesssim \sqrt{n} p^{1/4}+ n^{1/4}\sqrt{p}. 
$$
Finally, we use the fact that, with probability at least $1-o(1)$,
\begin{enumerate}
    \item  $||\mathbf{F}_{0s}|| \asymp \sqrt{n_s}$ by Assumption \ref{ass:sv_lf},
    \item $||   \mathbf{F}_{0s}^\top  \mathbf{F}_{0s} - n_s \mathbf{I}_{q_s}  || \lesssim \sqrt{n_s \log n_s}$ by Assumption \ref{ass:sv_lf},
    \item  $||\mathbf{E}_s||\lesssim   \sqrt{n_s} p^{1/4}+ n_s^{1/4}\sqrt{p}$ by Lemma \ref{lemma:E},
    \item $||\mathbf Y_s|| \lesssim \sqrt{n_s p}$ by Lemma \ref{lemma:Y_s},
    \item $|| \mathbf{U}_s^{\perp} \mathbf{U}_s^{\perp \top} - \mathbf{U}_{0s}^\perp \mathbf{U}_{0s}^{\perp \top} || \lesssim \frac{1}{\sqrt{n_s}} + \frac{1}{\sqrt{p}}$ by Proposition \ref{prop:U_s_outer},
\end{enumerate}
together with $ \frac{n_s}{(n_s + \tau_{\Gamma_s}^{-2})^2} \asymp \frac{1}{n_s}$, to prove the second result. 
\end{proof}
}

\subsection{Additional lemmas}\label{subsec:additional_lemmas}

\begin{lemma}\label{lemma:E}
Under Assumptions \ref{assumption:model} and \ref{assumption:sigma}, with probability at least $1-o(1)$, we have 
\begin{equation*}
    ||\mathbf{E}_s|| \lesssim C_{sr, n_s} C_{sc, p}  (\sqrt{n_s} + \sqrt{p}), \quad (s=1, \dots, S), \quad  ||\mathbf{E}|| \lesssim C_{r, \{n_s\}_{s}^{S}} C_{c, p}  (\sqrt{n} + \sqrt{p}).
\end{equation*}
Replacing Assumptions \ref{assumption:distributions} and \ref{assumption:sigma} with Assumptions \ref{ass:bai-mean}--\ref{ass:sv_lf},
 with probability at least $1-o(1)$, we have 
\begin{equation*}
    ||\mathbf{E}_s|| \lesssim   \sqrt{n_s} p^{1/4} + n_s^{1/4}\sqrt{p}, \quad (s=1, \dots, S), \quad  ||\mathbf{E}|| \lesssim  \sqrt{n} p^{1/4} + n^{1/4}\sqrt{p}.
\end{equation*}
\end{lemma}
\begin{proof}[Proof of Lemma \ref{lemma:E}]
   The first result follows from $ ||\mathbf{E}_s|| = ||\mathbf{\Sigma}_{0sr}^{1/2} ||  ||\mathbf{\Sigma}_{0sr}^{-1/2}\mathbf{E}_s \mathbf{\Sigma}_{0sc}^{-1/2}|| ||\mathbf{\Sigma}_{0sc}^{1/2}   ||$ and \\$||\mathbf{\Sigma}_{0sr}^{-1/2}\mathbf{E}_s \mathbf{\Sigma}_{0sc}^{-1/2}|| \lesssim \sqrt{n_s} + \sqrt{p}$ by Theorem 4.6.1 of \citet{vershynin2018hdp}. The second results is a consequence of $||\mathbf{E}_s|| \leq ||\mathbf{E}_s||_F \lesssim \sqrt{n_s} p^{1/4} + n_s^{1/4}\sqrt{p}$, where the second inequality follows from Lemma 1 of \citet{Bai2020SimplerProofs}.
\end{proof}

\begin{lemma}\label{lemma:y_j}
    Suppose Assumptions \ref{assumption:model}--\ref{assumption:sigma} hold. Then, with probability at least $1-o(1)$,
    \begin{equation*}
    \begin{aligned}
        \max_{j=1, \dots, p} ||\mathbf{y}_s^{(j)}|| \lesssim C_{\sigma} C_{sr, n_s} \sqrt{n_s}, \quad (s=1, \dots, S), \quad   \max_{j=1, \dots, p} ||\mathbf y^{(j)}|| \lesssim C_{\sigma} C_{r, \{n_s\}_{s}^{S}} \sqrt{n}.\\
        \max_{j=1, \dots, p} ||\mathbf{e}_s^{(j)}|| \lesssim C_{\sigma} C_{sr, n_s}\sqrt{n_s}, \quad (s=1, \dots, S), \quad  
        \max_{j=1, \dots, p} ||\mathbf{e}^{(j)}|| \lesssim C_{\sigma} C_{r, \{n_s\}_{s}^{S}}\sqrt{n}.
    \end{aligned}   
    \end{equation*}
\end{lemma}
\begin{proof}[Proof of Lemma \ref{lemma:y_j}]
Consider the following 
\begin{equation*}
    \begin{aligned}
        ||\mathbf{y}_s^{(j)}|| \leq ||\mathbf{M}_{0s}|| ||\mathbf{\lambda}_{0j}|| + ||\mathbf{F}_{0s}|| ||\mathbf{\gamma}_{0sj}|| + ||\mathbf{e}_s^{(j)}||.
    \end{aligned}
\end{equation*}
The first result follows from $|| \mathbf{M}_{0s} || \asymp ||\mathbf{F}_{0s}|| \asymp \sqrt{n_s}$ and $\max_{j=1, \dots, p} || \mathbf{e}_s^{(j)} || \lesssim C_{\sigma} C_{sr, n_s}  \sqrt{n_s}$ with probability at least $1-o(1)$, by Corollary 5.35 of \citet{vershynin_12} and Lemma 1 of \citet{laurent_massart}, and $\max_{j=1, \dots, p}  ||\mathbf{\lambda}_{0j}|| \leq ||{\mathbf{\Lambda}}_0 ||_\infty \sqrt{k_0} \asymp 1$, $\max_{j=1, \dots, p}  ||\mathbf{\gamma}_{0sj}|| \leq ||{\mathbf{\Gamma}}_s ||_\infty \sqrt{q_s} \asymp 1$ by Assumptions \ref{assumption:Lambda} and \ref{assumption:Gammas}.
We derive the analogous result for $y^{(j)}$ by noting that $|| \mathbf{M}_{0} || \asymp ||F_{0}|| \asymp \sqrt{n}$ and $\max_{j=1, \dots, p} || \mathbf{e}^{(j)} || \lesssim C_{\sigma} C_{r, \{n_s\}_{s}^{S}} \sqrt{n}$ with probability at least $1-o(1)$, by Corollary 5.35 of \citet{vershynin_12} and Lemma 1 of \citet{laurent_massart} respectively.
\end{proof}

\begin{lemma}\label{lemma:y_hat}
    Suppose Assumptions \ref{assumption:model}--\ref{assumption:sigma} hold. Then, with probability at least $1-o(1)$,
    \begin{equation*}
        \max_{j=1, \dots, p} ||\mathbf{\hat y}_s^{c(j)}|| \lesssim C_{\sigma} C_{sr, n_s}\sqrt{n_s}, \quad (s=1, \dots, S), \quad  
        \max_{j=1, \dots, p} ||\mathbf{\hat y}^{c(j)}|| \lesssim C_{\sigma} C_{r, \{n_s\}_{s}^{S}}\sqrt{n}.
    \end{equation*}
\end{lemma}
\begin{proof}[Proof of Lemma \ref{lemma:y_hat}]
    Recall that $\mathbf{\hat y}_s^{c (j)} = (\mathbf{I}_{n_s} - \mathbf{U}_s^{\perp}\mathbf{U}_s^{\perp \top})\mathbf{y}_s^{(j)}$. The first result follows from
    \begin{equation*}
        ||\mathbf{\hat y}_s^{c (j)} || = ||(\mathbf{I}_{n_s} - \mathbf{U}_s^{\perp}\mathbf{U}_s^{\perp \top})\mathbf{y}_s^{(j)} || \leq ||\mathbf{y}_s^{(j)}||
    \end{equation*}
    and Lemma \ref{lemma:y_j}. 
    For the second result follows.
\end{proof}

\begin{lemma}\label{lemma:Y_s}
Suppose Assumptions \ref{assumption:model}--\ref{assumption:sigma} hold. Then, with probability at least $1-o(1)$,
    \begin{equation*}
    \begin{aligned}
        ||\mathbf{Y}_s|| \lesssim \sqrt{n_s p}, 
        \quad (s=1, \dots, S), \quad 
        ||\mathbf{Y}|| \lesssim \sqrt{n p}. 
    \end{aligned}
    \end{equation*}
    The same result holds Replacing Assumptions \ref{assumption:distributions} and \ref{assumption:sigma} with Assumptions \ref{ass:bai-mean}--\ref{ass:sv_lf}.
\end{lemma} 
\begin{proof}[Proof of Lemma \ref{lemma:Y_s}]
First, note $||\mathbf{Y}_s|| \leq ||\mathbf{M}_{0s} {\mathbf{\Lambda}}_0^\top || + ||\mathbf{F}_{0s} {\mathbf{\Gamma}}_s{0s}^\top || + ||\mathbf{E}_s||$. 
 By Theorem 4.6.1 of \citet{vershynin2018hdp}, we have $|| \mathbf{M}_{0s}|| \asymp ||\mathbf{F}_{0s}|| \asymp \sqrt{n_s}$ with probability at least $1-o(1)$, and Assumptions \ref{assumption:Lambda} and \ref{assumption:Gammas} imply $||{\mathbf{\Lambda}}_{0}|| \asymp||{\mathbf{\Gamma}}_{0s}|| \asymp \sqrt{p}$. Therefore, we have $||\mathbf{M}_{0s}{\mathbf{\Lambda}}_0^\top + \mathbf{F}_{0s} {\mathbf{\Gamma}}_{0s}^\top|| \lesssim \sqrt{n_s p}$, with probability at least $1-o(1)$. Lemma \ref{lemma:E} provides a probabilistic bound on $||\mathbf{E}_s||$. 
\end{proof}
\begin{lemma}\label{lemma:Mt_P_M}
Let $\mathbf{P}_0^\perp$ be the matrix defined in \eqref{eq:P_0_perp}, $M \in \mathbb R^{n \times k_0}$ be a matrix of independent standard Gaussian random variables and $e_j$ be the $j$-th column of an $n \times p$ matrix of zero-mean sub-Gaussian random variables with bounded variances and independent rows. 
Then, with probability at least $1 - o(1)$, we have 
$$||\mathbf{P}_0^\perp  M||\asymp || M^\top  \mathbf{P}_0^\perp  M||\asymp 1, \text{ and } \max_{j=1, \dots, p}  ||\mathbf{P}_0^\perp  e_j|| \lesssim  \max_{j=1, \dots, p}||e_j^\top \mathbf{P}_0^\perp e_j|| \asymp \max_{j=1, \dots, p}|| M^\top  \mathbf{P}_0^\perp  e_j|| \lesssim \log p.$$
\end{lemma}
\begin{proof}[Proof of Lemma \ref{lemma:Mt_P_M}]
   The results follow from Theorem 4.6.1 of \citet{vershynin2018hdp} and Theorem 2.1 of \citet{rudelson_13}.
\end{proof}

\begin{lemma}\label{lemma:U_c_e_j}
    Let $\mathbf{U}_0^c \in \mathbb R^{n \times k_0}$ be the matrix of left singular vectors associated to the $k_0$ leading singular values of $\mathbf{M}_0$. Then, we have 
    \begin{equation*}
       \max_{j=1, \dots, p} ||\mathbf{U}_0^{c \top} \mathbf e^{(j)} || \lesssim \log p
    \end{equation*}
    with probability at least $1-o(1)$.
\end{lemma}
\begin{proof}[Proof of Lemma \ref{lemma:U_c_e_j}]
The result follows from Theorem 2.1 of \citet{rudelson_13}. 
\end{proof}

\begin{lemma}\label{lemma:U_c_F}
     Let $\mathbf{U}_0^c= \big[\mathbf{U}_{01}^{c\top} ~ \cdots ~   \mathbf{U}_{0S}^{c\top}\big]^\top  \in \mathbb R^{n \times k_0}$, where $\mathbf{U}_{0s}^{c} \in \mathbb R^{n_s \times k_0}$ is the matrix of  left singular vectors associated to the $k_0$ leading singular values of $\mathbf{M}_0$ and $\mathbf{F} \in \mathbb R^{n_s \times q_s}$ is a matrix with independent standard Gaussian entries. Then, with probability at least $1-o(1)$, we have
\begin{equation*}
    ||\mathbf{F}^\top \mathbf{U}_{0s}^{c} || \lesssim 1.
\end{equation*}
\end{lemma}
\begin{proof}[Proof of Lemma \ref{lemma:U_c_F}]
    Note $  \mathbf{F}^\top \mathbf{U}_{0s}^{c} \sim MN_{ q_s, k_0}(0, \mathbf{I}_{q_s}, \mathbf{U}_{0s}^{c \top }\mathbf{U}_{0s}^c)$. Recall that $\mathbf{U}_{0}^{c \top }\mathbf{U}_0^c = \mathbf{I}_{ k_0}$ and $\mathbf{U}_{0s}^{c \top }\mathbf{U}_{0s}^c = \mathbf{U}_{0}^{c \top }\mathbf{U}_0^c - \mathbf{U}_{0-s}^{c \top }\mathbf{U}_{0-s}^c$, where $\mathbf{U}_{0-s}$ denotes the matrix obtained by removing the block $\mathbf{U}_{0s}$ from $\mathbf{U}_{0}$, which implies $ \mathbf{U}_{0s}^{c \top }\mathbf{U}_{0s}^c \preceq \mathbf{I}_{ k_0}$ and, consequently, by Theorem 4.6.1 of \citet{vershynin2018hdp}, $|| \mathbf{F}^\top \mathbf{U}_{0s}^{c} || \lesssim \sqrt{k_0} + \sqrt{q_s} $ with probability $1- o(1)$.
\end{proof}

\begin{lemma}\label{lemma:U_perp_M}
     Let $U \in \mathbb R^{n \times q}$ be such that $U^\top U = \mathbf{I}_{q}$, and $M \in\mathbb R^{n \times k}$ be a matrix with independent standard Gaussian entries. Then, with probability at least $1-o(1)$, we have
\begin{equation*}
    ||M^\top U || \lesssim \sqrt{k} + \sqrt{p}.
\end{equation*}
\end{lemma}
\begin{proof}[Proof of Lemma \ref{lemma:U_perp_M}]
    It is enough to note that $M^\top U \sim MN_{k, q}(0, \mathbf{I}_k, \mathbf{I}_q)$ and apply Theorem 4.6.1 of \citet{vershynin2018hdp}.
\end{proof}

\begin{lemma}\label{lemma:Delta_y_j}
    Let $\Delta \in \mathbb R^{n \times n}$ be the matrix defined in \eqref{eq:Delta}, then under Assumptions \ref{assumption:model}--\ref{assumption:sigma}, we have
    $$
    \max_{j=1, \dots, p} ||\mathbf{\Delta} y^{(j)}|| \lesssim C_{\sigma} C_{r, \{n_s\}_{s}^{S}}\big( \frac{1}{\sqrt{n_{\min}}} + \frac{\sqrt{n_{\max}}}{p} \big), 
    $$ 
    $$
     ||\mathbf{\Delta} \mathbf{Y}|| \lesssim (C_{r, \{n_s\}_{s}^{S}} C_{c, p}   )^2 \big(\frac{\sqrt{p}}{\sqrt{n_{\min}}} + \frac{\sqrt{n_{\max}}}{\sqrt{p}}\big),
    $$
    and
    $$
     \max_{j=1, \dots, p} ||\mathbf{e}^{(j)\top }\mathbf{\Delta} \mathbf y^{(j)}|| \lesssim C_{\sigma}^2 C_{r, \{n_s\}_{s}^{S}}^4 C_{c, p}^2 \big(1 + \frac{n_{\max}}{p}\big), 
    $$
    with probability at least $1 -o(1)$, where $n_{\max} = \max_{s=1, \dots, S} n_s$ and $n_{\min}= \min_{s=1, \dots, S} n_s$. 
    Replacing Assumptions \ref{assumption:distributions} and \ref{assumption:sigma} with Assumptions \ref{ass:bai-mean}--\ref{ass:sv_lf},
 with probability at least $1-o(1)$, we have 
 $$
     ||\mathbf{\Delta} \mathbf{Y}|| \lesssim \sqrt{n} + \sqrt{p}.
    $$
\end{lemma}
\begin{proof}[Proof of Lemma \ref{lemma:Delta_y_j}]
First, note that
    \begin{equation*}
    \mathbf{\Delta} y^{(j)} = \begin{bmatrix}
        (\mathbf{U}_{01}^\perp \mathbf{U}_{01}^{\perp \top}- \mathbf{U}_1^\perp \mathbf{U}_1^{\perp \top}) \mathbf{y}_1^{(j)}\\
        \vdots \\
                (\mathbf{U}_{0S}^\perp \mathbf{U}_{0S}^{\perp \top}- \mathbf{U}_S^\perp \mathbf{U}_S^{\perp \top}) \mathbf{y}_S^{(j)}
    \end{bmatrix}
\end{equation*}
where $\mathbf{y}_s^{(j)}$ denotes the $j$-th column of $\mathbf{Y}_s$. 
Note that $|| (\mathbf{U}_{0s}^\perp \mathbf{U}_{0s}^{\perp \top}- \mathbf{U}_s^\perp \mathbf{U}_s^{\perp \top})\mathbf{y}_s^{(j)}|| \lesssim ||\mathbf{U}_{0s}^\perp \mathbf{U}_{0s}^{\perp \top}- \mathbf{U}_s^\perp \mathbf{U}_s^{\perp \top}|| ||\mathbf{y}_s^{(j)}|| \lesssim 
(C_{sr, n_s} C_{sc, p})^2  
\big(\frac{1}{n_s} + \frac{1}{p}\big)C_{\sigma} C_{sr, n_s}\sqrt{n_s} = C_{\sigma}C_{sr, n_s}^3 C_{sc, p}^2\big(\frac{1}{\sqrt{n_s}} + \frac{\sqrt{n_s}}{p}\big)$. Thus, $|| \mathbf{\Delta} y^{(j)} || \lesssim C_{\sigma}C_{r, \{n_s\}_{s}^{S}}^3 C_{c, p}^2\big( \frac{1}{\sqrt{n_{\min}}} + \frac{\sqrt{n_{\max}}}{p} \big)$, with probability at least $1- o(1)$, where $n_{\max} = \max_{s=1, \dots, S} n_s$ and $n_{\min}= \min_{s=1, \dots, S} n_s$, since $\max_{j=1, \dots, p} ||  \mathbf{y}_s^{(j)}|| \lesssim C_{\sigma} C_{sr, n_s}\sqrt{n_s}$ with probability at least $1- o(1)$, by Lemma \ref{lemma:y_j}.
For the second result, consider
\begin{equation*}
   \mathbf{\Delta Y} = \begin{bmatrix}
        (\mathbf{U}_{01}^\perp \mathbf{U}_{01}^{\perp \top}- \mathbf{U}_1^\perp \mathbf{U}_1^{\perp \top}) \mathbf Y_1\\
        \vdots \\
        (\mathbf{U}_{0S}^\perp \mathbf{U}_{0S}^{\perp \top}- \mathbf{U}_S^\perp \mathbf{U}_S^{\perp \top}) \mathbf Y_S
    \end{bmatrix}
\end{equation*}
with $$|| (\mathbf{U}_{0s}^\perp \mathbf{U}_{0s}^{\perp \top}- \mathbf{U}_s^\perp \mathbf{U}_s^{\perp \top})\mathbf{Y}_s|| \lesssim ||\mathbf{U}_{0s}^\perp \mathbf{U}_{0s}^{\perp \top}- \mathbf{U}_s^\perp \mathbf{U}_s^{\perp \top}|| ||\mathbf{Y}_s||
.$$ Applying the bounds in Proposition \ref{prop:U_s_outer} and Lemma \ref{lemma:Y_s} with $$||\mathbf{\Delta Y} || \lesssim S\max_{s=1, \dots, S}|| (\mathbf{U}_{0s}^\perp \mathbf{U}_{0s}^{\perp \top}- \mathbf{U}_s^\perp \mathbf{U}_s^{\perp \top})\mathbf{Y}_s|| 
$$ proves the result.
Finally, note $  \mathbf{e}^{(j)\top} \mathbf{\Delta} \mathbf y^{(j)} = \sum_{s=1}^S   \mathbf{e}_s^{(j)\top} (\mathbf{U}_{0s}^\perp \mathbf{U}_{0s}^{\perp \top}- \mathbf{U}_s^\perp \mathbf{U}_s^{\perp \top}) \mathbf{y}_s^{(j)} $, and $|\mathbf{e}_s^{(j)\top} (\mathbf{U}_{0s}^\perp \mathbf{U}_{0s}^{\perp \top}- \mathbf{U}_s^\perp \mathbf{U}_s^{\perp \top}) \mathbf{y}_s^{(j)}| \lesssim ||\mathbf{e}_s^{(j)}|| ||\mathbf{U}_{0s}^\perp \mathbf{U}_{0s}^{\perp \top}- \mathbf{U}_s^\perp \mathbf{U}_s^{\perp \top}|| ||\mathbf{y}_s^{(j)}||$. The bounds in Proposition \ref{prop:U_s_outer} and Lemma \ref{lemma:y_j} complete the proof. 

\end{proof}

\begin{lemma}\label{lemma:r_j}
    Let $\mathbf{r}^{(j)} =\mathbf{\Delta} y^{(j)}   - \mathbf{P}_0^\perp ( \mathbf{M}_{0} \mathbf{\lambda}_{0j} + \mathbf{e}^{(j)} )$, where $\Delta$ and $\mathbf{P}_0^\perp$ are defined in \eqref{eq:Delta} and \eqref{eq:P_0_perp} respectively.
    Then, under Assumption \ref{assumption:model}--\ref{assumption:sigma}, we have
    $$
    ||\mathbf{r}^{(j)}|| \lesssim C_{r, \{n_s\}_{s}^{S}}
\bigg(\frac{1}{\sqrt{n_{\min}}} + \frac{\sqrt{n_{\max}}}{p}\bigg) + \log p, 
    $$
     with probability at least $1 -o(1)$, where $n_{\max} = \max_{s=1, \dots, S} n_s$ and $n_{\min}= \min_{s=1, \dots, S} n_s$. 
\end{lemma}
\begin{proof}[Proof of Lemma \ref{lemma:r_j}]
    First, consider
    \begin{equation*}
         ||\mathbf{r}^{(j)}|| \leq || \mathbf{P}_0^\perp  \mathbf{M}_0^{\top}|| ||\mathbf{\lambda}_{0j}|| + || \mathbf{P}_0^\perp  \mathbf{e}^{(j)}|| + ||\Delta  \mathbf y^{c (j)}||.
    \end{equation*}
   Moreover, $|| \mathbf{P}_0^\perp  \mathbf{M}_0^{\top}|| \lesssim 1$, $|| \mathbf{P}_0^\perp  \mathbf{e}^{(j)}|| \lesssim \log p$ and $||\Delta y^{c (j)}|| \lesssim C_{\sigma} C_{r, \{n_s\}_{s}^{S}}
(\frac{1}{\sqrt{n_{\min}}} + \frac{\sqrt{n_{\max}}}{p}) $ with probability at least $1 -o(1)$ by Lemma \ref{lemma:Mt_P_M} and Lemma \ref{lemma:Delta_y_j} and 
  $ ||\mathbf{\lambda}_{0j}|| \lesssim  ||{\mathbf{\Lambda}}||_{\infty}\sqrt{k_0} \asymp 1$ by Assumption \ref{assumption:Lambda}. 
\end{proof}

\begin{lemma}\label{lemma:convergence}
    Suppose Assumptions \ref{assumption:model}--\ref{assumption:sigma} hold, $n_s = \mathcal O(n_{\min}^2)$, where $n_{\min} = \min_{s=1, \dots, S} n_s$, for all $s=1, \dots, S$, and $\sqrt{n}/p = o(1)$. Then, as $n_1, \dots, n_S, p \to \infty$, we have
    \begin{equation*}
        \begin{aligned}
        \mathbf{\mu}_{\lambda_j}^\top\mathbf{\mu}_{\lambda_{j'}}&\overset{pr}{\to} \mathbf{\lambda}_{0j}^\top \mathbf{\lambda}_{0j'},\\
            \frac{1}{n}||(\mathbf{I}_n - \mathbf{U}^c  \mathbf{U}^{c \top})\mathbf{\hat y}^{c (j)}||^2 &\overset{pr}{\to} \sigma_{0j}^2.
        \end{aligned}
    \end{equation*}
\end{lemma}
\begin{proof}[Proof of Lemma \ref{lemma:convergence}]
    The first result follows easily from the proof of Theorem \ref{thm:clt_mu_Lambda_outer}. For the second result, note 
        \begin{equation*}
            \begin{aligned}
                \mathbf{\hat y}^{c (j)\top}\mathbf{\hat y}^{c (j)} =& \left( \mathbf{M}_0 \mathbf{\lambda}_{j} +  \mathbf{e}^{(j)}\right)^\top \left(\mathbf{I}_n -  \mathbf{P}_0^\perp\right) \left( \mathbf{M}_0 \mathbf{\lambda}_{j} +  \mathbf{e}^{(j)}\right) +\mathbf{\hat y}^{c (j)\top}\mathbf{\Delta}^2\mathbf{\hat y}^{c (j)}\\
                &+ 2  \left( \mathbf{M}_0 \mathbf{\lambda}_{j} +  \mathbf{e}^{(j)}\right)^\top \left(\mathbf{I}_n -  \mathbf{P}_0^\perp\right) \mathbf{\Delta} y^{(j)}.
            \end{aligned}
        \end{equation*} 
        Moreover, $$\left( \mathbf{M}_0 \mathbf{\lambda}_{j} +  \mathbf{e}^{(j)}\right)^\top \left(\mathbf{I}_n -  \mathbf{P}_0^\perp\right) \left( \mathbf{M}_0 \mathbf{\lambda}_{j} +  \mathbf{e}^{(j)}\right) \sim \chi_n + \psi_n,$$
        with $\chi_n  \sim ||\mathbf{\lambda}_{0j}|| ^2\chi_{n - \sum_s q_s}^2$, $E[\psi_n] = (n -k) \sigma_{0j}^2$, $V(\psi_n) \asymp n$,  and $\chi_n$ being independent of  $\psi_n$,   while $  y^{c (j) \top}\mathbf{\Delta}^2\hat  y^{c (j)}  \lesssim \frac{1}{n} + \frac{1}{p}$ and $\left( \mathbf{M}_0 \mathbf{\lambda}_{j} +  \mathbf{e}^{(j)}\right)^\top \left(\mathbf{I}_n -  \mathbf{P}_0^\perp\right)\mathbf{\Delta} y^{ (j)} \asymp 1$ with probability at least $1 - o(1)$. Hence, $$\frac{1}{n} \left( \mathbf{M}_0 \mathbf{\lambda}_{j} +  \mathbf{e}^{(j)}\right)^\top \left(\mathbf{I}_n -  \mathbf{P}_0^\perp\right) \left( \mathbf{M}_0 \mathbf{\lambda}_{j} +  \mathbf{e}^{(j)}\right) \overset{pr}{\to} \sigma_{0j}^2 + ||\mathbf{\lambda}_{0j}||^2,$$ by the weak law of large numbers and $\frac{n}{n - \sum_s q_s} \asymp 1$, which in combination with\\ $ \frac{1}{n}\mathbf{\hat y}^{c (j)}  \mathbf{U}^c  \mathbf{U}^{c \top}\mathbf{\hat y}^{c (j')} \overset{pr}{\to} \mathbf{\lambda}_{0j}^\top \mathbf{\lambda}_{0j'}$, proves the result. 
\end{proof}

\begin{lemma}\label{lemma:convergence_Gamma_s}
    Under Assumptions \ref{assumption:model}--\ref{assumption:sigma}, $n_s = \mathcal O(n_{\min}^2)$, where $n_{\min} = \min_{s=1, \dots, S} n_s$, for all $s=1, \dots, S$, and $\sqrt{n}/p = o(1)$. Then, as $n_1, \dots, n_S, p \to \infty$, we have $$\mathbf{\mu}_{\gamma_{sj}}^\top\mathbf{\mu}_{\gamma_{sj'}}  \overset{pr}{\to} \mathbf{\gamma}_{0sj}^\top\mathbf{\gamma}_{0sj'}.$$
\end{lemma}
\begin{proof}
    It follows from Theorem \ref{thm:clt_mu_Lambda_outer}.
\end{proof}

\begin{lemma}\label{lemma:sigma_j_tilde}
    Suppose Assumptions \ref{assumption:model}--\ref{assumption:sigma} hold. Then,
    \begin{equation*}
        \tilde \Pi \big(\max_{j=1, \dots, p} \tilde \sigma_j^2 \leq C_{\tilde \sigma}\big) = 1 -o(1)
    \end{equation*}
    with probability at least $1-o(1)$, where $C_{\tilde \sigma}<\infty$ is a finite constant. 
\end{lemma}
\begin{proof}[Proof of Lemma \ref{lemma:sigma_j_tilde}]
    Recall that $\tilde \sigma_j^2 \sim IG \big(\mathbf{\gamma}_n/2, \mathbf{\gamma}_n \delta_j^2/2\big)$, where $\mathbf{\gamma}_n = \nu_0 + n$ and $\delta_j^2 = \frac{1}{\mathbf{\gamma}_n} \big(\nu_0 \sigma_0^2  +\mathbf{\hat y}^{c (j) \top}\mathbf{\hat y}^{c (j)} -    \frac{1}{n + \tau_{\Lambda}^{-2}}\mathbf{\mu}_{\lambda_j} \mathbf{\mu}_{\lambda_j}\big).$
    Recall that $\mathbf{\hat y}^{c (j)} = \mathbf{M}_{0} \mathbf{\lambda}_{0j} +  \mathbf{e}^{(j)} + \mathbf{r}^{(j)}$, where $ \mathbf{r}^{(j)} = \Delta   y^{(j)}  - \mathbf{P}_0^\perp ( \mathbf{M}_{0} \mathbf{\lambda}_{0j} + \mathbf{e}^{(j)} )$, and $\Delta$ and $\mathbf{P}_0^\perp$ are defined in \eqref{eq:Delta} and \eqref{eq:P_0_perp}, respectively.
    Note that $|| \mathbf{M}_{0} \mathbf{\lambda}_{0j} +  \mathbf{e}^{(j)}||^2 = || \mathbf{M}_{0} \mathbf{\lambda}_{0j}||^2 + ||\mathbf{e}^{(j)}||^2$ and $  || \mathbf{M}_{0} \mathbf{\lambda}_{0j}||^2 \sim \sum_{s=1}^S  ||\mathbf{\lambda}_{0j}||^2 \chi_{n_s}^2 $. Applying Proposition 5.16  of \citet{vershynin_12}, we have $||\mathbf{\hat y}^{c (j)}||^2 \lesssim n$ with probability at least $1 - o(1)$. Moreover, $||\mathbf{\hat y}^{(j)}||^2 \lesssim || \mathbf{M}_{0} \mathbf{\lambda}_{0j} +  \mathbf{e}^{(j)}||^2 + ||\mathbf{r}^{(j)}||^2 \lesssim n$, since $||\mathbf{r}^{(j)}||^2 \lesssim 1$. Therefore, $\mathbf{\gamma}_s \delta_j^2 \lesssim ||\mathbf{\hat y}^{c (j)} ||^2 + n ||\mathbf{\mu}_{\lambda_j} ||^2 \lesssim n$, since $n ||\mathbf{\mu}_{\lambda_j}||^2 \lesssim ||\mathbf{\hat y}^{c (j)}||^2$ by 
    Theorem 5 of \citet{Zhang_Zou} with $v_n \asymp n$ concludes the proof. 
\end{proof}

\begin{lemma}\label{lemma:convergence_tilde_sigma}
     Suppose Assumptions \ref{assumption:model}--\ref{assumption:sigma} hold, $n_s = \mathcal O(n_{\min}^2)$, where $n_{\min} = \min_{s=1, \dots, S} n_s$, for all $s=1, \dots, S$, and $\sqrt{n}/p = o(1)$. Then, as $n_1, \dots, n_S, p \to \infty$, we have \begin{equation*}
       \max_{j=1, \dots, p} |   \tilde \sigma_j^2 - \sigma_{0j}| \lesssim \big(\frac{\log{p}}{n}\big)^{1/3} + \frac{1}{p}
\end{equation*}
with probability at least $1-o(1)$,  where $\tilde \sigma_j^2$ is a sample for the $j$-th error variance from $\tilde \Pi$.
\end{lemma}
\begin{proof}[Proof of Lemma \ref{lemma:convergence_tilde_sigma}]
Recall that $\tilde \sigma_j^2 \sim IG(\mathbf{\gamma}_n/2, \mathbf{\gamma}_n \delta_j^2 /2)$. First consider 
\begin{equation*}
    \begin{aligned}
        \delta_j^2 = \frac{1}{\mathbf{\gamma}_n}\bigg\{\nu_0 \sigma_0^2 +\mathbf{\hat y}^{c(j)\top} \Big(\mathbf{I}_n - \frac{n}{n + \tau_{\Lambda}^{-2}} \mathbf{U}^c \mathbf{U}^{c \top} \Big)y^{c(j)} \bigg\}
    \end{aligned}
\end{equation*}
and recall from the proof of Proposition \ref{prop:recovery_M} that $\mathbf{\hat y}^{c (j)} = (\mathbf{I}_n - \mathbf{P}_0^\perp ) (\mathbf{M}_0 \mathbf{\lambda}_{0j} + \mathbf{e}^{(j)}) + \mathbf{\Delta} y^{(j)}$.
In the following, we show
\begin{equation}\label{eq:delta_j_representation}
    \delta_j^2 = \frac{C_j^2}{n} + R_j
\end{equation}
and $C_j^2 =  (\mathbf M_0 \mathbf{\lambda}_{0j} + \mathbf{e}^{(j)})^\top \big(\mathbf{I}_n -  \mathbf{U}_0^c \mathbf{U}_0^{c \top}\big) (\mathbf M_0 \mathbf{\lambda}_{0j} + \mathbf{e}^{(j)})$, with $E[C_j^2] = (n - k_0) \sigma_{0j}^2$ and $C_j^2$ admitting the representation as the sum of $n-k_0$ of squares of independent zero-mean sub-Gaussian random variables, 
and $\max_{j=1, \dots, p} |R_j| \lesssim \frac{1}{n} + \frac{1}{p}$, with probability at least $1-o(1)$,
where   
\begin{equation*}
\begin{aligned}
        R_j = & \frac{1}{\nu_n} \Bigg[\nu_0 \sigma_0^2 + \mathbf{e}^{(j)\top}(\mathbf{I}_n -  \mathbf{U}_0^c \mathbf{U}_0^{c \top}) \{-\mathbf{P}_0^\perp (\mathbf{M}_0\mathbf{\lambda}_{0j} + \mathbf{e}^{(j)}) + \mathbf{\Delta} y^{(j)}\}\\
        &+ \{-\mathbf{P}_0^\perp (\mathbf{M}_0\mathbf{\lambda}_{0j} + \mathbf{e}^{(j)}) + \mathbf{\Delta} y^{(j)}\}^\top(\mathbf{I}_n -  \mathbf{U}_0^c \mathbf{U}_0^{c \top})\mathbf{e}^{(j)}\\
        &+\mathbf{\hat y}^{c(j)\top} \big\{(\mathbf{U}_0^c \mathbf{U}_0^{c \top} -  \mathbf{U}^c \mathbf{U}^{c \top}) + \big(1 -\frac{n}{n + \tau_{\Lambda}^{-2}} \big)\mathbf{U}^c \mathbf{U}^{c \top}\big\}\mathbf{\hat y}^{c(j)} \Bigg] \\ & + \big(\frac{1}{\mathbf{\gamma}_n} - \frac{1}{n}\big) \mathbf{e}^{(j)\top} (\mathbf{I}_n -  \mathbf{U}_0^c \mathbf{U}_0^{c \top})\mathbf{e}^{(j)} 
    \end{aligned}.
\end{equation*}
Note 
\begin{equation*}
    \begin{aligned}
       \mathbf{\hat y}^{c(j)\top} \big(\mathbf{I}_n - \frac{n}{n + \tau_{\Lambda}^{-2}} \mathbf{U}^c \mathbf{U}^{c \top} \big)y^{c(j)}  = &\mathbf{\hat y}^{c(j)\top} (\mathbf{I}_n -  \mathbf{U}_0^c \mathbf{U}_0^{c \top})y^{c(j)} \\
        &+ \mathbf{\hat y}^{c(j)\top} \big[ (\mathbf{U}_0^c \mathbf{U}_0^{c \top} -  \mathbf{U}^c \mathbf{U}^{c \top}) + \big(1 -\frac{n}{n + \tau_{\Lambda}^{-2}} \big)\mathbf{U}^c \mathbf{U}^{c \top}\big]\mathbf{\hat y}^{c(j)}. 
    \end{aligned}
\end{equation*} 
First, we analyze $\mathbf{\hat y}^{c(j)\top} (\mathbf{I}_n -  \mathbf{U}_0^c \mathbf{U}_0^{c \top})\mathbf{\hat y}^{c(j)}$:
\begin{equation*}
    \begin{aligned}
       \mathbf{\hat y}^{c(j)\top} (\mathbf{I}_n -  \mathbf{U}_0^c \mathbf{U}_0^{c \top})\mathbf{\hat y}^{c(j)}  = & \mathbf{e}^{(j)\top} (\mathbf{I}_n -  \mathbf{U}_0^c \mathbf{U}_0^{c \top})\mathbf{e}^{(j)} \\
        &+ \mathbf{e}^{(j)\top}(\mathbf{I}_n -  \mathbf{U}_0^c \mathbf{U}_0^{c \top}) \{-\mathbf{P}_0^\perp (\mathbf{M}_0\mathbf{\lambda}_{0j} + \mathbf{e}^{(j)}) + \mathbf{\Delta} y^{(j)}\}\\
        &+ \{-\mathbf{P}_0^\perp (\mathbf{M}_0\mathbf{\lambda}_{0j} + \mathbf{e}^{(j)}) + \mathbf{\Delta} y^{(j)}\}^\top(\mathbf{I}_n -  \mathbf{U}_0^c \mathbf{U}_0^{c \top})\mathbf{e}^{(j)}
    \end{aligned}
\end{equation*}
since $(I - \mathbf{U}_0^c \mathbf{U}_0^{c \top})\mathbf{M}_0 =  0$. 
Note that $\mathbf{e}^{(j)\top} (\mathbf{I}_n -  \mathbf{U}_0^c \mathbf{U}_0^{c \top})\mathbf{e}^{(j)} \sim \sigma_{0j}^2 \chi_{n-k_0}^2$, while
\begin{equation*}
    \begin{aligned}
        \max_{j=1, \dots, p} |\mathbf{e}^{(j)\top}(\mathbf{I}_n -  \mathbf{U}_0^c \mathbf{U}_0^{c \top}) &\{-\mathbf{P}_0^\perp (\mathbf{M}_0\mathbf{\lambda}_{0j} + \mathbf{e}^{(j)}) + \mathbf{\Delta} y^{(j)}\} | \\ & \lesssim ||\mathbf{e}^{(j)\top} [\mathbf{P}_0^\perp (\mathbf{M}_0\mathbf{\lambda}_{0j} + \mathbf{e}^{(j)}) + \mathbf{\Delta} y^{(j)}]||\\
        & \quad  +|| \mathbf{e}^{(j)} \mathbf{U}_0^c || ||\mathbf{U}_0^{c \top}||  ||-\mathbf{P}_0^\perp (\mathbf{M}_0\mathbf{\lambda}_{0j} + \mathbf{e}^{(j)}) + \mathbf{\Delta} y^{(j)} ||\\
   & \lesssim  ||\mathbf{e}^{(j)\top}\mathbf{P}_0^\perp \mathbf{M}_0 || || \mathbf{\lambda}_{0j} ||  + ||\mathbf{e}^{(j)\top}\mathbf{P}_0^\perp \mathbf{e}^{(j)}|| + ||\mathbf{e}^{(j)\top} \mathbf{\Delta} y^{(j)} || \\
        & \quad + || \mathbf{e}^{(j)} \mathbf{U}_0^c || ||\mathbf{U}_0^{c \top}||  (||\mathbf{P}_0^\perp (\mathbf{M}_0\mathbf{\lambda}_{0j} + \mathbf{e}^{(j)}) ||+ || \mathbf{\Delta} y^{(j)} ||) \\
        & \lesssim \log p + (C_{r, \{n_s\}_{s}^{S}} C_{c, p}   )^4\frac{n_{\max}}{p}, 
    \end{aligned}
\end{equation*}
since, with probability at least $1-o(1)$, we have $$ \max_{j=1, \dots, p} || \mathbf{P}_0^\perp \mathbf{e}^{(j)}|| \lesssim \max_{j=1, \dots, p} ||\mathbf{e}^{(j)\top } \mathbf{P}_0^\perp \mathbf{M}_0|| \asymp \max_{j=1, \dots, p}||\mathbf{e}^{(j)\top } \mathbf{P}_0^\perp \mathbf{e}^{(j)}|| \lesssim  \log p $$
and $$|| \mathbf{P}_0^\perp \mathbf{M}_0 || \asymp 1,$$ by Lemma \ref{lemma:Mt_P_M}, and $\max_{j=1, \dots, p} |\mathbf{e}^{(j)\top }\mathbf{\Delta} y^{(j)}| \leq C_{\sigma}^2 C_{r, \{n_s\}_{s}^{S}}^4 C_{c, p}^2 \big(1 + \frac{n_{\max}}{p}\big) $, $\max_{j=1, \dots, p} ||\mathbf{\Delta} y^{(j)}|| \asymp C_{\sigma} C_{r, \{n_s\}_{s}^{S}} \big(\frac{1}{\sqrt{n_{\min}}} + \frac{\sqrt{n_{\max}}}{p}\big)$,  where $n_{\max} = \max_{s=1, \dots, S} n_s$ and $n_{\min}= \min_{s=1, \dots, S} n_s$, by Lemma \ref{lemma:Delta_y_j}, and $||  \mathbf{I}_n -  \mathbf{U}_0^c \mathbf{U}_0^{c \top}|| = 1$. Also, 
\begin{equation*}
    \begin{aligned}
   \mathbf{\hat y}^{c(j)\top} \big\{(\mathbf{U}_0^c \mathbf{U}_0^{c \top} -  \mathbf{U}^c \mathbf{U}^{c \top}) &+ \big(1 -\frac{n}{n + \tau_{\Lambda}^{-2}} \big)\mathbf{U}^c \mathbf{U}^{c \top}\big\}\mathbf{\hat y}^{c(j)} \\
    & \lesssim ||\mathbf{\hat y}^{c(j)} ||^2 \big(||\mathbf{U}_0^c \mathbf{U}_0^{c \top} -  \mathbf{U}^c \mathbf{U}^{c \top} || + \frac{\tau_{\Lambda}^2}{n + \tau_{\Lambda}^{-2}} ||\mathbf{U}^c \mathbf{U}^{c \top}|| \big)\\
    & \lesssim  (C_{r, \{n_s\}_{s}^{S}} C_{c, p} )^2  \big(1 + \frac{n}{p} \big),
    \end{aligned}
\end{equation*}
since, with probability at least $1-o(1)$, we have  $\max_{j=1, \dots, p} ||\mathbf{\hat y}^{c(j)} || \lesssim \sqrt{n}$ and $||\mathbf{U}_0^c \mathbf{U}_0^{c \top} -  \mathbf{U}^c \mathbf{U}^{c \top} || \lesssim (C_{r, \{n_s\}_{s}^{S}} C_{c, p})^2  \bigg(\frac{1}{n} +\frac{1}{p}\big))$ by Lemma \ref{lemma:y_hat} and Proposition \ref{prop:recovery_M}, and $||\mathbf{U}^c \mathbf{U}^{c \top}||= 1$. 
Finally, note that $\big(\frac{1}{\mathbf{\gamma}_n} - \frac{1}{n}\big) \mathbf{e}^{(j)\top} (\mathbf{I}_n -  \mathbf{U}_0^c \mathbf{U}_0^{c \top})\mathbf{e}^{(j)} \lesssim \frac{1}{n}$, with probability at least $1-o(1)$, since $\frac{1}{\mathbf{\gamma}_n} + \frac{1}{n} \asymp \frac{1}{n^2}$, and 
$\max_{j=1, \dots, p} |\mathbf{e}^{(j)\top} (\mathbf{I}_n -  \mathbf{U}_0^c \mathbf{U}_0^{c \top})\mathbf{e}^{(j)} - \sigma_{0j}^2|\lesssim \sqrt{n}$,  with probability at least $1-o(1)$, by Proposition 5.16 of \citet{vershynin_12}. The fact that $\mathbf{\gamma}_s \asymp n$ proves $\max_{j=1, \dots, p} |R_j| \lesssim \frac{1}{n} + \frac{1}{p}$, with probability at least $1-o(1)$.
Next, we show $\max_{j=1, \dots, p} |\tilde \sigma_j^2 - \sigma_{0j}^2|$ decreases to $0$ as $n$ and $p$ diverge, following the same steps as in the proof for Theorem 3.6 of \citet{fable}. 
Let $U_j = \frac{\mathbf{\gamma}_n \delta_j^2}{2} \tilde \sigma_j^{-2} - \frac{\mathbf{\gamma}_n}{2}$ and write $\tilde \sigma_j^2 - \sigma_{0j}^2$ as 
$$
\tilde \sigma_j^2 - \sigma_{0j} = \big(1 + \frac{2}{\mathbf{\gamma}_n} U_j\big)^{-1}\big(\delta_j^2 - \sigma_{0j}^2 - \frac{2}{\mathbf{\gamma}_n} U_j \sigma_{0j}^2\big).
$$
By Lemma E.7 in \citet{fable}, we have $\max_{j=1,\dots, p} |U_j|/\mathbf{\gamma}_n \lesssim (\log{p}/n)^{1/3}$ with probability at least $1-o(1)$. Thus, $\min_{j=1,\dots, p}\big|1 + \frac{2}{\mathbf{\gamma}_n} U_j\big| \gtrsim 1/2$ with probability at least $1-o(1)$. Hence,
\begin{equation*}
    \max_{j=1, \dots, p} |\tilde \sigma_j^2 - \sigma_{0j}^2| \lesssim  \max_{j=1, \dots, p} | \delta_j^2 - \sigma_{0j}^2| + \sigma_{0j}^2 \max_{j=1, \dots, p} \frac{|U_j|}{\mathbf{\gamma}_n}.
\end{equation*}
From the representation in \eqref{eq:delta_j_representation}, we have
\begin{equation*}
    \delta_j^2 - \sigma_{0j}^2 = \frac{\sigma_{0j}^2}{n} \bigg\{\frac{C_j}{\sigma_{0j}^2} - (n-k)\bigg\} - \frac{k}{n}\sigma_{0j}^2 + R_j,
\end{equation*}
    and Lemma E.7 of \citet{fable} gives
    \begin{equation*}
         \max_{j=1, \dots, p} \bigg| \frac{C_j}{\sigma_{0j}^2} - (n-k)  \bigg|
 \lesssim n \big(\frac{\log{p}}{n}\big)^{1/3},   \end{equation*}
 with probability at least $1-o(1)$. Combining all of the above, we get
\begin{equation*}
       \max_{j=1, \dots, p} |   \tilde \sigma_j^2 - \sigma_{0j}| \lesssim \big(\frac{\log{p}}{n}\big)^{1/3} + \frac{1}{p},
\end{equation*}
with probability at least $1-o(1)$.
\end{proof}

\section{Extensions to the heteroscedastic case}\label{sec:extensions}
We present an extension of the methodology presented in the paper to the heteroscedastic design. 
In such a case, the conditional conjugacy on the loading matrices seen in the homoscedastic is lost. Hence, instead of simply sampling sequentially ${\mathbf{\Lambda}}$ and then each ${\mathbf{\Gamma}}_s$ given $\mathbf{\Lambda}$, it becomes necessary to alternate between the shared and specific loadings within a Gibbs sampler. 
For the study-specific loadings and variances we specify conjugate Normal-Inverse Gamma priors
\begin{equation*}
\begin{aligned}
      \mathbf{\gamma}_{sj} &\mid \sigma_{sj}^2 \sim N_{q_s} \left(0, \tau_{\Gamma_s} \sigma_{sj}^2\right), \quad \sigma_{sj}^2 \sim IG \left( \frac{\nu_0}{2}, \frac{\nu_0 \sigma_0^2}{2}  \right)\\
\end{aligned}\quad j=1, \dots, p, \quad s=1, \dots, S
    \label{eq:prior_study_specific}
\end{equation*}
and 
assign semi-conjugate Normal inverse gamma priors
       \begin{equation*}
      \mathbf{\lambda}_{j} \mid \tau_{\mathbf{\lambda}_j}^2 \sim N_{k_0} \left(0, \tau_{\mathbf{\lambda}_j}^2\right), \quad \tau_{\mathbf{\lambda}_j}^2 \sim IG(\nu/2, \nu/2 \tau_0^2).\end{equation*}
After estimating the latent factors as in Section \ref{subsec:factor_pretraining}, the posterior computation for the loadings can be performed using the following Gibbs sampler.
\begin{enumerate}
    \item Given the previous sample for ${\mathbf{\Lambda}}$,  for $s=1, \dots, S$, define $\mathbf{\tilde Y_s} = \mathbf{Y}_s - \mathbf{M_s} {\mathbf{\Lambda}}^\top$ and sample $\{(\mathbf{\gamma}_{sj}, \sigma_{sj}^2)\}_{j=1}^p$ from their conditional $NIG$ posterior $(\mathbf{\gamma}_{sj}, \sigma_{sj}^2)$ as 
        \begin{equation*}
         (\mathbf{\gamma}_{sj}, \sigma_{sj}^2)\mid \mathbf{\tilde Y_s}\sim NIG \left(\mathbf{\mu}_{\gamma_{sj}}, \mathbf{K}_j, \nu_{n_s}, \nu_{n_s}\Delta_{sj}^2\right), \quad j=1, \dots, p,
        \end{equation*}
        where
        \begin{equation*}
            \begin{aligned}
                \mathbf{\mu}_{\gamma_{sj}} &= \left(\mathbf{\hat  F_s}^\top \mathbf{\hat  F_s} + \frac{\mathbf{I}_{q_s}}{\tau_{\mathbf{\gamma}_{sj}}^2}\right)^{-1}\mathbf{\hat  F_s}^\top\mathbf{\tilde y}_s^{(j)},\\
                \mathbf{K}_j &= \left(\mathbf{\hat  F_s}^\top \mathbf{\hat  F_s} + \frac{\mathbf{I}_{q_s}}{\tau_{\mathbf{\gamma}_{sj}}^2}\right)^{-1},\\
                \nu_{n_s} &= \nu_0 + n_s,\\
                \nu_{n_s}\Delta_{sj}^2 &= \nu_0\Delta_{0}^2 + \left(||\mathbf{\tilde y}_s^{(j)}||^2 - \mathbf{\mu}_{\gamma_{sj}}^{\top} K_j^{-1}\mathbf{\mu}_{\gamma_{sj}} \right),
            \end{aligned}
        \end{equation*}
        and $\mathbf{\tilde y}_s^{(j)}$ is the $j$-th column of $\mathbf{\tilde Y_s}$
            \item Given the samples for $\{{\mathbf{\Gamma}}_s\}_{s=1}^S$, sample $\{(\mathbf{\lambda}_j, \tau_{\mathbf{\lambda}_j}^2)\}_{j=1}^p$ from its full conditional. In particular, for $j=1, \dots, p$, define $ \mathbf{\bar Y} =\big[
                \mathbf{\bar Y}_s^\top ~  \cdots ~ \mathbf{\bar Y}_S^\top\big]^\top$, where $ \mathbf{\bar Y}_s = \mathbf{Y}_s - \mathbf{\hat  F_s} {\mathbf{\Gamma}}_s$, then
            \begin{enumerate}
                \item Sample $\mathbf{\lambda}_j$ from its full conditional
                \begin{equation*}
                    \mathbf{\lambda}_j \mid \bar y^{(j)}, \tau_{\mathbf{\lambda}_j}^2 \sim N_{k}(\mathbf{\mu}_{\lambda_j}, V_j),
                \end{equation*}
                where
                \begin{equation*}
                    \begin{aligned}
                        \mathbf{\mu}_{\lambda_j} &= V_j^{-1}\mathbf{\hat M}^\top {D_j}^{-1} \mathbf{\bar y}^{(j)},\\
                        V_j &= \left(\mathbf{\hat M}^\top {D_j}^{-1}\mathbf{\hat M} + \frac{\mathbf{I}_{k_0}}{\tau_{\mathbf{\lambda}_j}^2} \right)^{-1},\\
                       {D_j}& = \begin{bmatrix}
                           \sigma_{j1}^2 \mathbf{I}_{n_1} & 0 & \cdots & 0\\
                            0 &   \sigma_{j2}^2 \mathbf{I}_{n_2} & & 0 \\
                           \vdots\\
                             0  & \cdots &  0  & \sigma_{sj}^2 \mathbf{I}_{n_s}
                       \end{bmatrix},
                    \end{aligned}
                \end{equation*}
                and $\mathbf{\bar y}^{(j)}$ is the $j$-th column of $\mathbf{\bar Y}$.
                \item Sample $\tau_{\mathbf{\lambda}_j}^2$ from its full conditional
                \begin{equation*}
                \tau_{\mathbf{\lambda}_j}^2 \mid \mathbf{\bar y}^{(j)}, {\mathbf{\Lambda}} \sim  IG\left(\frac{\nu_0 + n}{2}, \frac{\nu_0 \tau_0^2  + ||\mathbf{\bar y}^{(j)} -\mathbf{\hat M} \mathbf{\lambda}_j||^2}{2}\right).
                \end{equation*}
            \end{enumerate}

\end{enumerate}
We expect this method to still suffer from mild undercoverage. A strategy similar to the one devised in Section \ref{subsec:rho} would be necessary to ensure the frequentist validity of coverage of credible intervals.


\section{Details about the estimation of the latent dimensions}
For the study $s$ and each $k$, we let $\mathbf{\hat M}_s^{(k)}$ be the matrix of left singular vectors of $\mathbf{Y}_s$ scaled by $\sqrt{n_s}$ and let $\mathbf{\hat \Lambda}_s^{(k)} = \frac{1}{n  + 1/\tau_s^2} \mathbf{Y}_s^\top \mathbf{\hat M}_s^{(k)}$ be the conditional posterior mean for the loadings, where $\tau_s$ is chosen as in Section \ref{subsec:hyperparam}. We estimate the residual error variances as $\sigma_{sj}^{(k)2} =  \frac{1}{n_s} ||(\mathbf{I}_{n_s} - \mathbf{\hat M}_s^{(k)}\mathbf{\hat M}_s^{(k)\top}/n_s)\mathbf{y}_s^{(j)} ||_F^2.$ Finally, we approximate the likelihood computed at the joint likelihood estimate  as 
\begin{equation*}
    \begin{aligned}
        l_{sk} \approx \hat l_{sk} &= l(\mathbf{\hat M}_s^{(k)}, \mathbf{\hat \Lambda}_s^{(k)},  \mathbf{\hat \Sigma}_s^{(k)}; \mathbf{Y}_s) \\
        &= - \frac{n_s}{2} \sum_{j=1}^p \log |\mathbf \Sigma_s^{(k)}| - \frac{1}{2} tr\big\{(\mathbf{Y}_s - \mathbf{\hat M}_s^{(k)} \mathbf{\hat \Lambda}_s^{(k) \top})\mathbf \Sigma_s^{(k)-1}  (\mathbf{Y}_s - \mathbf{\hat M}_s^{(k)} \mathbf{\hat \Lambda}_s^{(k)\top}) \big\},
    \end{aligned}
\end{equation*}
where $\mathbf \Sigma_s^{(k)} = \text{diag}(\sigma_{s1}^{(k)2}, \dots, \sigma_{sp}^{(k)2})$.

\section{Additional details about the numerical experiments} \label{sec:additional_details_experiments}
We report additional details about the numerical experiments presented in Section \ref{sec:numerical_experiments} of the main article. For \texttt{SVI}, we use a batch size of $20\%$ of the sample size. All other hyperparameters are set to default values. For all methods, we take $\mathbf{\hat \Lambda} \mathbf{\hat \Lambda}^\top$ and $\mathbf{\hat \Gamma}_s \mathbf{\hat \Gamma}_s^\top$ as estimators for $\mathbf{\Lambda} \mathbf{\Lambda}^\top$ and $\mathbf{\Gamma}_s \mathbf{\Gamma}_s^\top$ where $\mathbf{\hat \Lambda}$ and $\mathbf{\hat \Gamma}_s$ are the posterior mean for $\mathbf{\Lambda}$ and $\mathbf{\Gamma}_s$, respectively. 
We also considered alternative strategies, such as using posterior means for $\mathbf{\Lambda} \mathbf{\Lambda}^\top$ and $\mathbf{\Gamma}_s \mathbf{\Gamma}_s^\top$, but we noticed negligible differences. To estimate frequentist coverage of credible intervals, we obtained 500 posterior samples from the posterior distribution and computed the entry-wise equal-tail credible intervals. To implement the two variational inference algorithms from \citet{vi_bmsfa}, we use code available at \href{https://github.com/blhansen/VI-MSFA}{https://github.com/blhansen/VI-MSFA}. Code for implementing BLAST and replicating
results is available at \href{https://github.com/maurilorenzo/BLAST}{https://github.com/maurilorenzo/BLAST}. 
To compare computational costs, we performed the experiments on a Laptop with 11th Gen Intel(R) Core(TM) i7-1165G7 2.80 GHz and 16GB RAM.

\section{Additional experiments}\label{sec:additional_experiments}
{First, we report in Tables \ref{tab:accuracy_het} and \ref{tab:uq_het} the results on the heteroscedastic case of the first experiment presented in Section \ref{sec:numerical_experiments}. 
There are no significant differences in conclusions compared with the homoscedastic case. \texttt{BLAST} and \texttt{SPECTRAL} are the best performing methods, with  \texttt{BLAST} (\texttt{SPECTRAL}) having better accuracy when $p=500$ ($p=5000$ respectively). Table \ref{tab:uq_het} provides additional support for \texttt{BLAST} in terms of providing well-calibrated credible intervals. } 
\begin{table}
\caption{Comparison of the methods in terms of estimation accuracy with independently generated loading matrices in the heteroscedastic case. Estimation errors have been multiplied by $10^2$. We report the mean and standard error over 50 replications.   \label{tab:accuracy_het}}
\centering
{	\begin{tabular}{crrrr}
& \multicolumn{4}{c}{$p=500$, $n_s=500$}\\
Method & ${\mathbf{\Lambda}} {\mathbf{\Lambda}}^\top$ & $\mathbf{\Gamma}_s\mathbf{\Gamma}_s^\top$ & $\mathbf{M}_s$ & $\mathbf{F}_s$   \\
\hline
  \texttt{CAVI} & $52.78^{0.49}$ & $49.34^{0.79}$& $56.05^{0.57}$ & $32.88^{1.08}$ \\ 
  \texttt{SVI} & $78.29^{0.15}$ & $65.57^{0.74}$& $81.76^{20.10}$& $64.63^{1.66}$   \\ 
	 	 \texttt{SPECTRAL} & $17.34^{0.08}$ & $42.58^{0.30}$& $24.56^{0.13}$& $22.53^{0.10}$ \\ %
\texttt{BLAST} & $15.47^{0.06}$ & $37.25^{0.22}$& $23.46^{0.09}$& $22.25^{0.10}$\\ 
   & \multicolumn{4}{c}{$p=500$, $n_s=1000$}\\
Method & ${\mathbf{\Lambda}} {\mathbf{\Lambda}}^\top$ & $\mathbf{\Gamma}_s\mathbf{\Gamma}_s^\top$ & $\mathbf{M}_s$ & $\mathbf{F}_s$   \\ 
\hline
	\texttt{CAVI} & $57.12^{0.26}$ & $43.71^{0.65}$& $63.24^{0.33}$& $35.69^{0.78}$ \\ 
  		\texttt{SVI} & $72.22^{0.18}$ & $54.76^{0.73}$& $76.10^{0.12}$& $50.36^{1.10}$  \\ 
	       \texttt{SPECTRAL} & $12.84^{0.06}$ & $34.86^{0.29}$& $24.34^{0.12}$& $21.87^{0.09}$ \\ %
\texttt{BLAST}& $11.37^{0.04}$ & $26.89^{0.17}$& $22.22^{0.07}$& $21.88^{0.09}$\\ 
     & \multicolumn{4}{c}{$p=5000$, $n_s=500$}\\
Method & ${\mathbf{\Lambda}} {\mathbf{\Lambda}}^\top$ & $\mathbf{\Gamma}_s\mathbf{\Gamma}_s^\top$ & $\mathbf{M}_s$ & $\mathbf{F}_s$   \\
\hline
	\texttt{CAVI} & $77.68^{0.09}$ & $52.72^{0.38}$& $176.18^{2.14}$& $35.75^{0.46}$ \\ 
  		\texttt{SVI} &  $83.13^{0.05}$ & $58.60^{0.33}$& $118.79^{0.68}$& $43.09^{0.56}$\\ 
	       \texttt{SPECTRAL} & $15.08^{0.05}$ & $33.46^{0.13}$& $8.74^{0.04}$& $8.67^{0.08}$ \\ %
\texttt{BLAST}&$14.75^{0.04}$ & $35.87^{0.17}$& $12.36^{0.11}$& $8.67^{0.08}$\\ 
     & \multicolumn{4}{c}{$p=5000$, $n_s=1000$}\\
Method & ${\mathbf{\Lambda}} {\mathbf{\Lambda}}^\top$ & $\mathbf{\Gamma}_s\mathbf{\Gamma}_s^\top$ & $\mathbf{M}_s$ & $\mathbf{F}_s$   \\
\hline
	\texttt{CAVI} & $79.00^{0.09}$ & $39.70^{0.29}$& $166.60^{1.26}$& $25.49^{0.31}$ \\ 
   \texttt{SVI} & $80.41^{0.06}$ & $44.58^{0.32}$& $170.52^{0.61}$& $30.07^{0.39}$ \\ 
	       \texttt{SPECTRAL} & $10.66^{0.03}$ & $24.16^{0.08}$& $8.23^{0.03}$& $7.75^{0.05}$ \\ %
\texttt{BLAST} & $10.47^{0.03}$ & $25.90^{0.12}$& $9.94^{0.06}$& $7.75^{0.05}$\\ 
	\end{tabular}}
\end{table}

\begin{table}
\caption{Comparison of the methods in terms of frequentist coverage of $95\%$ credible intervals with independently generated loading matrices in the heteroscedastic case. Coverage level have been multiplied by $10^2$. We report the mean and standard error over 50 replications.   \label{tab:uq_het}}
\centering
{	\begin{tabular}{crrrrr}
& \multicolumn{5}{c}{$p=500$}\\
& \multicolumn{2}{c}{$n_s=500$} & & \multicolumn{2}{c}{$n_s=1000$}\\
Method & ${\mathbf{\Lambda}} {\mathbf{\Lambda}}^\top$ & $\mathbf{\Gamma}_s\mathbf{\Gamma}_s^\top$ & & ${\mathbf{\Lambda}} {\mathbf{\Lambda}}^\top$ & $\mathbf{\Gamma}_s\mathbf{\Gamma}_s^\top$   \\
\hline
	\texttt{CAVI} & $29.68^{0.52}$ & $75.25^{0.61}$& &$20.80^{0.34}$ & $64.70^{0.70}$ \\ 
  		\texttt{SVI} &$17.69^{0.29}$ & $62.22^{0.68}$& &$15.31^{0.23}$ & $56.02^{0.65}$  \\ 
	 \texttt{BLAST}& $93.08^{0.11}$ & $95.00^{0.10}$& &$92.40^{0.13}$ & $94.98^{0.08}$ \\ 

& \multicolumn{5}{c}{$p=5000$}\\
& \multicolumn{2}{c}{$n_s=500$} & & \multicolumn{2}{c}{$n_s=1000$}\\
Method & ${\mathbf{\Lambda}} {\mathbf{\Lambda}}^\top$ & $\mathbf{\Gamma}_s\mathbf{\Gamma}_s^\top$ & & ${\mathbf{\Lambda}} {\mathbf{\Lambda}}^\top$ & $\mathbf{\Gamma}_s\mathbf{\Gamma}_s^\top$   \\
\hline
	\texttt{CAVI} & $19.05^{0.25}$ & $72.61^{0.51}$& &$13.73^{0.23}$ & $70.94^{0.49}$  \\ 
  		\texttt{SVI} & $17.18^{0.24}$ & $70.84^{0.47}$& &$13.80^{0.23}$ & $67.45^{0.51}$   \\ 
	 \texttt{BLAST}& $94.32^{0.12}$ & $94.38^{0.10}$& &$94.59^{0.10}$ & $94.29^{0.07}$ \\ 
	\end{tabular}}
\end{table}

We also performed an additional experiment in the homescedastic case in a lower dimensional scenario.  In particular, we adopted the same setting as for the first experiment in Section \ref{sec:numerical_experiments} with $k_0 = 5$ and $q_s = 4$ and let $S=3$, $n_s=300$, and $p=200$. We also compare to the maximum likelihood estimate of the multi-study factor model obtained via the expectation maximization algorithm \citep{msfa} (\texttt{EM}, henceforth) and a Bayesian estimate performing Bayesian computation via a Gibbs sampler \citep{bmsfa} (\texttt{GIBBS}, henceforth), using the code available at \href{https://github.com/rdevito/MSFA}{https://github.com/rdevito/MSFA}.
For \texttt{EM}, we estimate latent factors via the Thomson's factor score, which corresponds to their conditional mean given the estimates for factor loadings and residuals variances, and we take $\mathbf{\hat \Lambda} \mathbf{\hat \Lambda}^\top$ and $\mathbf{\hat \Gamma}_s \mathbf{\hat \Gamma}_s^\top$ as estimates for $\mathbf{\Lambda} \mathbf{\Lambda}^\top$ and $\mathbf{\Gamma}_s \mathbf{\Gamma}_s^\top$ where $\mathbf{\hat \Lambda}$ and $\mathbf{\hat \Gamma}_s$ are the maximum likelihood estimates for $\mathbf{\Lambda}$ and $\mathbf{\Gamma}_s$ respectively. 
For \texttt{GIBBS}, we obtain point estimates for latent factors by averaging the posterior samples after alignment using the method of \citet{matchalign}, and estimate  $\mathbf{\Lambda} \mathbf{\Lambda}^\top$ and $\mathbf{\Gamma}_s \mathbf{\Gamma}_s^\top$ via their posterior mean. We run the Gibbs sampler for 20000 iterations, discarding the first 10000 as burn in, and retain one sample every 20, resulting in a total of 500 thinned posterior samples. 
For \texttt{BLAST}, we set $\tau = 0.2$.

Table \ref{tab:accuracy_additional} reports a comparison in terms of estimation accuracy and frequentist coverage of $95\%$ credible intervals. The maximum likelihood estimate (\texttt{EM}) and \texttt{SVI} achieve the best estimation accuracy. \texttt{BLAST} estimation accuracy is slightly worse but provides more precise uncertainty quantification than all competitors considered. 
Moreover, \texttt{BLAST} and \texttt{SPECTRAL} are much faster than all the alternatives (about $10$ times faster than \texttt{CAVI} and $150$ times faster than \texttt{SVI}) (see Table \ref{tab:time_additional}). Hence, even if \texttt{BLAST} was mainly motivated by high-dimensional applications, it still performs competitively in lower dimensional examples.

\begin{table}
\caption{Comparison of the methods in terms of estimation accuracy and frequentist coverage of $95\%$ credible intervals in the additional experiment. Values have been multiplied by $10^2$. We report the mean and standard error over 50 replications.   \label{tab:accuracy_additional} }
\centering
{	\begin{tabular}{crrrrrr}
\multirow{2}{*}{Method} &\multicolumn{4}{c}{Estimation Accuracy}&\multicolumn{2}{c}{Coverage}  \\
& ${\mathbf{\Lambda}} {\mathbf{\Lambda}}^\top$ & $\mathbf{\Gamma}_s\mathbf{\Gamma}_s^\top$ & $\mathbf{M}_s$ & $\mathbf{F}_s$  & ${\mathbf{\Lambda}} {\mathbf{\Lambda}}^\top$ & $\mathbf{\Gamma}_s\mathbf{\Gamma}_s^\top$  \\
\hline
  \texttt{EM} & $24.12^{0.19}$ & $46.99^{0.49}$& $30.64^{0.23}$ &$29.68^{0.24}$& NA & NA \\ 
  \texttt{GIBBS}  & $76.45^{0.14}$ & $77.76^{0.17}$ & $90.38^{0.61}$ & $86.31^{0.61}$ & $39.98^{0.54}$ & $53.99^{0.57}$  \\ 
  \texttt{CAVI} & $27.82^{0.27}$ & $48.22^{0.28}$ & $32.61^{0.24}$ & $31.54^{0.24}$ & $85.43^{0.32}$ & $84.03^{0.037}$ \\ 
  \texttt{SVI}  & $26.31^{0.31}$ & $46.22^{0.40}$ & $30.97^{0.23}$ & $30.89^{0.25}$ & $88.66^{0.31}$ & $87.51^{0.29}$ \\ 
	 \texttt{BLAST}  & $27.87^{0.25}$ & $47.14^{0.44}$ & $35.87^{0.20}$ & $35.92^{0.27}$ & $90.60^{0.16}$ & $95.49^{0.11}$  \\ 
     \texttt{SPECTRAL}  & $35.01^{0.45}$ & $62.15^{0.80}$ & $39.70^{0.44}$ & $35.87^{0.27}$ & NA & NA  \\ %
	\end{tabular}}
\end{table}

\begin{table}
\caption{Comparison of the methods in terms of running time in the additional experiment. We report the mean and standard error over 20 replications.   \label{tab:time_additional}}
\centering
{	\begin{tabular}{cr}
Method & Time (s)\\
\hline
  \texttt{EM} & $>10000$\\
  \texttt{GIBBS} & $4451.56^{190.23}$  \\ 
  \texttt{CAVI} & $43.36^{2.57}$ \\ 
  \texttt{SVI}  & $458.73^{24.35}$ \\
	 \texttt{BLAST} & $3.61^{0.49}$  \\
     \texttt{SPECTRAL} & $4.24^{0.58}$  \\
	\end{tabular}}
\end{table}

\section{Additional details and figures for the application to gene association among immune cells data}\label{sec:additional_details_application}
As pre-processing, we normalize the data by applying the inverse CDF transform estimating the marginal CDFs via their empirical counterpart and we standardize the variables to have zero mean and unit standard deviation. To apply the criterion in \eqref{eq:k_hat}, we set $\tau = 0.2$. 

We report some results from fitting our \texttt{BLAST} approach to data on $p=2846$ genes.  
Figure \ref{fig:emp_vs_est_corr} shows the reconstructed within-study correlation matrices (left and middle panels) and rescaled shared variances for 1000 genes. All elements whose 95\% credible intervals for the correlation included zero were set to zero in the plot. There are clear similarities in the correlation structure across the three studies. However, non-negligible differences in the strength of the signal are present, with study GSE109125, which is the study where all the methods had the lowest out-of-sample performance, having the weakest signal. A clear shared low-rank structure arises, which is present but less visible in the  reconstructed covariance for each study.

\begin{figure}
\begin{center}

\vfill
\begin{subfigure}[b]{\textwidth}
    \includegraphics[width=6.3in]{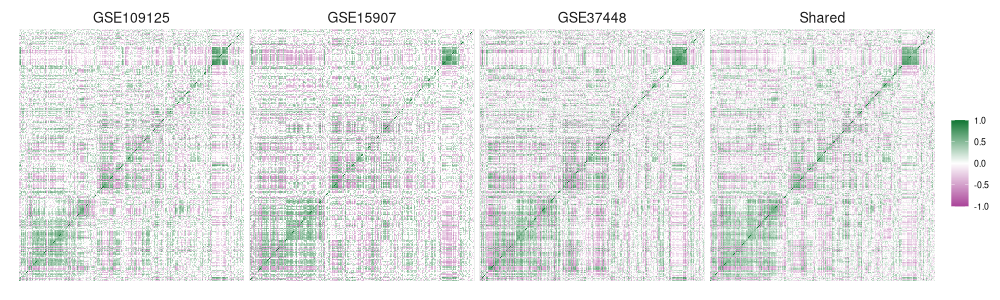}
\end{subfigure}
\end{center}
\vspace{-1em}
\caption{
Reconstructed 
within-study
correlation matrices (left and middle panels) and rescaled shared component (right panel) for 1000 genes. 
Elements for which the $95\%$ credible intervals included $0$ were set to $0$. \label{fig:emp_vs_est_corr}}
\end{figure}

Next, we study the correlation induced by the estimated shared covariance matrix $\mathbf\Lambda \mathbf\Lambda + \mathbf\Sigma$. Among the 1000 genes in Figure \ref{fig:emp_vs_est_corr}, we select $392$ ``hub'' genes that have absolute correlations greater than $0.5$ with at least 20 other genes. Focusing on these hub genes, we show the dependence network in Figure \ref{fig:network}, where the size of each gene (i.e. node) is proportional to the number of connections. We identify four main clusters of genes using the Louvain method \citep{cluster_algo}.
The first cluster (shown in green) is characterized by a strong representation of cell cycle and DNA replication genes, such as \textit{Dhfr}, \textit{Chaf1b}, \textit{Aspm}, and \textit{Aurkb}. 
In the second cluster (blue), we observe genes immune activation and inflammatory responses (e.g. \textit{Il18rap}, \textit{Tnfaip2} \citep{TNFAIP2}). 
The third (violet) contains genes related to TCR activation signaling and Lymphocyte proliferation, such as \textit{Zap70},  \textit{Cd247}, \textit{Cd28} and  \textit{Top2a}.
Finally, the fourth cluster (orange) has a distinct profile with an emphasis on metabolic enzymes alongside immune regulators and effector molecules (e.g. \textit{Mthfd1l} and \textit{Ahcy} \citep{Ducker2017OneCarbon}).

\begin{figure}
    \begin{center}
    \includegraphics[width=0.85\linewidth]{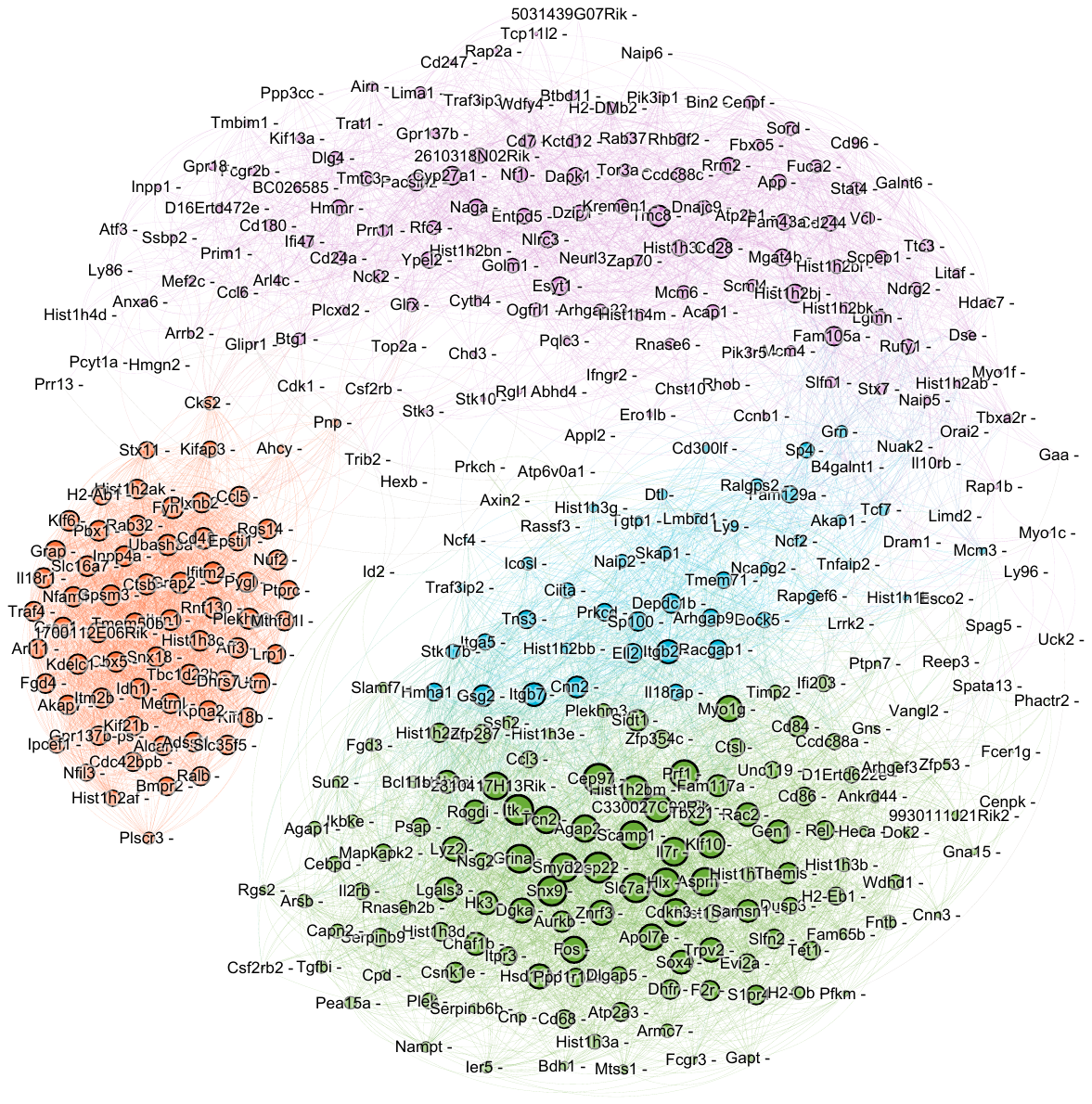}
    \end{center}
      \caption{ Common gene co-expression network obtained using \texttt{GEPHI} \citep{gephi} among 392 genes. Nodes (edges) represent genes (positive dependencies). Node size is proportional to the node degree. Nodes are divided into four main clusters based on their connections.}
    \label{fig:network}
\end{figure}

\end{document}